\documentclass[preprint,3p,times]{elsarticle}
\RequirePackage{multirow,booktabs,subfigure,color,array,hhline,makecell}
\usepackage{amssymb}
\usepackage{amsmath}
\usepackage{graphicx}
\usepackage{subfigure}
\usepackage{amsthm}
\usepackage{mathrsfs}
\usepackage{indentfirst}
\usepackage[colorlinks,citecolor=blue,urlcolor=blue]{hyperref}
\allowdisplaybreaks
\usepackage[table]{xcolor}
\usepackage{tikz-network}
\usepackage{pgf}
\usepackage{tikz}
\usepackage{subfigure}
\usetikzlibrary{arrows, decorations.pathmorphing, backgrounds, positioning, fit, petri, automata}
\usetikzlibrary{shadows,arrows,positioning}
\RequirePackage{algorithm,algpseudocode}
\allowdisplaybreaks
\usepackage{changes}

\newtheorem{thm}{Theorem}
\newtheorem{defin}{Definition}
\newtheorem{lem}{Lemma}
\newtheorem{assum}{Assumption}
\newtheorem{rem}{Remark}
\newtheorem{cor}{Corollary}

\newtheorem{prop}{Proposition}
\newtheorem{Ex}{Example}
\allowdisplaybreaks[4]
\AtBeginDocument{%
	\providecommand\BibTeX{{%
			\normalfont B\kern-0.5em{\scshape i\kern-0.25em b}\kern-0.8em\TeX}}}

\journal{}
\begin{document}
\begin{frontmatter}
\title{Latent class analysis with weighted responses}

\author[label1]{Huan Qing\corref{cor1}}
\ead{qinghuan@u.nus.edu\&qinghuan07131995@163.com}
\cortext[cor1]{Corresponding author.}
\address[label1]{School of Economics and Finance, Chongqing University of Technology, Chongqing, 400054, China}
\begin{abstract}

The latent class model has been proposed as a powerful tool for cluster analysis of categorical data in various fields such as social, psychological, behavioral, and biological sciences. However, one important limitation of the latent class model is that it is only suitable for data with binary responses, making it fail to model real-world data with continuous or negative responses. In many applications, ignoring the weights throws out a lot of potentially valuable information contained in the weights. To address this limitation, we propose a novel generative model, the weighted latent class model (WLCM). Our model allows data's response matrix to be generated from an arbitrary distribution with a latent class structure. In comparison to the latent class model, our WLCM is more realistic and more general. To our knowledge, our WLCM is the first model for latent class analysis with weighted responses. We investigate the identifiability of the model and propose an efficient algorithm for estimating the latent classes and other model parameters. We show that the proposed algorithm enjoys consistent estimation. The performance of the proposed algorithm is investigated using both computer-generated and real-world weighted response data.
\end{abstract}

\end{frontmatter}
\section{Introduction}\label{sec1}
Latent class model (LCM) \citep{dayton1988concomitant,hagenaars2002applied,magidson2004latent} is a powerful tool for categorical data, with many applications across various areas such as social, psychological, behavioral, and biological sciences. These applications include movie rating \citep{guo2014etaf,harper2015movielens},  psychiatric evaluation \citep{meyer2001psychological,silverman2015american,de2018analysis,chen2019joint}, educational assessments \cite{shang2021partial}, political surveys \citep{poole2000nonparametric,clinton2004statistical,bakker2013bayesian,chen2021unfolding}, transport economics personal interview \citep{martinez2022shippers}, and disease etiology detection \citep{formann1996latent,kongsted2017latent,wu2017nested}. In categorical data, subjects (individuals) typically respond to several items (questions). LCM is a theoretical model that categorizes subjects into disjoint groups, known as latent classes, according to their response pattern to a collection of categorical items. For example, in movie rating, latent classes may represent different groups of users with an affinity for certain movie themes; in psychological tests, latent classes may represent different types of personalities. In educational assessments, latent classes may indicate different levels of abilities. In political surveys, latent classes may represent distinct types of political ideologies. In transport economics personal interview, each latent class stands for a partition of the population. In disease etiology detection, latent classes may represent different disease categories. To infer latent classes for categorical data generated from LCM, various approaches have been developed in recent years, including maximum likelihood estimation techniques \citep{van1996estimating,bakk2016robustness,chen2022beyond,gu2023joint} and tensor-based methods \citep{anandkumar2014tensor,zeng2023tensor}.

To mathematically describe categorical data, let $R$ be the $N$-by-$J$ observed response matrix such that $R(i,j)$ represents subject $i$'s response to item $j$, where $N$ denotes the number of subjects and $J$ denotes the number of items. For LCM, researchers commonly focus on binary choice data where elements of the observed response matrix $R$ only take 0 or 1 \citep{formann1996latent,shang2021partial,zeng2023tensor,formann1985constrained,lindsay1991semiparametric,zhang2004hierarchical,yang2006evaluating,xu2017aos,xu2018identifying,ma2019cognitive,gu2020partial}. LCM models binary response matrix by generating its elements from a Bernoulli distribution. In categorical data, binary responses can be agree/disagree responses in psychiatric evaluation, correct/wrong responses in educational assessments, and presence/absence of symptoms in disease etiology detection. However, categorical data is more than binary response. Categorical data with weighted responses is also commonly encountered in the real world and ignoring weighted data may lose potentially meaningful information \citep{newman2004analysis}. For example, in movie rating \citep{guo2014etaf}, rating scores range in $\{1,2,3,4,5\}$ and simply letting $R$ be binary by recording rated/not rated loses valuable information that can reflect users’ preference patterns; for real-world categorical data from various online personality tests in the link \url{https://openpsychometrics.org/_rawdata/}, the range of most responses are $\{0,1,2,\ldots,m\}$, where $m$ is an integer like 2, 5, and 10; in the buyer-seller rating e-commerce data \citep{derr2019balance}, elements of the observed response matrix take values in $\{-1,0,1\}$ (for convenience, we call such $R$ as signed response matrix in this paper) since sellers are rated by users by applying three levels of rating, ``Positive”, ``Neutral”, and ``Negative”. In the users-jokes ratting categorical data Jester 100 \citep{goldberg2001eigentaste}, all responses (i.e., ratings) are continuous numbers ranging in $[-10,10]$. All aforementioned real-world data with weighted responses cannot be generated from a Bernoulli distribution. Therefore, the classical latent class model is inadequate for handling the aforementioned data with weighted responses. As a result, it is desirable to develop a more flexible model for data with weighted responses. With this motivation, our key contributions to the literature of latent class analysis are summarized as follows.
\begin{itemize}
  \item Model. We propose a novel, identifiable, and generative statistical model, the weighted latent class model (WLCM), for categorical data with weighted responses, where the responses can be continuous or negative values. Our WLCM allows the elements of an observed weighted response matrix $R$ to be generated from any distribution provided that the population version of $R$ under WLCM enjoys a latent class structure. For example, our WLCM allows $R$ to be generated from Bernoulli, Normal, Poisson, Binomial, Uniform, and Exponential distributions, etc. By considering a specifically designed discrete distribution, our WLCM can also model signed response matrices. For details, please refer to Examples \ref{Bernoulli}-\ref{Signed}. For comparison, LCM requires $R$ to be generated from Bernoulli distribution and LCM is a sub-model of our WLCM. Under the proposed model, the elements of the observed weighted response matrix $R$ can take any value. Therefore, WLCM is more flexible than LCM. As far as we know, our WLCM is the first statistical model for categorical data in which weighted responses can be continuous or negative values.
  \item Algorithm. We develop an easy-to-implement algorithm, spectral clustering with K-means (SCK), to infer latent classes for weighted response matrices generated from arbitrary distribution under the proposed model. Our algorithm is designed based on a combination of two popular techniques: the singular value decomposition (SVD) and the K-means algorithm.
  \item Theoretical property. We build a theoretical framework to show that SCK enjoys consistent estimation under WLCM. We also provide Examples \ref{Bernoulli}-\ref{Signed} to show that the theoretical performance of the proposed algorithm can be different when the observed weighted response matrices $R$ are generated from different distributions under the proposed model.
  \item Empirical validation. We conduct extensive simulations to validate our theoretical insights. Additionally, we apply our SCK approach to two real-world datasets with meaningful interpretations.
\end{itemize}
The remainder of this paper is organized as follows. Section \ref{sec2} describes the model. Section \ref{sec3} details the algorithm. Section \ref{sec4} establishes the consistency results and provides examples for further analysis. Section \ref{sec5} contains numerical studies that verify our theoretical findings and examine the performance of the proposed method. Section \ref{sec6realdata} demonstrates the proposed method using two real-world datasets. Section \ref{sec7} concludes the paper with a brief discussion of contributions and future work.

The following notations will be used throughout the paper. For any positive integer $m$, let $[m]$ and $I_{m\times m}$ be $[m]:=\{1,2,\ldots,m\}$ and the $m\times m$ identity matrix, respectively. For any vector $x$ and any $q>0$, $\|x\|_q$ denotes $x$’s $l_q$-norm. For any matrix $M$, $M'$ denotes its transpose, $\|M\|$ denotes its spectral norm, $\|M\|_F$ denotes its Frobenius norm, $\mathrm{rank}(M)$ denotes its rank, $\sigma_i(M)$ denotes its $i$-th largest singular value, $\lambda_i(M)$ denotes its $i$-th largest eigenvalue ordered by magnitude, $M(i,:)$ denotes its $i$-th row, and $M(:,j)$ denotes its $j$-th column. Let $\mathbb{R}$ and $\mathbb{N}$ be the set of real numbers and nonnegative integers, respectively. For any random variable $X$, $\mathbb{E}(X)$ and $\mathbb{P}(X=a)$ are the expectation and the probability that $X$ equals to $a$, respectively. Let $\mathbb{M}_{m, K}$ be the collection of all $m\times K$ matrices where each row has only one $1$ and all others $0$.
\section{Weighted latent class model}\label{sec2}
Unlike most researchers that only focus on binary responses, in our weighted response setting in this paper, all elements of the observed weighted response matrix $R$ are allowed to be any real value, i.e., $R\in\mathbb{R}^{N\times J}$.

Consider categorical data with $N$ subjects and $J$ items, where the $N$ subjects belong to $K$ disjoint extreme latent profiles (also known as latent classes). Throughout this paper, the number of classes $K$ is assumed to be a known integer. To describe the membership of each subject, we let $Z$ be a $N\times K$ matrix such that $Z(i,k)$ is 1 if subject $i$ belongs to the $k$-th extreme latent profile and $Z(i,k)$ is 0 otherwise. Call $Z$ the classification matrix in this paper. For each subject $i\in[N]$, it is assumed to belong to a single extreme latent profile. For convenience, define $\ell$ as a $N$-by-$1$ vector whose $i$-th entry $\ell(i)$ is $k$ if the $i$-th subject belongs to the $k$-th extreme latent profile for $i\in[N]$. Thus for subject $i\in[N]$, we have $Z(i,\ell(i))=1$ and the other $(K-1)$ entries of the $K\times 1$ classification vector $Z(i,:)$ is 0.

Introduce the $J\times K$ item parameter matrix $\Theta\in\mathbb{R}^{J\times K}$. For $k\in[K]$, our weighted latent class model (WLCM) assumes that $\Theta(j,k)$ collects the conditional-response expectation for the response of the $i$-th subject to the $j$-th item under arbitrary distribution $\mathcal{F}$ provided that subject $i$ belongs to the $k$-th extreme latent profile. Specifically, for $i\in[N], j\in[J]$, given the classification vector $Z(i,:)$ of subject $i$ and the item parameter matrix $\Theta$, our WLCM assumes that for arbitrary distribution $\mathcal{F}$, the conditional response expectation of the $i$-th subject to the $j$-th item is
\begin{align}\label{Rij}
\mathbb{E}(R(i,j)|Z(i,:),\Theta)=\sum_{k=1}^{K}Z(i,k)\Theta(j,k)\equiv\Theta(j,\ell(i)).
\end{align}
Based on Equation (\ref{Rij}), our WLCM can be simplified as follows.
\begin{defin}\label{WLCM}
Let $R\in\mathbb{R}^{N\times J}$ denote the observed weighted response matrix. Let $Z\in\mathbb{M}_{N,K}$ be the classification matrix and $\Theta\in\mathbb{R}^{J\times K}$ be the item parameter matrix. For $i\in[N], j\in[J]$, our weighted latent class model (WLCM) assumes that for an arbitrary distribution $\mathcal{F}$, $R(i,j)$ are independent random variables generated from the distribution $\mathcal{F}$ and the expectation of $R(i,j)$ under the distribution $\mathcal{F}$ should satisfy the following formula:
\begin{align}\label{RFR0}
\mathbb{E}(R(i,j))=R_{0}(i,j), \mathrm{where~}R_{0}:=Z\Theta'.
\end{align}
\end{defin}
Definition \ref{WLCM} says that WLCM is determined by the classification matrix $Z$, the item parameter matrix $\Theta$, and the distribution $\mathcal{F}$. For brevity, we denote WLCM by $WLCM(Z,\Theta,\mathcal{F})$. Under WLCM, $\mathcal{F}$ is allowed to be any distribution as long as Equation (\ref{RFR0}) is satisfied under $\mathcal{F}$, i.e., WLCM only requires the expectation (i.e., population) response matrix $R_{0}$ of the observed weighted response matrix $R$ to be $Z\Theta'$ under any distribution $\mathcal{F}$.
\begin{rem}
For the case that $\mathcal{F}$ is Bernoulli distribution, all elements of $\Theta$ range in $[0,1]$, $R$ only contains binary responses (i.e., $R(i,j)\in\{0,1\}$ for $i\in[N],j\in[J]$ when  $\mathcal{F}$ is Bernoulli distribution), and Equation (\ref{Rij}) becomes $\mathbb{P}(R(i,j)=1|Z(i,:),\Theta)=\Theta(j,\ell(i))$. For this case, WLCM reduces to the LCM model, i.e., LCM is a special case of our WLCM.
\end{rem}
\begin{rem}
It should be noted that Equation (\ref{RFR0}) does not hold for all distributions. For instance, we cannot set $\mathcal{F}$ as a t-distribution because the expectation of a t-distribution is always 0, which cannot capture the latent structure required by the WLCM model; $\mathcal{F}$ cannot be a Cauchy distribution whose expectation even does not exist; $\mathcal{F}$ cannot be a Chi-square distribution because the expectation of a Chi-square distribution is its degrees of freedom, which is a fixed positive integer and cannot capture the latent structure required by WLCM. We will provide some examples to demonstrate that Equation (\ref{RFR0}) can be satisfied for different distribution $\mathcal{F}$. For details, please refer to Examples \ref{Bernoulli}-\ref{Signed}.
\end{rem}
\begin{rem}
It should be also noted that the ranges of the observed weighted response matrix $R$ and the item parameter matrix $\Theta$ depend on distribution $\mathcal{F}$. For example, when $\mathcal{F}$ is Bernoulli distribution, $R\in\{0,1\}^{N\times J}$ and $\Theta\in[0,1]^{J\times K}$; when $\mathcal{F}$ is Poisson distribution, $R\in\mathbb{N}^{N\times J}$ and $\Theta\in[0,+\infty)^{J\times K}$; If we let $\mathcal{F}$ be Normal distribution, $R\in\mathbb{R}^{N\times J}$ and $\Theta\in(-\infty,+\infty)^{J\times K}$. For details, please refer to Examples \ref{Bernoulli}-\ref{Signed}.
\end{rem}
The following proposition shows that the WLCM model is identifiable as long as there exists at least one subject for every extreme latent profile.
\begin{prop}\label{idWLCM}
(Identifiability). Consider a WLCM model as in Equation (\ref{RFR0}), when each extreme latent profile has at least one subject, the model is identifiable: for any other valid parameter set $(\tilde{Z},\tilde{\Theta})$, if $\tilde{Z}\tilde{\Theta}'=Z\Theta'$, then $(Z,\Theta)$ and $(\tilde{Z},\tilde{\Theta})$ are identical up to a permutation of the $K$ extreme latent profiles.
\end{prop}
All proofs of theoretical results developed in this paper are given in the Appendix. The condition that each extreme latent profile must contain at least one subject means that each extreme latent profile cannot be an empty set and we have $\mathrm{rank}(Z)=K$.
\begin{rem}
Note that $Z$ and $\tilde{Z}$ are the same up to a permutation of the $K$ latent classes in Proposition \ref{idWLCM}. A permutation is acceptable since the equivalence of $Z$ and $\tilde{Z}$ should not rely on how we label each of the $K$ extreme latent profiles. A similar argument holds for the identity of $\Theta$ and $\tilde{\Theta}$.
\end{rem}
The observed weighted response matrix $R$ along with the ground-truth classification matrix $Z$ and the item parameter matrix $\Theta$ can be generated using our WLCM model as follows: let $R(i,j)$ be a random variable generated by distribution $\mathcal{F}$ with expected value $R_{0}(i,j)$ for $i\in[N],j\in[J]$, where $R_{0}=Z\Theta'$ satisfies the latent structure required by WLCM. In latent class analysis, given the observed weighted response matrix $R$ generated from $WLCM(Z,\Theta,\mathcal{F})$, our goal is to infer the classification matrix $Z$ and the item parameter matrix $\Theta$. Proposition \ref{idWLCM} ensures that the model parameters $Z$ and $\Theta$ can be reliably inferred from the observed weighted response matrix $R$. In the following two sections, we will develop a spectral algorithm to fit WLCM and show that this algorithm yields consistent estimation.
\section{A spectral method for parameters estimation}\label{sec3}
We have presented our model, WLCM, and demonstrated its superiority over the classical latent class model. In addition to providing a more general model for latent class analysis, we are also interested in estimating the model parameters. In this section, we focus on the parameter estimation problem within the WLCM framework by developing an efficient and easy-to-implement spectral method.

To provide insight into developing an algorithm for the WLCM model, we first consider an oracle case where we observe the expectation response matrix $R_{0}$ given in Equation (\ref{RFR0}). We would like to estimate $Z$ and $\Theta$ from $R_{0}$. Recall that  the item parameter matrix $\Theta$ is a $J$-by-$K$ matrix , here we let $\mathrm{rank}(\Theta)=K_{0}$, where $K_{0}$ is a positive integer and it is no larger than $K$. As $R_{0}=Z\Theta', \mathrm{rank}(Z)=K$, and $\mathrm{rank}(\Theta)=K_{0}\leq K$, we see that $R_{0}$ is a rank-$K_{0}$ matrix. As the number of extreme latent profiles $K$ is usually far smaller than the number of subjects $N$ and the number of items $J$, the $N$-by-$J$ population response matrix $R_{0}$ enjoys a low-dimensional structure. Next, we will demonstrate that we can greatly benefit from the low-dimensional structure of $R_{0}$ when we aim to develop a method to infer model parameters under the WLCM model.

Let $R_{0}=U\Sigma V'$ be the compact singular value decomposition (SVD) of $R_{0}$ such that $\Sigma$ is a $K_{0}\times K_{0}$ diagonal matrix collecting the $K_{0}$ nonzero singular values of $R_{0}$. Write $\Sigma=\mathrm{diag}(\sigma_{1}(R_{0}), \sigma_{2}(R_{0}),\ldots, \sigma_{K_{0}}(R_{0}))$. The $N\times K_{0}$ matrix $U$ collects the corresponding left singular vectors and it satisfies $U'U=I_{K_{0}\times K_{0}}$. Similarly, the $J\times K_{0}$ matrix $V$ collects the corresponding right singular vectors and it satisfies $V'V=I_{K_{0}\times K_{0}}$. For $k\in[K]$, let $N_{k}$ be the number of subjects that belong to the $k$-th extreme latent profile, i.e., $N_{k}=\sum_{i=1}^{N}Z(i,k)$. The ensuing lemma constitutes the foundation of our estimation method.
\begin{lem}\label{UVWLCM}
Under $WLCM(Z,\Theta,\mathcal{F})$, let $R_{0}=U\Sigma V'$ be the compact SVD of $R_{0}$. The following statements are true.
\begin{itemize}
\item (1) The left singular vectors matrix $U$ can be written as
\begin{align}\label{UX}
U=ZX,
\end{align}
where $X$ is a $K\times K_{0}$ matrix.
\item (2) $U$ has $K$ distinct rows such that for any two distinct subjects $i$ and $\bar{i}$ that belong to the same extreme latent profile (i.e., $\ell(i)=\ell(\bar{i})$), we have $U(i,:)=U(\bar{i},:)$.
\item (3) $\Theta$ can be written as
\begin{align}\label{ThetaZ}
\Theta=V\Sigma U'Z(Z'Z)^{-1}.
\end{align}
\item (4) Furthermore, when $K_{0}=K$, for all $k\in[K], l\in[K]$, and $k\neq l$, we have
\begin{align}\label{Xd}
\|X(k,:)-X(l,:)\|_{F}=(N_{k}^{-1}+N_{l}^{-1})^{1/2}.
\end{align}
\end{itemize}
\end{lem}
From now on, for the simplicity of our further analysis, we let $K_{0}\equiv K$. Hence, the last statement of Lemma \ref{UVWLCM} always holds.

The second statement of Lemma \ref{UVWLCM} indicates that the rows of $U$ corresponding to subjects assigned to the same extreme latent profile are identical. This circumstance implies that the application of a clustering algorithm to the rows of $U$ can yield an exact reconstruction of the classification matrix $Z$ after a permutation of the $K$ extreme latent profiles.

In this paper, we adopt the K-means clustering algorithm, an unsupervised learning technique that groups similar data points into $K$ clusters. This clustering technique is detailed as follows,
\begin{align}\label{Kmeans}
(\bar{\bar{Z}},\bar{\bar{X}})=\mathrm{arg~}\mathrm{min}_{\bar{Z}\in \mathbb{M}_{N,K}, \bar{X}\in \mathbb{R}^{K\times K}}\|\bar{Z}\bar{X}-\bar{U}\|^{2}_{F},
\end{align}
where $\bar{U}$ is any $N\times K$ matrix. For convenience, call Equation (\ref{Kmeans}) as ``Run K-means algorithm on all rows of $\bar{U}$ with $K$ clusters to obtain $\bar{\bar{Z}}$'' because we are interested in the classification matrix $\bar{\bar{Z}}$. Let $\bar{U}$ in Equation (\ref{Kmeans}) be $U$, the second statement of Lemma \ref{UVWLCM} guarantees that $\bar{\bar{Z}}=Z\mathcal{P}, \bar{\bar{X}}=\mathcal{P}'X$, where $\mathcal{P}$ is a $K\times K$ permutation matrix, i.e., running K-means algorithm on all rows of $U$ exactly recovers $Z$ up to a permutation of the $K$ extreme latent profiles.

After obtaining $Z$ from $U$, $\Theta$ can be recovered subsequently by Equation (\ref{ThetaZ}). The above analysis suggests the following algorithm, Ideal SCK, where SCK stands for Spectral Clustering with K-means. Ideal SCK returns a permutation of $(Z,\Theta)$, which also supports the identifiability of the proposed model as stated in Proposition \ref{idWLCM}.
\begin{algorithm}
\caption{\textbf{Ideal SCK}}
\label{alg:IdealSCK}
\begin{algorithmic}[1]
\Require The expectation response matrix $R_{0}$ and the number of extreme latent profiles $K$.
\Ensure A permutation of $Z$ and $\Theta$.
\State Obtain $U\Sigma V'$, the top $K$ SVD of $R_{0}$.
\State Run K-means algorithm on all rows of $U$ with $K$ clusters to obtain $Z\mathcal{P}$, a permutation of $Z$.
\State Equation (\ref{ThetaZ}) gives $V\Sigma U'Z\mathcal{P}((Z\mathcal{P})'Z\mathcal{P})^{-1}=\Theta\mathcal{P}$, a permutation of $\Theta$.
\end{algorithmic}
\end{algorithm}

For the real case, the weighted response matrix $R$ is observed rather than the expectation response matrix $R_{0}$. We now move from the ideal scenario to the real scenario, intending to estimate $Z$ and $\Theta$ when the observed weighted response matrix $R$ is a random matrix generated from an unknown distribution $\mathcal{F}$ satisfying Equation (\ref{RFR0}) with $K$ extreme latent profiles under the WLCM model. The expectation of $R$ is  $R_{0}$ according to Equation (\ref{RFR0}) under WLCM, so intuitively, the singular values and singular vectors of $R$ will be close to those of $R_{0}$. Set $\hat{R}=\hat{U}\hat{\Sigma}\hat{V}'$ as the top $K$ SVD of $R$, where $\hat{\Sigma}$ is a $K\times K$ diagonal matrix collecting the top $K$ singular values of $R$. Write $\hat{\Sigma}=\mathrm{diag}(\sigma_{1}(R),\sigma_{2}(R),\ldots,\sigma_{K}(R))$. As $\mathbb{E}(R)=R_{0}$ and the $N\times J$ matrix $R_{0}$ has $K$ non-zero singular values while the other $(\mathrm{min}(N,J)-K)$ singular values are zeros, we see that $\hat{R}$ should be a good approximation of $R_{0}$. Matrices $\hat{U}\in\mathbb{R}^{N\times K}, \hat{V}\in\mathbb{R}^{J\times K}$ collect the corresponding left and right singular vectors and satisfy $\hat{U}'\hat{U}=\hat{V}'\hat{V}=I_{K\times K}$. The above analysis implies that $\hat{U}$ should have roughly $K$ distinct rows because $\hat{U}$ is a slightly perturbed version of $U$. Therefore, to obtain a good estimation of the classification matrix $Z$, we should apply the K-means algorithm on all rows of $\hat{U}$ with $K$ clusters. Let $\hat{Z}$ be the estimated classification matrix returned by applying the K-means method on all rows of $\hat{U}$ with $K$ clusters. Then we are able to obtain a good estimation of $\Theta$ according to Equation (\ref{ThetaZ}) by setting $\hat{\Theta}=\hat{V}\hat{\Sigma}\hat{U}'\hat{Z}(\hat{Z}'\hat{Z})^{-1}$.
Algorithm \ref{alg:SCK}, referred to as SCK, is a natural extension of the Ideal SCK from the oracle case to the real case. Note that in our SCK algorithm, there are only two inputs: the observed weighted response matrix $R$ and the number of latent classes $K$, i.e., SCK does not require any tuning parameters.
\begin{algorithm}
\caption{\textbf{Spectral Clustering with K-means (SCK for short)}}
\label{alg:SCK}
\begin{algorithmic}[1]
\Require The observed weighted response matrix $R\in\mathbb{R}^{N\times J}$ and the number of extreme latent profiles $K$.
\Ensure $\hat{Z}$ and $\hat{\Theta}$.
\State Obtain $\hat{R}=\hat{U}\hat{\Sigma} \hat{V}'$, the top $K$ SVD of $R$.
\State Run K-means algorithm on all rows of $\hat{U}$ with $K$ clusters to obtain $\hat{Z}$.
\State Obtain an estimate of $\Theta$ by setting $\hat{\Theta}=\hat{R}'\hat{Z}(\hat{Z}'\hat{Z})^{-1}$.
\end{algorithmic}
\end{algorithm}

Here, we evaluate the computational cost of our SCK algorithm. The computational cost of the SVD step involved in the SCK approach is $O(\max(N^{2}, J^{2})K)$. For the K-means algorithm, its complexity is $O(NlK^{2})$ with $l$ being the number of K-means iterations. In all experimental studies considered in this paper, $l$ is set as 100 for the K-means algorithm. The complexity of the last step in SCK is $O(JNK)$. Since $K \ll \min(N, J)$ in this paper, as a consequence, the total time complexity of our SCK algorithm is $O(\max(N^{2}, J^{2})K)$.
\section{Theoretical properties}\label{sec4}
In this section, we present comprehensive theoretical properties of the SCK algorithm when the observed weighted response matrix $R$ is generated from the proposed model. Our objective is to demonstrate that the estimated classification matrix $\hat{Z}$ and the estimated item parameter matrix $\hat{\Theta}$ both concentrate around the true classification matrix $Z$ and the true item parameter matrix $\Theta$, respectively.

Let $\mathcal{T}=\{\mathcal{T}_{1}, \mathcal{T}_{2}, \ldots, \mathcal{T}_{K}\}$ be the collection of true partitions for all subjects, where $\mathcal{T}_{k}=\{i: Z(i,k)=1 \mathrm{~for~}i\in[N]\}$ for $k\in[K]$, i.e., $\mathcal{T}_{k}$ is the set of true partition of subjects into the $k$-th extreme latent profile. Similarly, let $\hat{\mathcal{T}}=\{\hat{\mathcal{T}}_{1}, \hat{\mathcal{T}}_{2}, \ldots, \hat{\mathcal{T}}_{K}\}$ represent the collection of estimated partitions for all subjects, where $\hat{\mathcal{T}}_{k}=\{i: \hat{Z}(i,k)=1 \mathrm{~for~}i\in[N]\}$ for $k\in[K]$. We use the measure defined in \citep{joseph2016impact} to quantify the closeness of the estimated partition $\hat{\mathcal{T}}$ and the ground truth
partition $\mathcal{T}$. Denote the \emph{Clustering error} associated with $\mathcal{T}$ and $\hat{\mathcal{T}}$ as
\begin{align}\label{clusteringerror}
\hat{f}=\mathrm{min}_{\pi\in S_{K}}\mathrm{max}_{k\in[K]}\frac{|\mathcal{T}_{k}\cap \mathcal{\hat{\mathcal{T}}}^{c}_{\pi(k)}|+|\mathcal{T}^{c}_{k}\cap \mathcal{\hat{\mathcal{T}}}_{\pi(k)}|}{N_{K}},
\end{align}
where $S_{K}$ represents the set of all permutations of $\{1,2,\ldots, K\}$, $\mathcal{\hat{\mathcal{T}}}^{c}_{\pi(k)}$ and $\mathcal{T}^{c}_{k}$ denote the complementary sets. As stated in the reference \citep{joseph2016impact}, $\hat{f}$ evaluates the maximum proportion of subjects in the symmetric difference of $\mathcal{T}_{k}$ and $\hat{\mathcal{T}}_{\pi(k)}$. Since the observed weighted response matrix $R$ is generated from WLCM with expectation $R_{0}$, and $\hat{f}$ measures the performance of the SCK algorithm, it is expected that SCK estimates $Z$ with small Clustering error $\hat{f}$.

For convenience, let $\rho=\mathrm{max}_{j\in[J],k\in[K]}|\Theta(j,k)|$ and call it the scaling parameter. Let $B=\frac{\Theta}{\rho}$, we have $\mathrm{max}_{j\in[J],k\in[K]}|B(j,k)|=1$ and $R_{0}=\rho ZB'$. Let $\tau=\mathrm{max}_{i\in[N],j\in[J]}|R(i,j)-R_{0}(i,j)|$ and $\gamma=\mathrm{max}_{i\in[N],j\in[J]}\mathrm{Var}(R(i,j))$ where $\mathrm{Var}(R(i,j))$ means the variance of $R(i,j)$. We require the following assumption to establish theoretical guarantees of consistency for our SCK method.
\begin{assum}\label{asump}
Assume	$\gamma\geq \frac{\tau^{2}\mathrm{log}(N+J)}{\mathrm{max}(N,J)}$.
\end{assum}
The following theorem presents our main result, which provides upper bounds for the error rates of our SCK algorithm under our WLCM model.
\begin{thm}\label{mainWLCM}
	Under $WLCM(Z,\Theta,\mathcal{F})$, if Assumption \ref{asump} is satisfied, with probability at least $1-o((N+J)^{-3})$,
\begin{align*}
&\hat{f}=O(\frac{\gamma K^{2}N_{\mathrm{max}}\mathrm{max}(N,J)\mathrm{log}(N+J)}{\rho^{2}N^{2}_{\mathrm{min}}J}) \mathrm{~and~}\frac{\|\hat{\Theta}-\Theta\mathcal{P}\|_{F}}{\|\Theta\|_{F}}=O(\frac{\sqrt{\gamma K\mathrm{max}(N,J)\mathrm{log}(N+J)}}{\rho\sqrt{N_{\mathrm{min}}J}}),
\end{align*}
where $N_{\mathrm{max}}=\mathrm{max}_{k\in[K]}\{N_{k}\}, N_{\mathrm{min}}=\mathrm{min}_{k\in[K]}\{N_{k}\}$, and $\mathcal{P}$ is a permutation matrix.
\end{thm}
Because our WLCM is distribution-free, Theorem \ref{mainWLCM} provides a general theoretical guarantee of the SCK algorithm when $R$ is generated from WLCM for any distribution $\mathcal{F}$ as long as Equation (\ref{RFR0}) is satisfied. We can simplify Theorem \ref{mainWLCM} by considering additional conditions:
\begin{cor}\label{AddConditions}
Under $WLCM(Z,\Theta,\mathcal{F})$, when Assumption \ref{asump} holds, if we make the additional assumption that $\frac{N_{\mathrm{max}}}{N_{\mathrm{min}}}=O(1)$ and $K=O(1)$, with probability at least $1-o((N+J)^{-3})$,
\begin{align*}
\hat{f}=O(\frac{\gamma \mathrm{max}(N,J)\mathrm{log}(N+J)}{\rho^{2}NJ}) \mathrm{~and~}\frac{\|\hat{\Theta}-\Theta\mathcal{P}\|_{F}}{\|\Theta\|_{F}}=O(\frac{\sqrt{\gamma\mathrm{max}(N,J)\mathrm{log}(N+J)}}{\rho\sqrt{NJ}}).
\end{align*}
\end{cor}
For the case $J=\beta N$ for any positive constant $\beta$, Corollary \ref{AddConditions} implies that the SCK algorithm yields consistent estimation under WLCM since the error bounds in Corollary \ref{AddConditions} decrease to zero as $N\rightarrow+\infty$ when $\rho$ and distribution $\mathcal{F}$ are fixed.

Recall that $R$ is an observed weighted response matrix generated from a distribution $\mathcal{F}$ with expectation $R_{0}=Z\Theta'=\rho ZB'$ under the WLCM model and $\gamma$ is the maximum variance of $R(i,j)$ and it is closely related to the distribution $\mathcal{F}$, the ranges of $R,\rho,B$, and $\gamma$ can vary depending on the specific distribution $\mathcal{F}$. The following examples provide the ranges of $R,\rho, B$, the upper bound of $\gamma$, and the explicit forms of error bounds in Theorem \ref{mainWLCM} for different distribution $\mathcal{F}$ under our WLCM model. Meanwhile, based on the explicitly derived error bounds for different distribution $\mathcal{F}$, we also investigate how the scaling parameter $\rho$ influences the performance of the SCK algorithm in these examples. For all pairs $(i,j)$ with $i\in[N], j\in[J]$, we consider the following distributions when $\mathbb{E}(R)=R_{0}$ in Equation (\ref{RFR0}) holds.
\begin{Ex}\label{Bernoulli}
Let $\mathcal{F}$ be \textbf{Bernoulli distribution} such that $R(i,j)\sim \mathrm{Bernoulli}(R_{0}(i,j))$, where $R_{0}(i,j)$ is the Bernoulli probability, i.e., $\mathbb{E}(R(i,j))=R_{0}(i,j)$. For this case, our WLCM degenerates to the LCM model. According to the properties of the Bernoulli distribution, we have the following conclusions.
\begin{itemize}
\item $R(i,j)\in\{0,1\}$, i.e., $R(i,j)$ only takes two values 0 and 1.
\item $B(i,j)\in[0,1]$ and $\rho\in(0,1]$ because $R_{0}(i,j)$ is a probability located in $[0,1]$ and $\mathrm{max}_{i\in[N],j\in[J]}|B(i,j)|$ is assumed to be 1.
\item $\tau\leq1$ because $\tau=\mathrm{max}_{i\in[N],j\in[J]}|R(i,j)-R_{0}(i,j)|\leq1$.
\item $\gamma\leq\rho$ because $\gamma=\mathrm{max}_{i\in[N],j\in[J]}\mathrm{Var}(R(i,j))=\mathrm{max}_{i\in[N],j\in[J]}R_{0}(i,j)(1-R_{0}(i,j))\leq\mathrm{max}_{i\in[N],j\in[J]} R_{0}(i,j)=\mathrm{max}_{i\in[N],j\in[J]} \rho (ZB)(i,j)\leq\rho$.
\item Let $\tau$ be its upper bound 1 and $\gamma$ be its upper bound $\rho$, Assumption \ref{asump} becomes $\rho\geq\frac{\mathrm{log}(N+J)}{\mathrm{max}(N,J)}$, which means a sparsity requirement on $R$ because $\rho$ controls the probability of the numbers of ones in $R$ for this case.
\item Let $\gamma$ be its upper bound $\rho$ in Theorem \ref{mainWLCM}, we have
\begin{align*}
\hat{f}=O(\frac{K^{2}N_{\mathrm{max}}\mathrm{max}(N,J)\mathrm{log}(N+J)}{\rho N^{2}_{\mathrm{min}}J}) \mathrm{~and~}\frac{\|\hat{\Theta}-\Theta\mathcal{P}\|_{F}}{\|\Theta\|_{F}}=O(\frac{\sqrt{K\mathrm{max}(N,J)\mathrm{log}(N+J)}}{\sqrt{\rho N_{\mathrm{min}}J}}).
\end{align*}
We observe that increasing $\rho$ leads to a decrease in SCK's error rates when $\mathcal{F}$ is a Bernoulli distribution.
\end{itemize}
\end{Ex}
\begin{Ex}\label{Binomial}
Let $\mathcal{F}$ be \textbf{Binomial distribution} such that $R(i,j)\sim\mathrm{Binomial}(m,\frac{R_{0}(i,j)}{m})$ for any positive integer $m$, where $R(i,j)$ is a random variable that reflects the number of successes in a fixed number of independent trials $m$ with the same probability of success $\frac{R_{0}(i,j)}{m}$, i.e., $\mathbb{E}(R(i,j))=R_{0}(i,j)$. For Binomial distribution, we have $\mathbb{P}(R(i,j)=r)=\binom{m}{r}(\frac{R_{0}(i,j)}{m})^{r}(1-\frac{R_{0}(i,j)}{m})^{m-r}$ for $r=0,1,2,\ldots,m$, where $\binom{\bullet}{\bullet}$ is a binomial coefficient. By the property of the Binomial distribution, we have the following conclusions.
\begin{itemize}
  \item $R(i,j)\in\{0,1,2,\ldots,m\}$.
  \item $B(i,j)\in[0,1]$ and $\rho\in(0,m]$ because $\frac{R_{0}(i,j)}{m}$ is a probability that ranges in $[0,1]$.
  \item $\tau\leq m$ because $\tau=\mathrm{max}_{i\in[N],j\in[J]}|R(i,j)-R_{0}(i,j)|\leq m$.
  \item $\gamma\leq\rho$ because $\gamma=\mathrm{max}_{i\in[N],j\in[J]}\mathrm{Var}(R(i,j))=m\frac{R_{0}(i,j)}{m}(1-\frac{R_{0}(i,j)}{m})=R_{0}(i,j)(1-\frac{R_{0}(i,j)}{m})\leq\rho$.
  \item Let $\tau$ be its upper bound $m$ and $\gamma$ be its upper bound $\rho$, Assumption \ref{asump} becomes $\rho\geq\frac{m^{2}\mathrm{log}(N+J)}{\mathrm{max}(N,J)}$ which provides a lower bound requirement of the scaling parameter $\rho$.
  \item Let $\gamma$ be its upper bound $\rho$ in Theorem \ref{mainWLCM}, we obtain the exact forms of error bounds for SCK when $\mathcal{F}$ is a Binomial distribution, and we observe that increasing $\rho$ reduces SCK's error rates.
\end{itemize}
\end{Ex}
\begin{Ex}\label{Poisson}
Let $\mathcal{F}$ be \textbf{Poisson distribution} such that $R(i,j)\sim \mathrm{Poisson}(R_{0}(i,j))$, where $R_{0}(i,j)$ is the Poisson parameter, i.e., $\mathbb{E}(R(i,j))=R_{0}(i,j)$. By the properties of the Poisson distribution, the following conclusions can be obtained.
\begin{itemize}
  \item $R(i,j)\in\mathbb{N}$, i.e., $R(i,j)$ is an nonnegative integer.
  \item $B(i,j)]\in[0,1]$ and $\rho\in(0,+\infty)$ because Poisson distribution can take any positive value for its mean.
  \item $\tau$ is an unknown positive value because we cannot know the exact upper bound of $R(i,j)$ when $R$ is obtained from the Poisson distribution under the WLCM model.
  \item $\gamma\leq\rho$ because $\gamma=\mathrm{max}_{i\in[N],j\in[J]}\mathrm{Var}(R(i,j))=\mathrm{max}_{i\in[N],j\in[J]}R_{0}(i,j)\leq\rho$.
  \item Let $\gamma$ be its upper bound $\rho$, Assumption \ref{asump} becomes $\rho\geq\frac{\tau^{2}\mathrm{log}(N+j)}{\mathrm{max}(N,J)}$ which is a lower bound requirement of $\rho$.
\item Let $\gamma$ be its upper bound $\rho$ in Theorem \ref{mainWLCM} obtains the exact forms of error bounds for the SCK algorithm when $\mathcal{F}$ is a Poisson distribution. It is easy to observe that increasing $\rho$ leads to a decrease in SCK's error rates.
\end{itemize}
\end{Ex}
\begin{Ex}\label{Normal}
Let $\mathcal{F}$ be \textbf{Normal distribution} such that $R(i,j)\sim \mathrm{Normal}(R_{0}(i,j),\sigma^{2})$, where $R_{0}(i,j)$ is the mean ( i.e., $\mathbb{E}(R(i,j))=R_{0}(i,j)$) and $\sigma^{2}$ is the variance parameter for Normal distribution. For this case, we have
\begin{itemize}
\item $R(i,j)\in\mathbb{R}$, i.e., $R(i,j)$ is a real value.
\item $B(i,j)\in[-1,1]$ and $\rho\in(0,+\infty)$ because the mean of Normal distribution can take any value. Note that, unlike the cases when $\mathcal{F}$ is Bernoulli or Poisson, $B$ can have negative elements for the Normal distribution case.
\item Similar to Example \ref{Poisson}, $\tau$ is an unknown positive value.
\item $\gamma=\sigma^{2}$ because $\gamma=\mathrm{max}_{i\in[N],j\in[J]}\mathrm{Var}(R(i,j))=\mathrm{max}_{i\in[N],j\in[J]}\sigma^{2}=\sigma^{2}$ for Normal distribution.
\item Let $\gamma$ be its exact value $\sigma^{2}$, Assumption \ref{asump} becomes $\sigma^{2}\mathrm{max}(N,J)\geq\tau^{2}\mathrm{log}(N+J)$ which means that $\mathrm{max}(N,J)$ should be set larger than $\frac{\tau^{2}\mathrm{log}(N+J)}{\sigma^{2}}$ for our theoretical analysis.
\item Let $\gamma$ be its exact value $\sigma^{2}$ in Theorem \ref{mainWLCM} provides the exact forms of error bounds for SCK. We observe that increasing the scaling parameter $\rho$ (or decreasing the variance $\sigma^{2}$) reduces SCK's error rates.
\end{itemize}
\end{Ex}
\begin{Ex}\label{Exponential}
Let $\mathcal{F}$ be \textbf{Exponential distribution} such that $R(i,j)\sim \mathrm{Exponential}(\frac{1}{R_{0}(i,j)})$, where $\frac{1}{R_{0}(i,j)}$ is the Exponential parameter, i.e., $\mathbb{E}(R(i,j))=R_{0}(i,j)$. For this case, we have
\begin{itemize}
\item $R(i,j)\in\mathbb{R}_{+}$, i.e., $R(i,j)$ is a positive value.
\item $B(i,j)\in(0,1]$ and $\rho\in(0,+\infty)$ because the mean of Exponential distribution can be any positive value.
\item Similar to Example \ref{Poisson}, $\tau$ is an unknown positive value.
\item $\gamma\leq\rho^{2}$ because $\gamma=\mathrm{max}_{i\in[N],j\in[J]}\mathrm{Var}(R(i,j))=\mathrm{max}_{i\in[N],j\in[J]}R^{2}_{0}(i,j)\leq\rho^{2}$ for Exponential distribution.
\item Let $\gamma$ be its upper bound $\rho^{2}$, Assumption \ref{asump} becomes $\rho^{2}\geq\tau^{2}\mathrm{log}(N+J)/\mathrm{max}(N,J)$, a lower bound requirement of $\rho$.
\item Let $\gamma$ be its upper bound $\rho^{2}$ in Theorem \ref{mainWLCM}, the theoretical bounds demonstrate that $\rho$ vanishes, which indicates that increasing $\rho$ has no significant impact on the error rates of SCK.
\end{itemize}
\end{Ex}
\begin{Ex}\label{Uniform}
Let $\mathcal{F}$ be \textbf{Uniform distribution} such that $R(i,j)\sim\mathrm{Uniform}(0,2R_{0}(i,j))$, where $\mathbb{E}(R(i,j))=\frac{0+2R_{0}(i,j)}{2}=R_{0}(i,j)$ holds immediately. For this case, we have
\begin{itemize}
\item $R(i,j)\in(0,2\rho)$ because $2R_{0}(i,j)\leq2\rho$.
\item $B(i,j)\in(0,1]$ and $\rho\in(0,+\infty)$ because $\mathrm{Uniform}(0,2R_{0}(i,j))$ allows $2R_{0}(i,j)$ to be any positive value.
\item $\tau$ is an unknown positive value with an upper bound $2\rho$.
\item $\gamma\leq\frac{\rho^{2}}{3}$ because $\gamma=\mathrm{max}_{i\in[N],j\in[J]}\mathrm{Var}(R(i,j))=\mathrm{max}_{i\in[N],j\in[J]}\frac{(2R_{0}(i,j)-0)^{2}}{12}=\mathrm{max}_{i\in[N],j\in[J]}\frac{R^{2}_{0}(i,j)}{3}\leq\frac{\rho^{2}}{3}$ for Uniform distribution.
\item Let $\gamma$ be its upper bound $\frac{\rho^{2}}{3}$, Assumption \ref{asump} becomes $\rho^{2}\geq3\tau^{2}\mathrm{log}(N+J)/\mathrm{max}(N,J)$, a lower bound requirement of $\rho$.
\item Since $\rho$ disappears in the error bounds when we let $\gamma=\frac{\rho^{2}}{3}$ in Theorem \ref{mainWLCM}, increasing $\rho$ does not significantly influence SCK's error rates, a conclusion similar to Example \ref{Exponential}.
\end{itemize}
\end{Ex}
\begin{Ex}\label{Signed}
Our WLCM can also model \textbf{signed response matrix} by setting $\mathbb{P}(R(i,j)=1)=\frac{1+R_{0}(i,j)}{2}$ and $\mathbb{P}(R(i,j)=-1)=\frac{1-R_{0}(i,j)}{2}$, where $\mathbb{E}(R(i,j))=\frac{1+R_{0}(i,j)}{2}-\frac{1-R_{0}(i,j)}{2}=R_{0}(i,j)$ and Equation (\ref{RFR0}) holds surely. For the signed response matrix, we have
\begin{itemize}
\item $R(i,j)\in\{-1,1\}$, i.e., $R(i,j)$ only takes two values -1 and 1.
\item $B(i,j)\in[-1,1]$ and $\rho\in(0,1]$ because $\frac{1+R_{0}(i,j)}{2}$ and $\frac{1-R_{0}(i,j)}{2}$ are two probabilities which should range in $[0,1]$. Note that similar to Example \ref{Normal}, $B(i,j)$ can be negative for the signed response matrix.
\item $\tau\leq2$ because $R(i,j)\in\{-1,1\}$ and $R_{0}(i,j)\in[-1,1]$.
\item $\gamma\leq1$ because $\gamma=\mathrm{max}_{i\in[N],j\in[J]}\mathrm{Var}(R(i,j))=\mathrm{max}_{i\in[N],j\in[J]}(1-R^{2}_{0}(i,j))\leq1$.
\item When setting $\tau=2$ and $\gamma=1$, Assumption \ref{asump} turns to be $\mathrm{max}(N,J)\geq4\mathrm{log}(N+J)$.
\item Setting $\gamma$ as its upper bound $1$ in Theorem \ref{mainWLCM} gives that increasing $\rho$ reduces SCK's error rates.
\end{itemize}
\end{Ex}
\section{Simulation studies}\label{sec5}
In this section, we conduct extensive simulation experiments to evaluate the effectiveness of the proposed method and validate our theoretical results in Examples \ref{Bernoulli}-\ref{Signed}.
\subsection{Baseline method}
More than the SCK algorithm, here we briefly provide an alternative spectral method that can also be applied to fit our WLCM model.  Recall that $R_{0}=Z\Theta'$ under WLCM, it is easy to see that $R_{0}(i,:)=R_{0}(\bar{i},:)$ when two distinct subjects $i$ and $\bar{i}$ belong to the same extreme latent profile for $i,\bar{i}\in[N]$. Therefore, the population response matrix $R_{0}$ features K disparate rows, and running the K-means approach on all rows of $R_{0}$ with $K$ clusters can faithfully recover the classification matrix $Z$ in terms of a permutation of the $K$ extreme latent profiles. $R_{0}=Z\Theta'$ also gives that $\Theta=R'_{0}Z(Z'Z)^{-1}$, which suggests the following ideal algorithm called Ideal RMK.
\begin{algorithm}
\caption{\textbf{Ideal RMK}}
\label{alg:IdealRMK}
\begin{algorithmic}[1]
\Require $R_{0}, K$.
\Ensure A permutation of $Z$ and $\Theta$.
\State Run K-means algorithm on all rows of $R_{0}$ with $K$ clusters to obtain $Z\mathcal{P}$, a permutation of $Z$.
\State Compute $R'_{0}Z\mathcal{P}((Z\mathcal{P})'Z\mathcal{P})^{-1}=\Theta\mathcal{P}$, a permutation of $\Theta$.
\end{algorithmic}
\end{algorithm}

Algorithm \ref{alg:RMK} called RMK is a natural generalization of the Ideal RMK from the oracle case to the real case because $\mathbb{E}(R)=R_{0}$ under the WLCM model. Unlike the SCK method, the RMK method does not need to obtain the SVD of the observed weighted response matrix $R$.
\begin{algorithm}
\caption{\textbf{Response Matrix with K-means (RMK for short)}}
\label{alg:RMK}
\begin{algorithmic}[1]
\Require $R,K$.
\Ensure $\hat{Z},\hat{\Theta}$.
\State Run K-means algorithm on all rows of $R$ with $K$ clusters to obtain $\hat{Z}$.
\State Obtain an estimate of $\Theta$ by setting $\hat{\Theta}=R'\hat{Z}(\hat{Z}'\hat{Z})^{-1}$.
\end{algorithmic}
\end{algorithm}

The computational cost of the first step in RMK is $O(lNJK)$, where $l$ denotes the number of iterations for the K-means algorithm. The complexity of the second step in RMK is $O(JNK)$. Therefore, the overall computational cost of RMK is $O(lNJK)$. When $J=\beta N$ for a constant value $\beta\in(0,1]$, the complexity of RMK is $O(\beta lKN^{2})$, and it is larger than the SCK's complexity  $O(KN^{2})$ when $\beta l>1$. Therefore, SCK runs faster than RMK when $\beta l>1$, as confirmed by our numerical results in this section.
\subsection{Evaluation metric}
For the classification of subjects, when the true classification matrix $Z$ is known, to evaluate how good the quality of the partition of the subjects into extreme latent profiles, four metrics are considered including the Clustering error $\hat{f}$ computed by Equation (\ref{clusteringerror}). The other three popular evaluation criteria are Hamming error \citep{SCORE}, normalized mutual
information (NMI) \citep{strehl2002cluster,danon2005comparing,bagrow2008evaluating,luo2017community}, and adjusted rand index (ARI) \citep{luo2017community,hubert1985comparing,vinh2009information}.
\begin{itemize}
  \item Hamming error is defined as
\begin{align*}
\mathrm{Hamming~error}=N^{-1}\mathrm{min}_{\mathcal{P}\in\mathcal{P}_{K}}\|\hat{Z}-Z\mathcal{P}\|_{0},
\end{align*}
where $\mathcal{P}_{K}$ denotes the collection of all $K$-by-$K$ permutation matrices. Hamming error falls within the range $[0,1]$, and a smaller Hamming error indicates better classification performance.
  \item Let $C$ be a $K\times K$ confusion matrix such that $C(k,l)$ is the number of common subjects between $\mathcal{T}_{k}$ and $\hat{\mathcal{T}}_{l}$ for $k,l\in[K]$. NMI is defined as
\begin{align*}
\mathrm{NMI}=\frac{-2\sum_{k,l}C(k,l)\mathrm{log}(\frac{C(k,l)N}{C_{k.}C_{.l}})}{\sum_{k}C_{k.}\mathrm{log}(\frac{C_{k.}}{N})+\sum_{l}C_{.l}\mathrm{log}(\frac{C_{.l}}{N})},
\end{align*}
where $C_{k.}=\sum_{m=1}^{K}C(k,m)$ and $C_{.l}=\sum_{m=1}^{K}C(m,l)$. NMI ranges in $[0,1]$ and it is the larger the better.
  \item ARI is defined as
  \begin{align*}
\mathrm{ARI}=\frac{\sum_{k,l}\binom{C(k,l)}{2}-\frac{\sum_{k}\binom{C_{k.}}{2}\sum_{l}\binom{C_{.l}}{2}}{\binom{N}{2}}}{\frac{1}{2}[\sum_{k}\binom{C_{k.}}{2}+\sum_{l}\binom{C_{.l}}{2}]-\frac{\sum_{k}\binom{C_{k.}}{2}\sum_{l}\binom{C_{.l}}{2}}{\binom{N}{2}}},
  \end{align*}
where $\binom{.}{.}$ is a binomial coefficient. ARI falls within the range [-1,1] and it is the larger the better.
\end{itemize}
For the estimation of $\Theta$, we use the Relative $l_{1}$ error and the Relative $l_{2}$ error to evaluate the performance. The two criteria are defined as
\begin{align*}
\mathrm{Relative~}l_{1}\mathrm{~error}=\mathrm{min}_{\mathcal{P}\in\mathcal{P}_{K}}\frac{\|\hat{\Theta}-\Theta\mathcal{P}\|_{1}}{\|\Theta\|_{1}}\mathrm{~and~}\mathrm{Relative~}l_{2}\mathrm{~error}=\mathrm{min}_{\mathcal{P}\in\mathcal{P}_{K}}\frac{\|\hat{\Theta}-\Theta\mathcal{P}\|_{F}}{\|\Theta\|_{F}}.
\end{align*}
Both measures are the smaller the better.
\subsection{Synthetic weighted response matrices}\label{sec5Synthetic}
We conduct numerical studies to examine the accuracy and the efficiency of our SCK and RMK approaches by changing the scaling parameter $\rho$ and the number of subjects $N$. Unless specified, in all computer-generated weighted response matrices, we set $K=3, J=\frac{N}{5}$, and the $N\times K$ classification matrix $Z$ is generated such that each subject belongs to one of the $K$ extreme latent profiles with equal probability. For distributions that require $B$'s entries to be nonnegative, we let $B(j,k)=\mathrm{rand}(1)$ for $j\in[J],k\in[K]$, where $\mathrm{rand}(1)$ is a random value simulated from the uniform distribution on $[0,1]$. For Normal distribution and signed response matrix that allow $B$ to have negative entries, we let $B(j,k)=2\mathrm{rand}(1)-1$ for $j\in[J],k\in[K]$, i.e., $B(j,k)$ ranges in $[-1,1]$. Set $B_{\mathrm{max}}=\mathrm{max}_{j\in[J],k\in[K]}|B(j,k)|$. Because the generation process of $B$ makes $|B(j,k)|\in[0,1]$ but cannot guarantee that $B_{\mathrm{max}}=1$ which is required in the definition of $B$. Therefore, we update $B$ by $\frac{B}{B_{\mathrm{max}}}$. For the scaling parameter $\rho$ and the number of subjects $N$, they are set independently for each distribution. After setting all model parameters $(K,N,J,Z,B,\rho)$, we can generate the observed weighted response matrix $R$ from distribution $\mathcal{F}$ with expectation $R_{0}=Z\Theta'=\rho ZB'$ under our WLCM model. By applying the SCK method (and the RMK method) to $R$ with $K$ extreme latent profiles, we can compute the evaluation metrics of SCK (and RMK). In every simulation scenario, we generate 50 independent replicates and report the mean of Clustering error (as well as Hamming error, NMI, ARI, Relative $l_{1}$ error, Relative $l_{2}$ error, and running time) computed from the 50 repetitions for each method. 
\subsubsection{Bernoulli distribution}
When $R(i,j)\sim \mathrm{Bernoulli}(R_{0}(i,j))$ for $i\in[N], j\in[J]$, we consider the following two simulations.

\textbf{Simulation 1(a): changing $\rho$.} Set $N=500$. For the Bernoulli distribution, the scaling parameter $\rho$ should be set within the range $(0,1]$ according to Example \ref{Bernoulli}. Here, for simulation studies, we let $\rho$ range in $\{0.1,0.2,0.3,\ldots,1\}$.

\textbf{Simulation 1(b): changing $N$.} Let $\rho=0.1$ and $N$ range in $\{1000,2000,\ldots,5000\}$.

The results are presented in Figure \ref{S1}. We observe that SCK outperforms RMK because SCK returns more accurate estimations of $(Z,\Theta)$ and SCK runs faster than RMK across all settings. Both methods achieve better performances as $\rho$ increases, which conforms to our analysis in Example \ref{Bernoulli}. Additionally, both algorithms enjoy better performances when the number of subjects $N$ increases, as predicted by our analysis following Corollary \ref{AddConditions}.

\begin{figure}
\centering
\subfigure[Simulation 1(a)]{\includegraphics[width=0.24\textwidth]{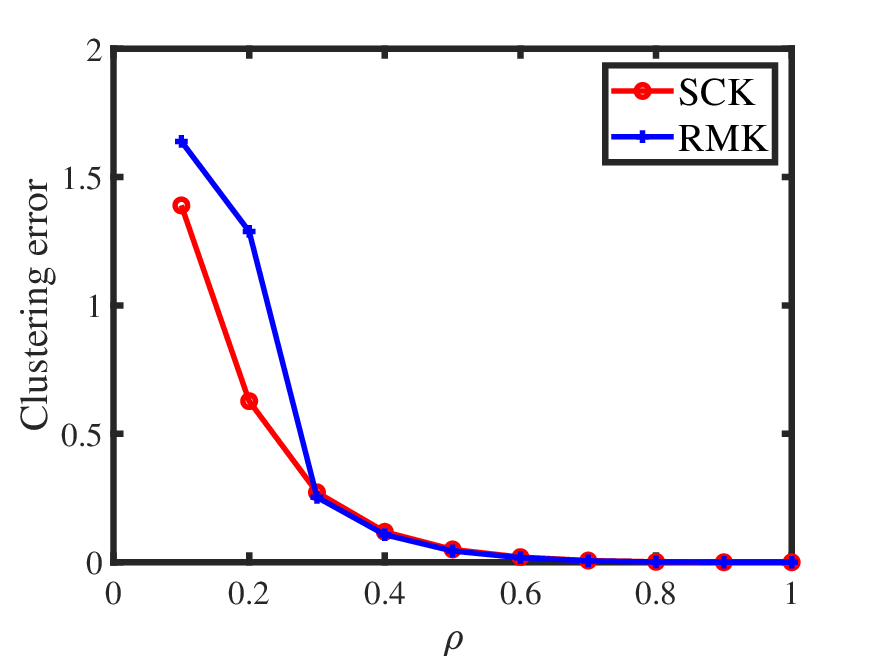}}
\subfigure[Simulation 1(a)]{\includegraphics[width=0.24\textwidth]{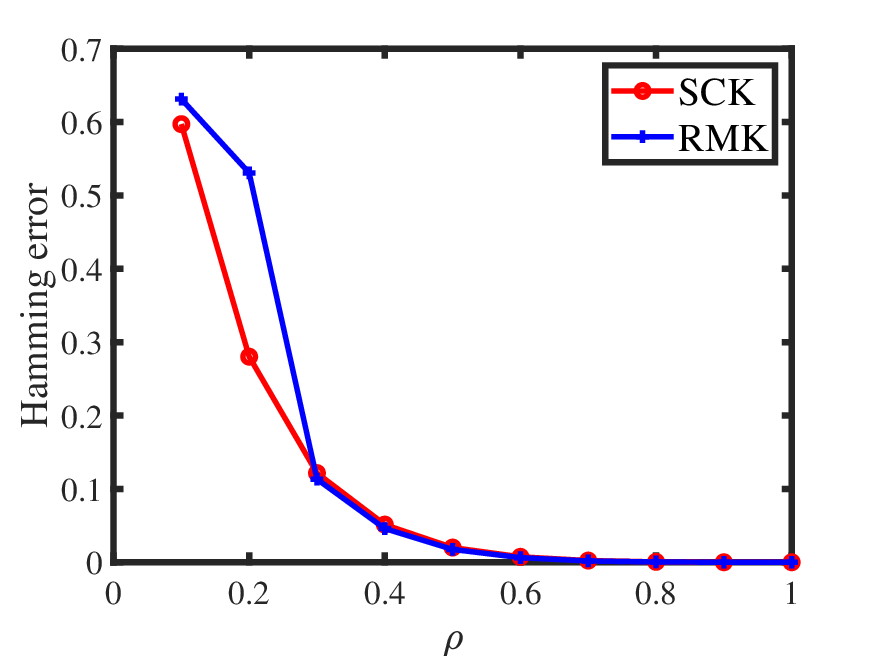}}
\subfigure[Simulation 1(a)]{\includegraphics[width=0.24\textwidth]{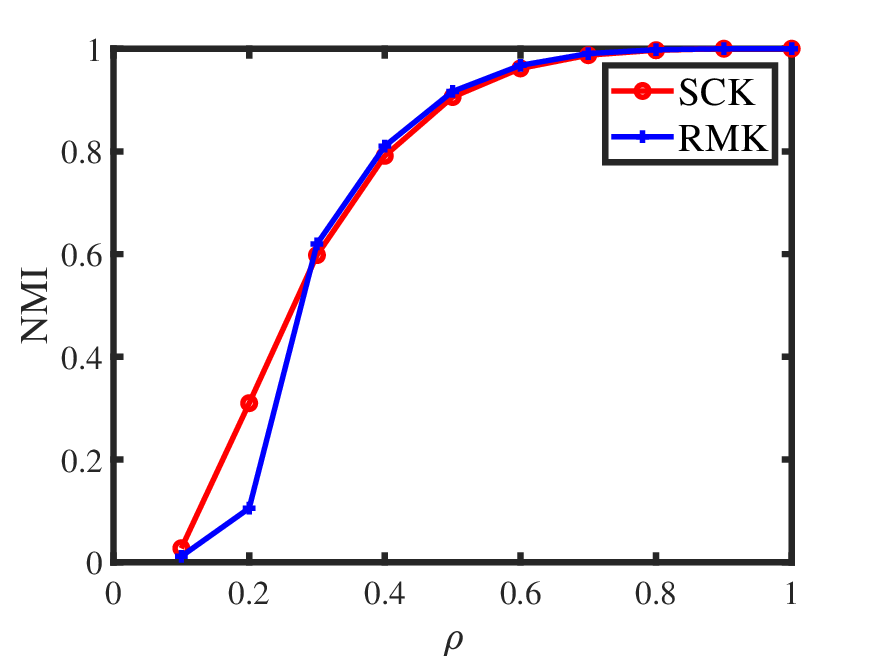}}
\subfigure[Simulation 1(a)]{\includegraphics[width=0.24\textwidth]{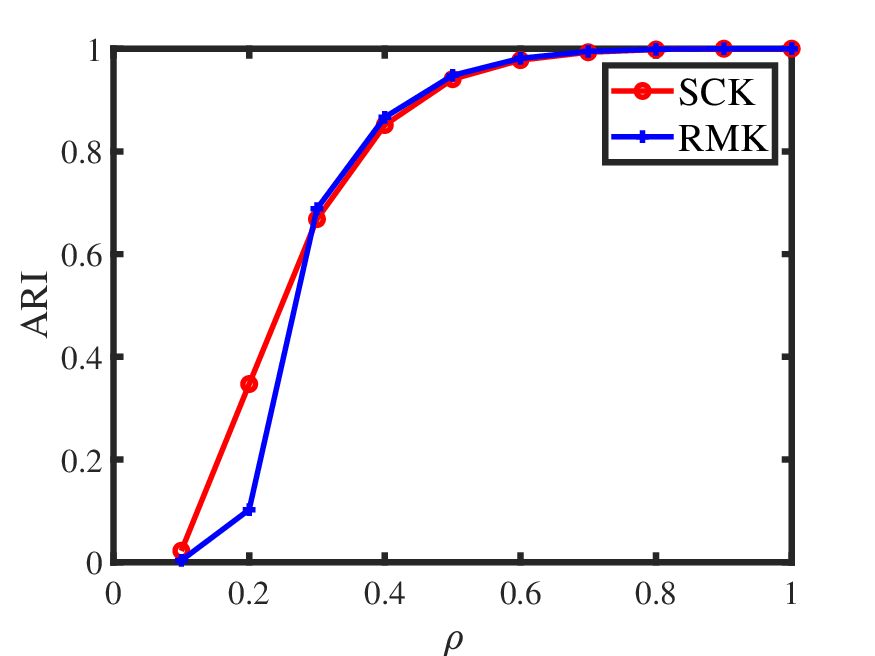}}
\subfigure[Simulation 1(a)]{\includegraphics[width=0.24\textwidth]{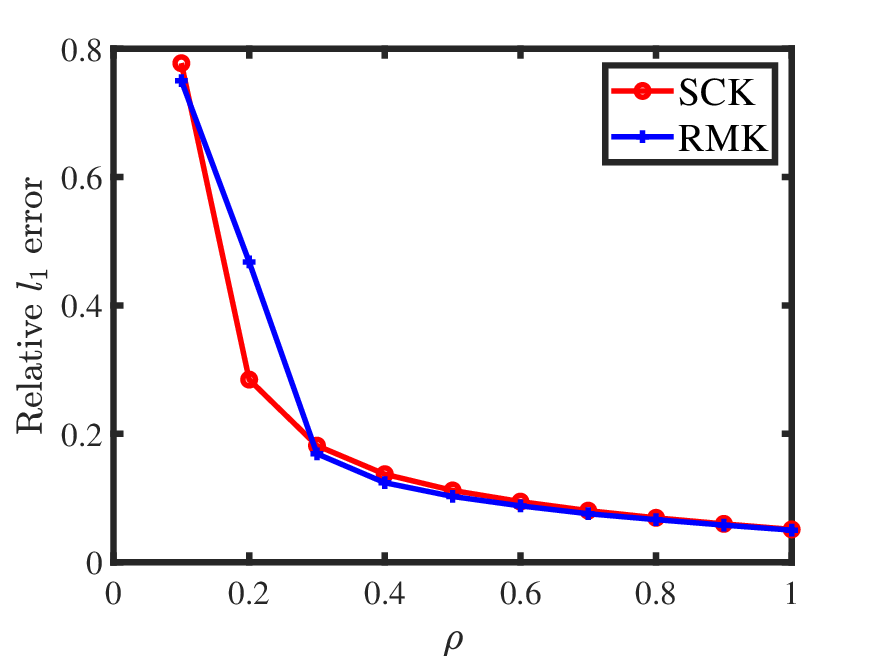}}
\subfigure[Simulation 1(a)]{\includegraphics[width=0.24\textwidth]{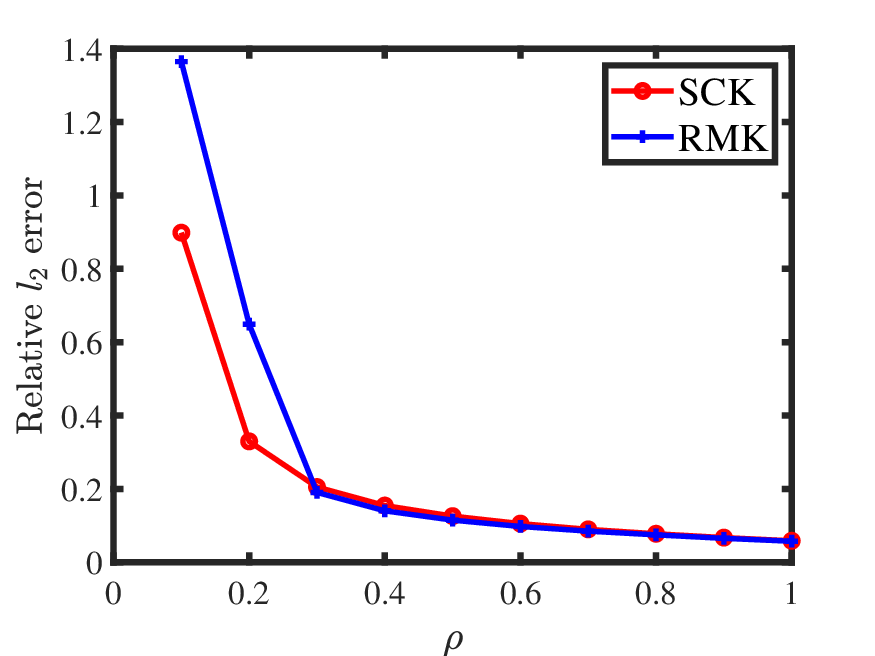}}
\subfigure[Simulation 1(a)]{\includegraphics[width=0.24\textwidth]{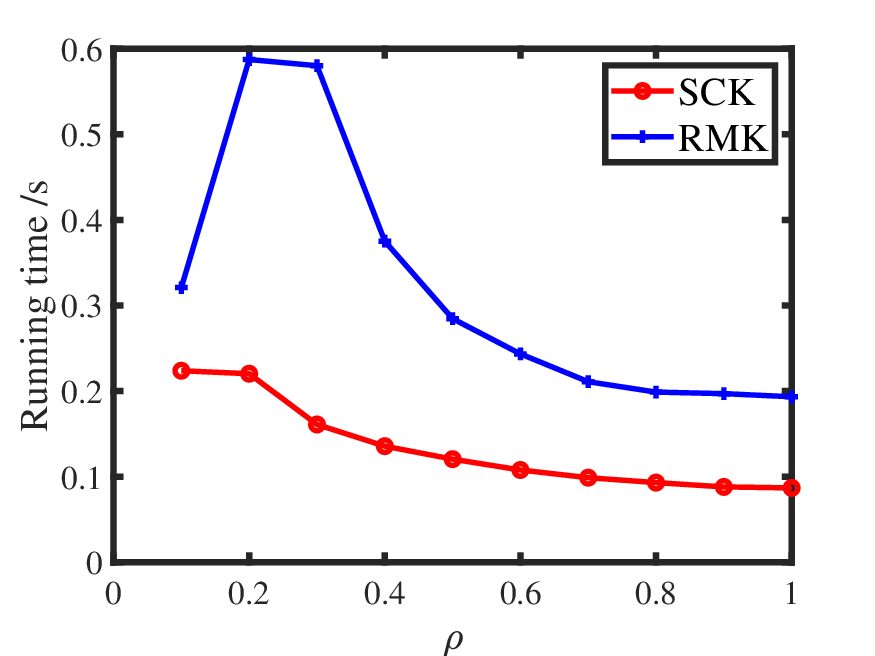}}
\subfigure[Simulation 1(b)]{\includegraphics[width=0.24\textwidth]{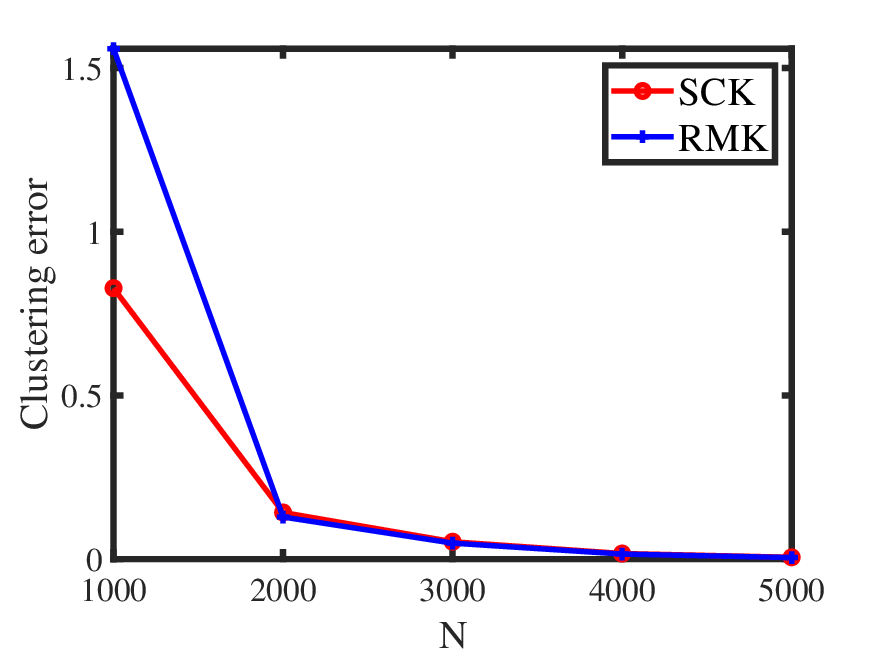}}
\subfigure[Simulation 1(b)]{\includegraphics[width=0.24\textwidth]{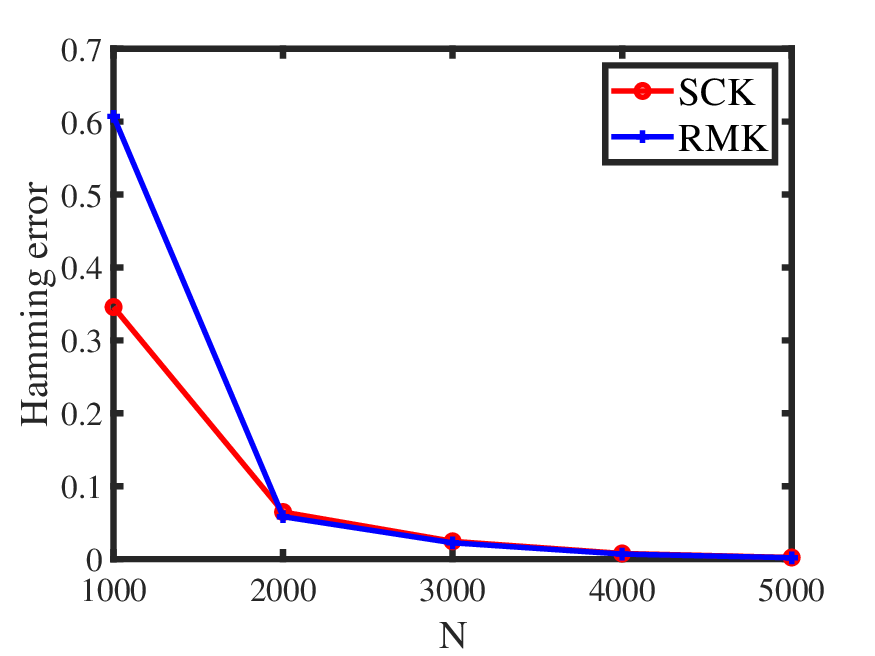}}
\subfigure[Simulation 1(b)]{\includegraphics[width=0.24\textwidth]{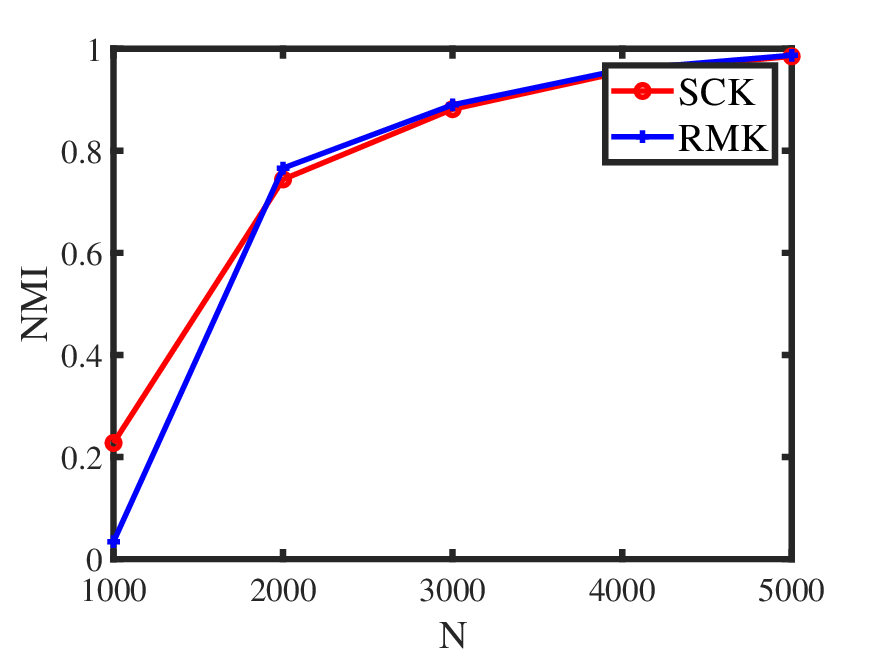}}
\subfigure[Simulation 1(b)]{\includegraphics[width=0.24\textwidth]{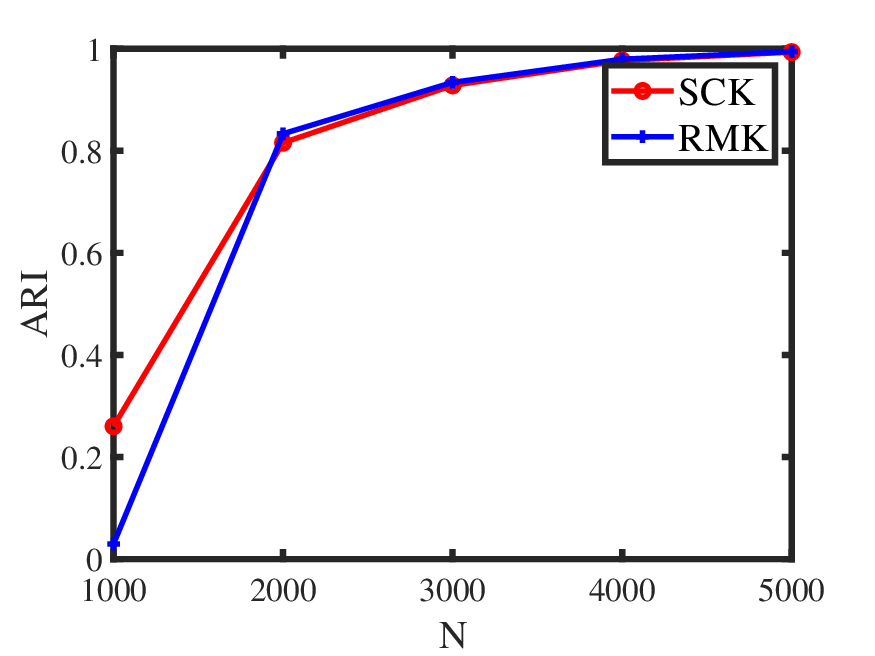}}
\subfigure[Simulation 1(b)]{\includegraphics[width=0.24\textwidth]{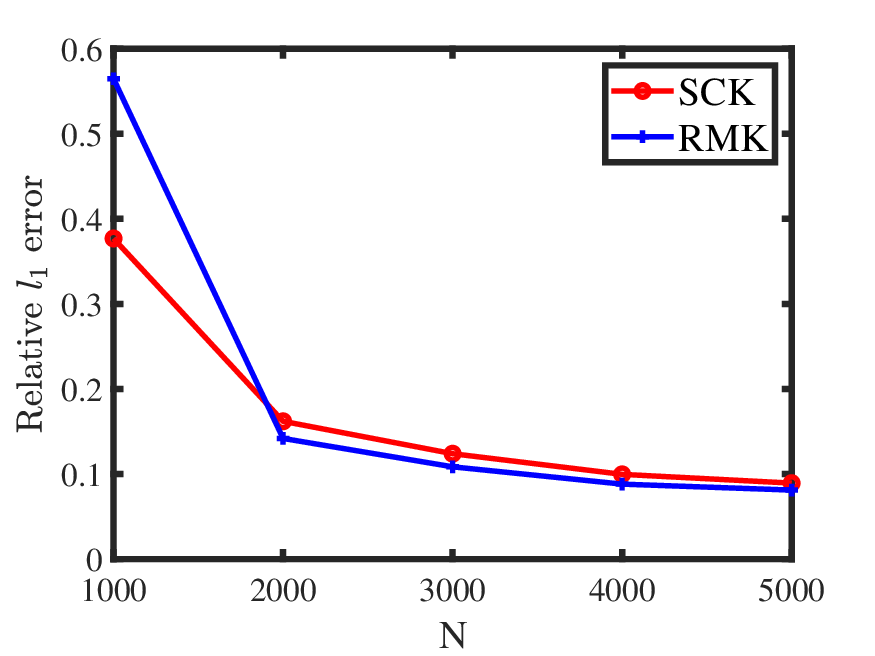}}
\subfigure[Simulation 1(b)]{\includegraphics[width=0.24\textwidth]{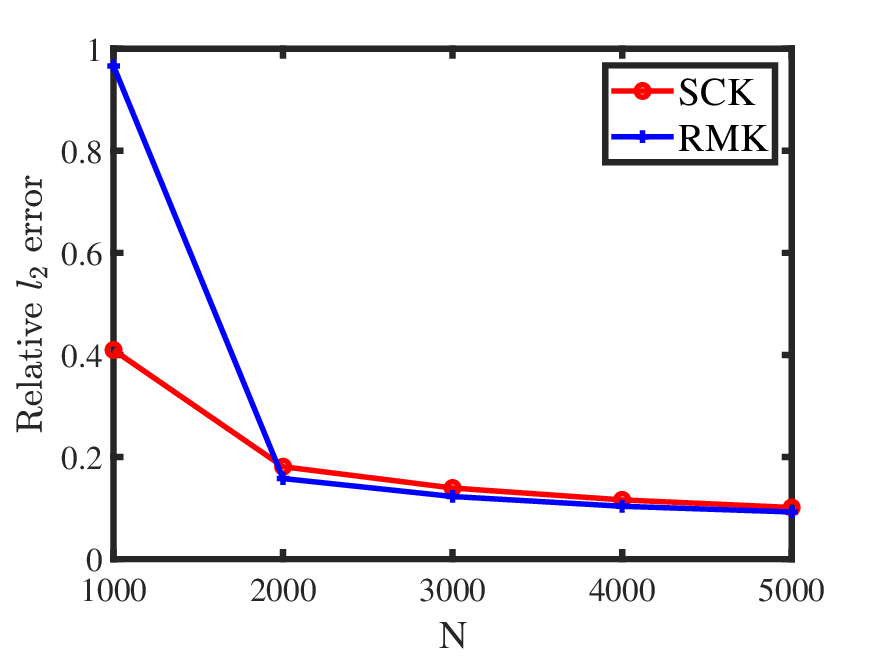}}
\subfigure[Simulation 1(b)]{\includegraphics[width=0.24\textwidth]{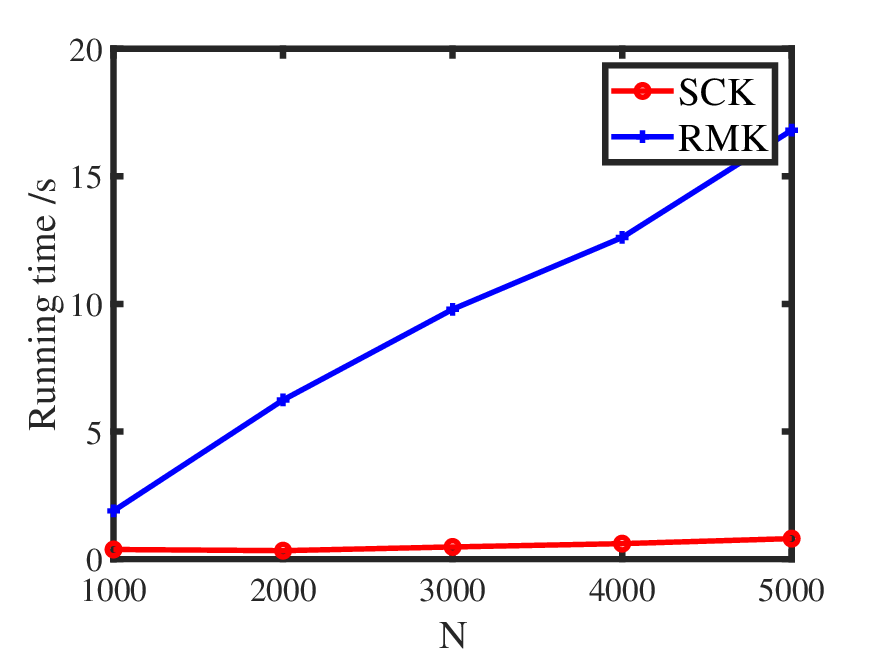}}
\caption{Numerical results of Simulation 1.}
\label{S1} 
\end{figure}
\subsubsection{Binomial distribution}
When $R(i,j)\sim \mathrm{Binomial}(m,\frac{R_{0}(i,j)}{m})$ for $i\in[N], j\in[J]$, we consider the following two simulations.

\textbf{Simulation 2(a): changing $\rho$.} Set $N=500$ and $m=5$. Recall that $\rho$'s range is $(0,m]$ when $\mathcal{F}$ is Binomial distribution according to Example \ref{Binomial}, here, we let $\rho$ range in $\{0.2,0.4,0.6,\ldots,2\}$.

\textbf{Simulation 2(b): changing $N$.} Let $\rho=0.1, m=5$, and $N$ range in $\{1000,2000,\ldots,5000\}$.

Figure \ref{S2} presents the corresponding results. We note that SCK and RMK have similar error rates, while SCK runs faster than RMK for this simulation. Meanwhile, increasing $\rho$ (and $N$) decreases error rates for both methods, which confirms our findings in Example \ref{Binomial} and Corollary \ref{AddConditions}.

\begin{figure}
\centering
\subfigure[Simulation 2(a)]{\includegraphics[width=0.24\textwidth]{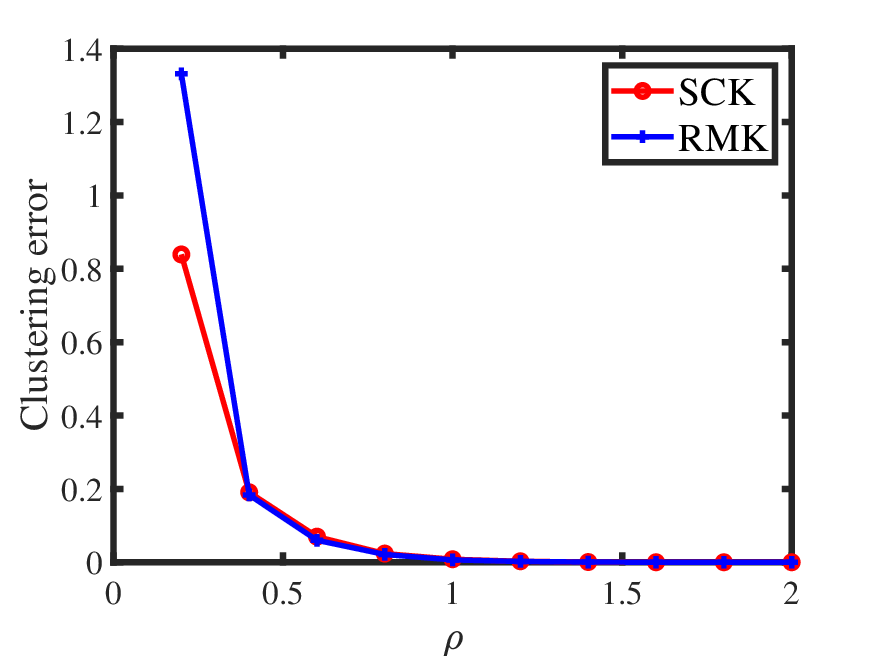}}
\subfigure[Simulation 2(a)]{\includegraphics[width=0.24\textwidth]{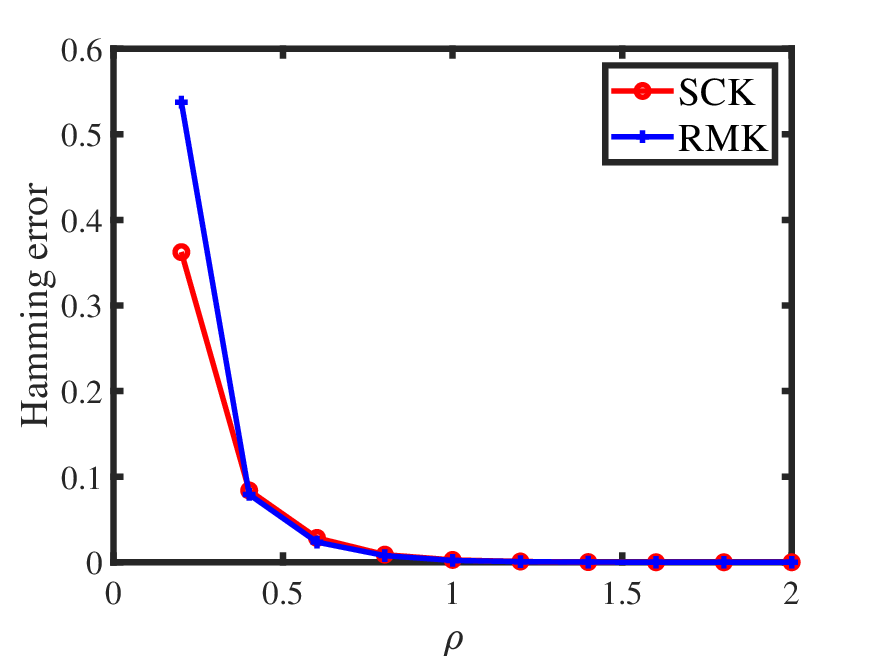}}
\subfigure[Simulation 2(a)]{\includegraphics[width=0.24\textwidth]{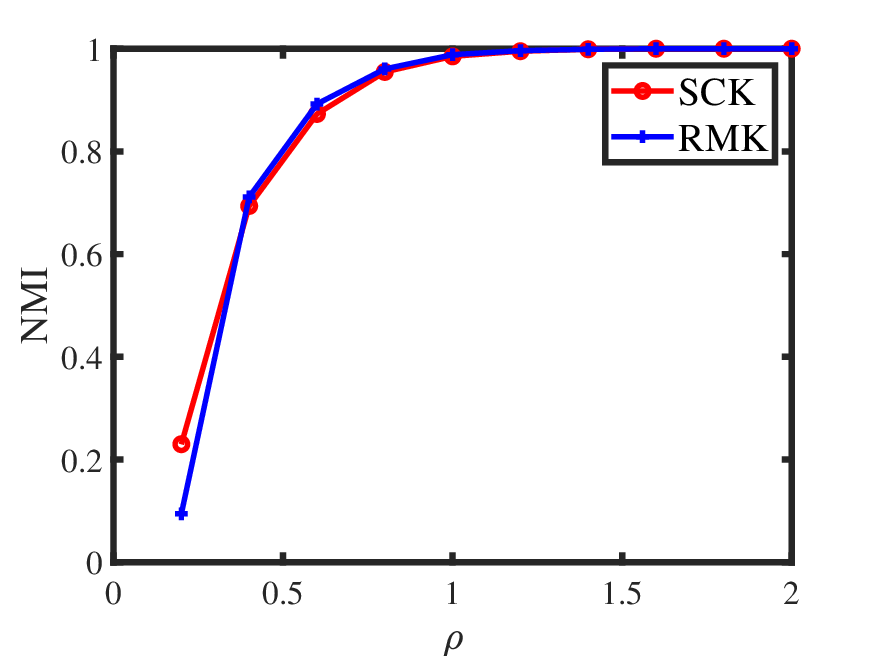}}
\subfigure[Simulation 2(a)]{\includegraphics[width=0.24\textwidth]{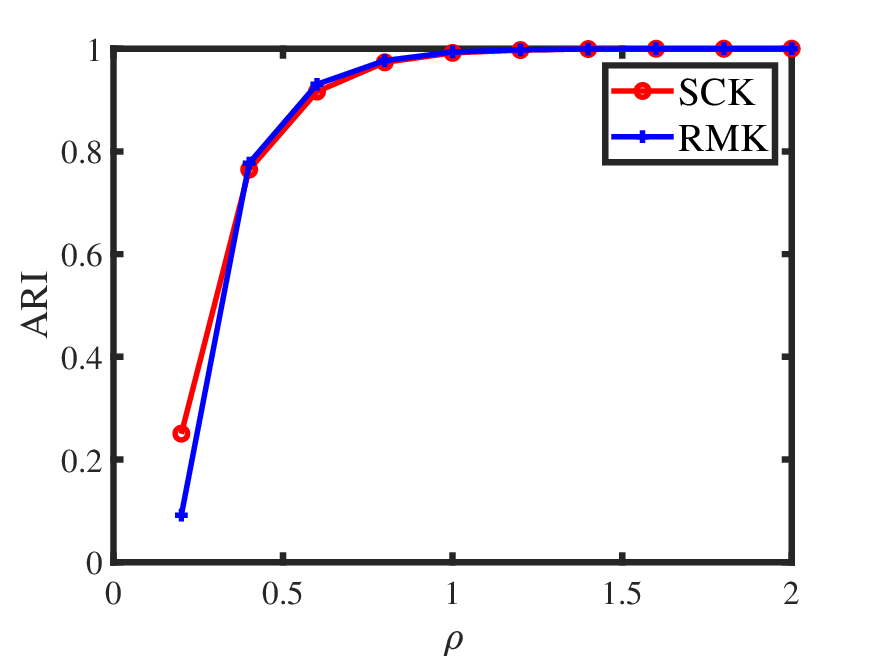}}
\subfigure[Simulation 2(a)]{\includegraphics[width=0.24\textwidth]{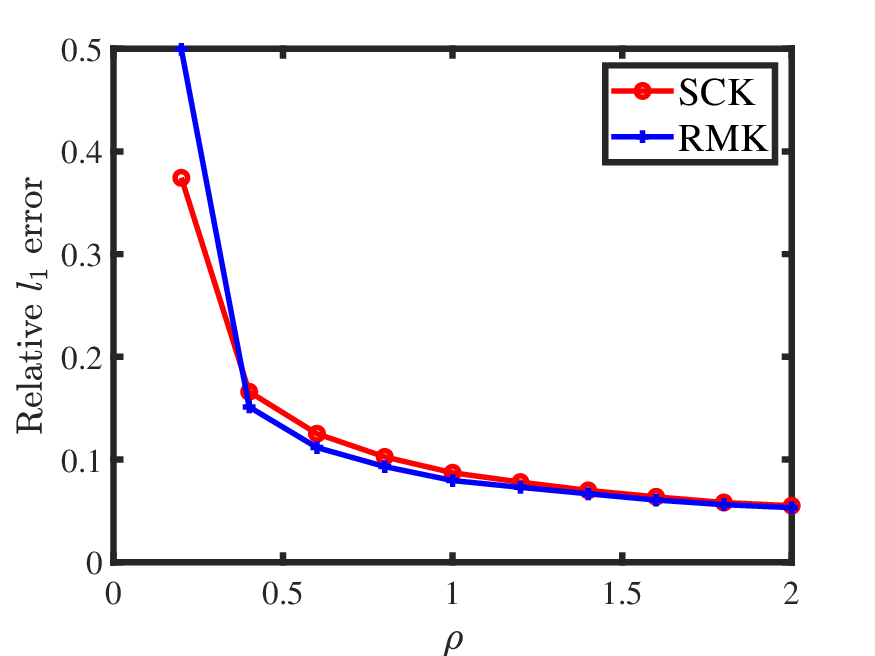}}
\subfigure[Simulation 2(a)]{\includegraphics[width=0.24\textwidth]{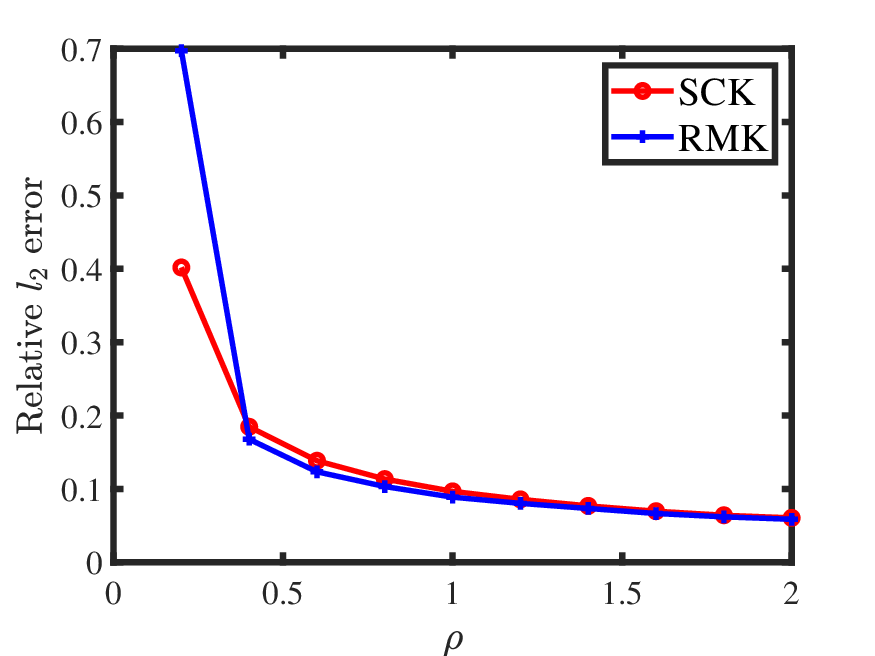}}
\subfigure[Simulation 2(a)]{\includegraphics[width=0.24\textwidth]{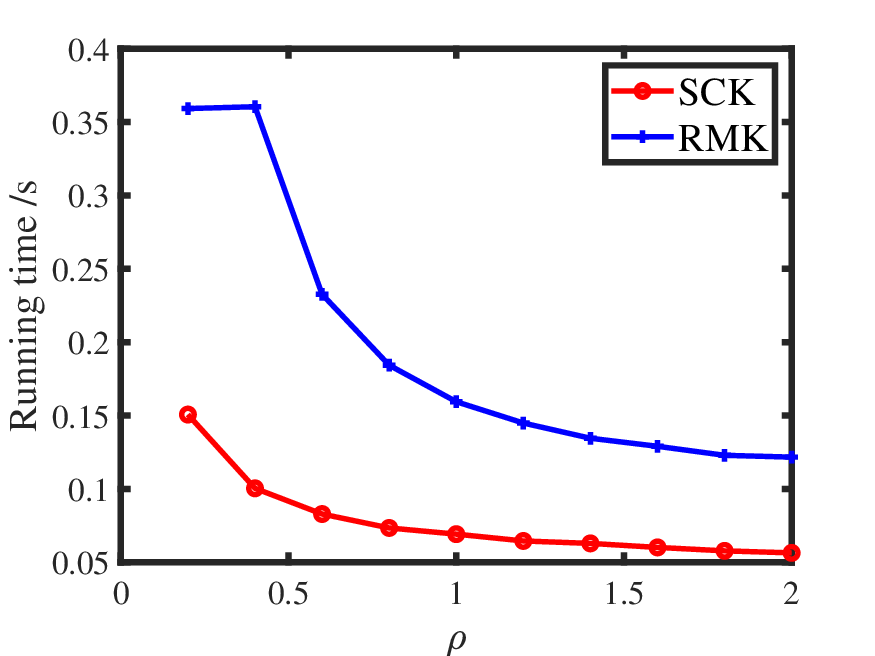}}
\subfigure[Simulation 2(b)]{\includegraphics[width=0.24\textwidth]{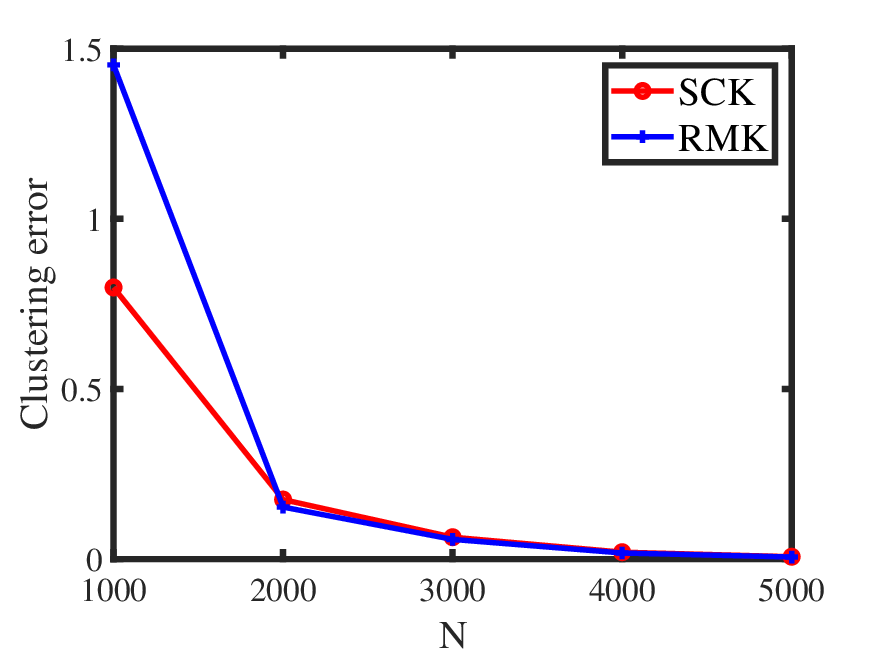}}
\subfigure[Simulation 2(b)]{\includegraphics[width=0.24\textwidth]{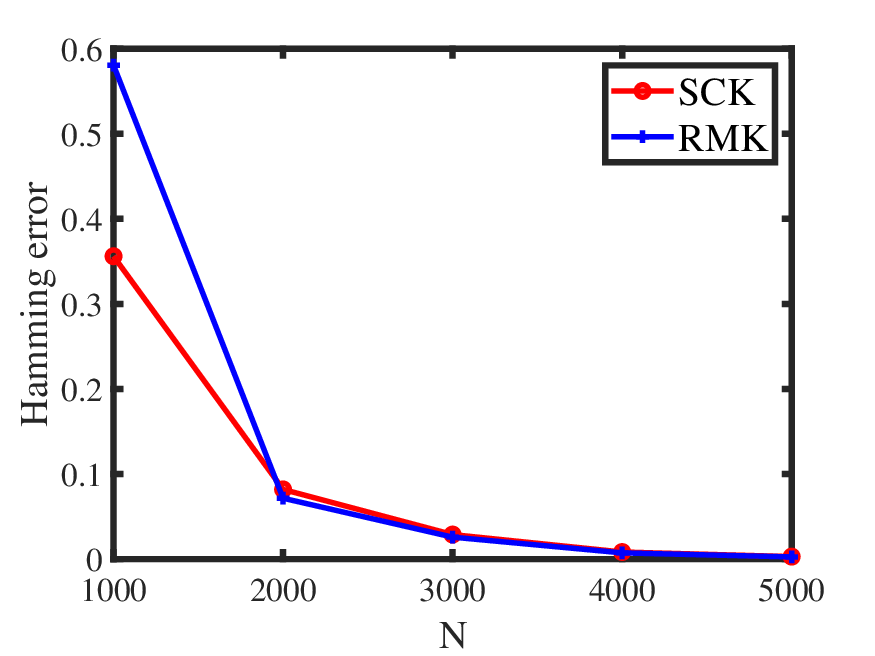}}
\subfigure[Simulation 2(b)]{\includegraphics[width=0.24\textwidth]{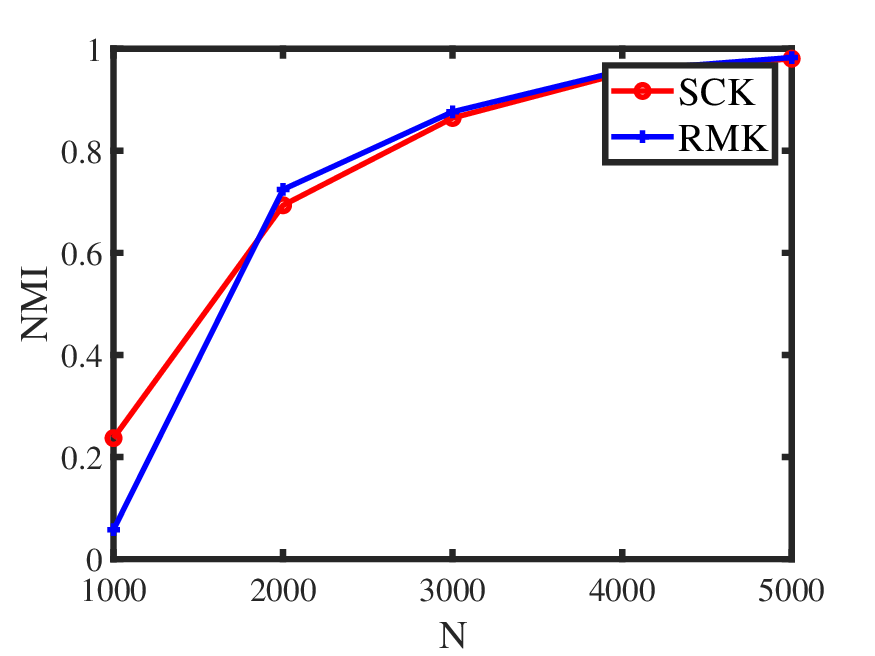}}
\subfigure[Simulation 2(b)]{\includegraphics[width=0.24\textwidth]{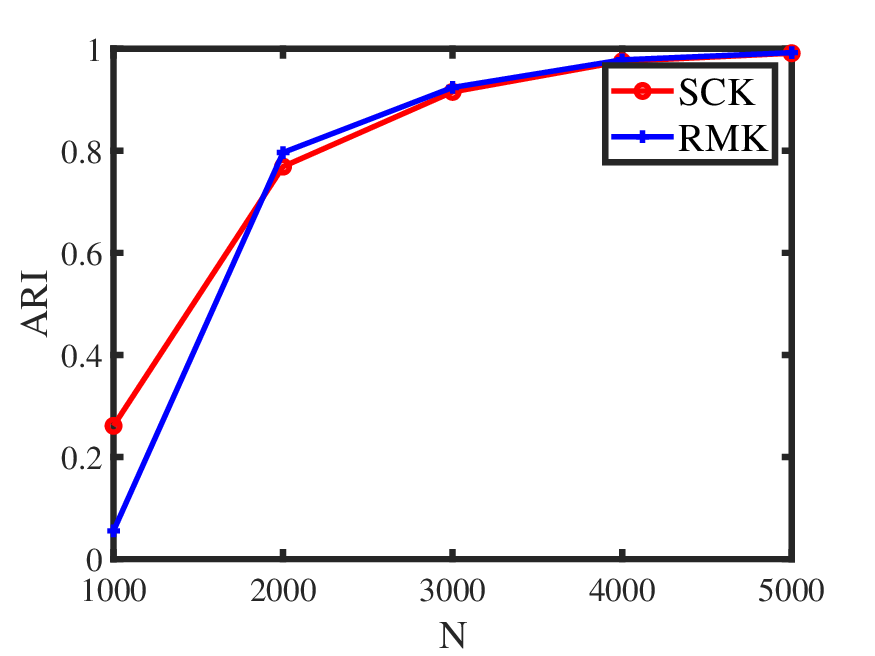}}
\subfigure[Simulation 2(b)]{\includegraphics[width=0.24\textwidth]{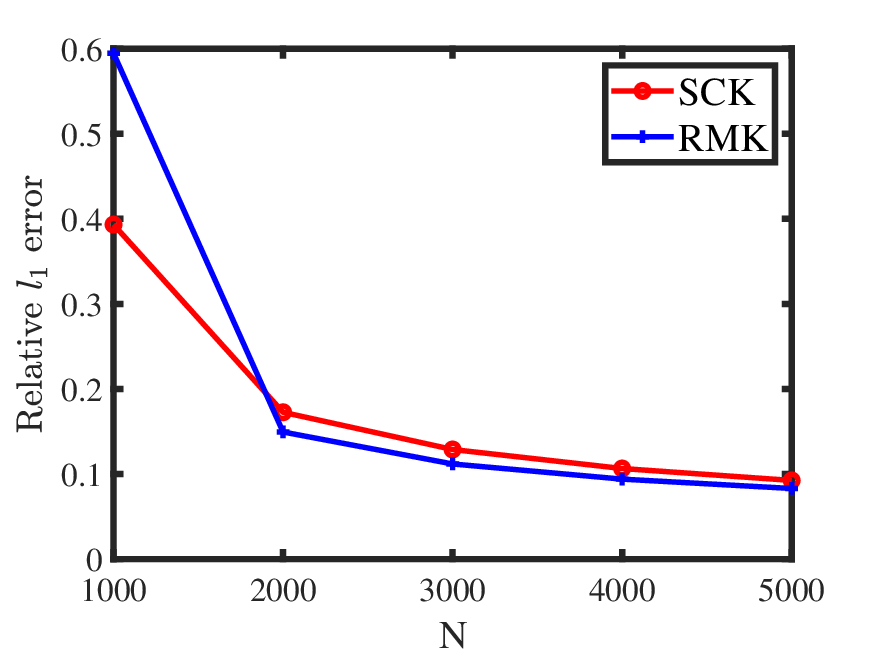}}
\subfigure[Simulation 2(b)]{\includegraphics[width=0.24\textwidth]{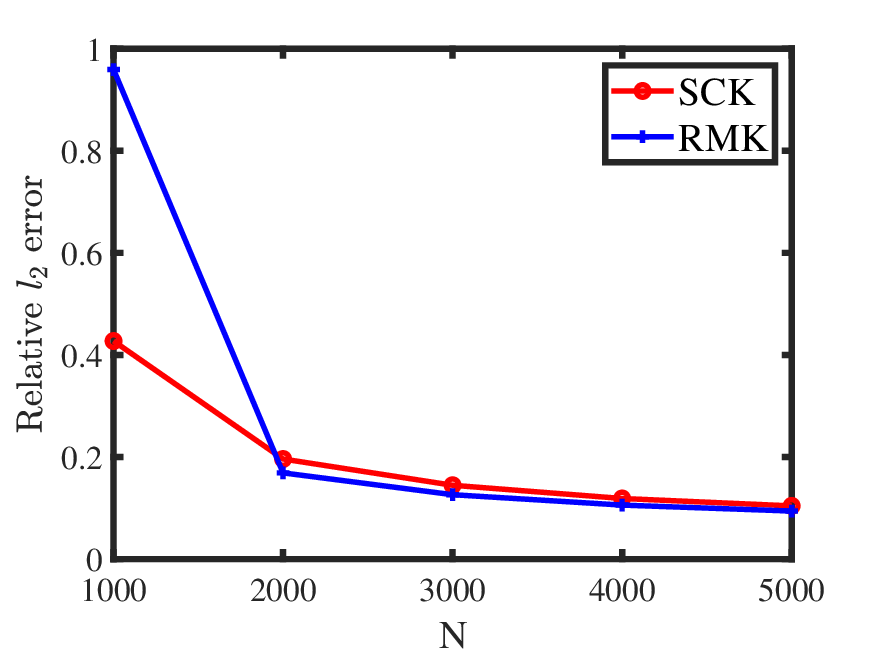}}
\subfigure[Simulation 2(b)]{\includegraphics[width=0.24\textwidth]{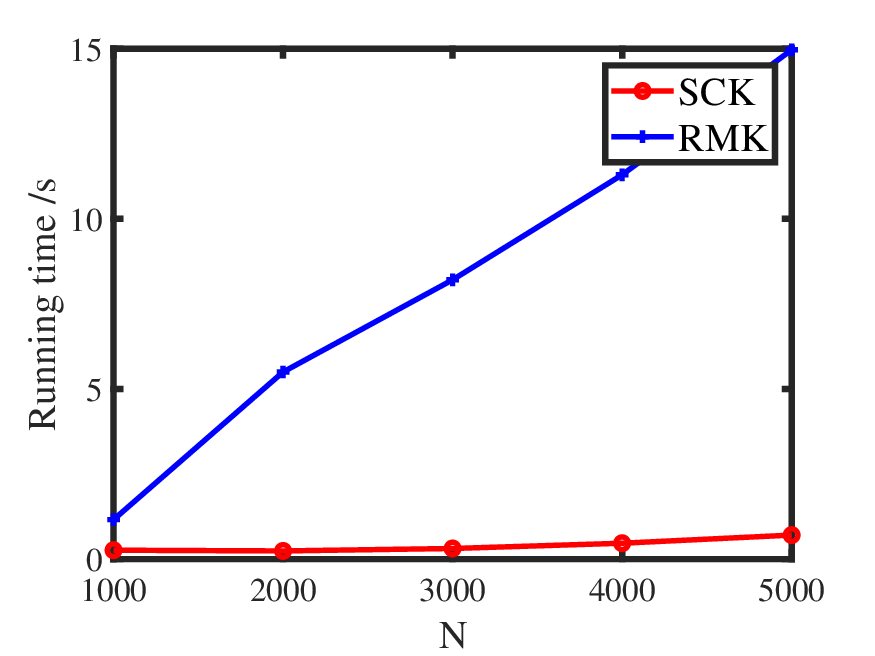}}
\caption{Numerical results of Simulation 2.}
\label{S2} 
\end{figure}
\subsubsection{Poisson distribution}
When $R(i,j)\sim \mathrm{Poisson}(R_{0}(i,j))$ for $i\in[N], j\in[J]$, we consider the following two simulations.

\textbf{Simulation 3(a): changing $\rho$.} Set $N=500$. Example \ref{Poisson} says that the theoretical range of $\rho$ is $(0,+\infty)$ when $\mathcal{F}$ is Poisson distribution. Here, we let $\rho$ range in $\{0.2,0.4,0.6,\ldots,2\}$.

\textbf{Simulation 3(b): changing $N$.} Let $\rho=0.1$ and $N$ range in $\{1000,2000,\ldots,5000\}$.

Figure \ref{S3} displays the numerical results of Simulation 3(a) and Simulation 3(b). The results are similar to those of the Bernoulli distribution case: SCK outperforms RMK in both estimating $(Z,\Theta)$ and running time; Both methods perform better as $\rho$ and $N$ increase, which supports our analysis in Example \ref{Poisson} and Corollary \ref{AddConditions}.
\begin{figure}
\centering
\subfigure[Simulation 3(a)]{\includegraphics[width=0.24\textwidth]{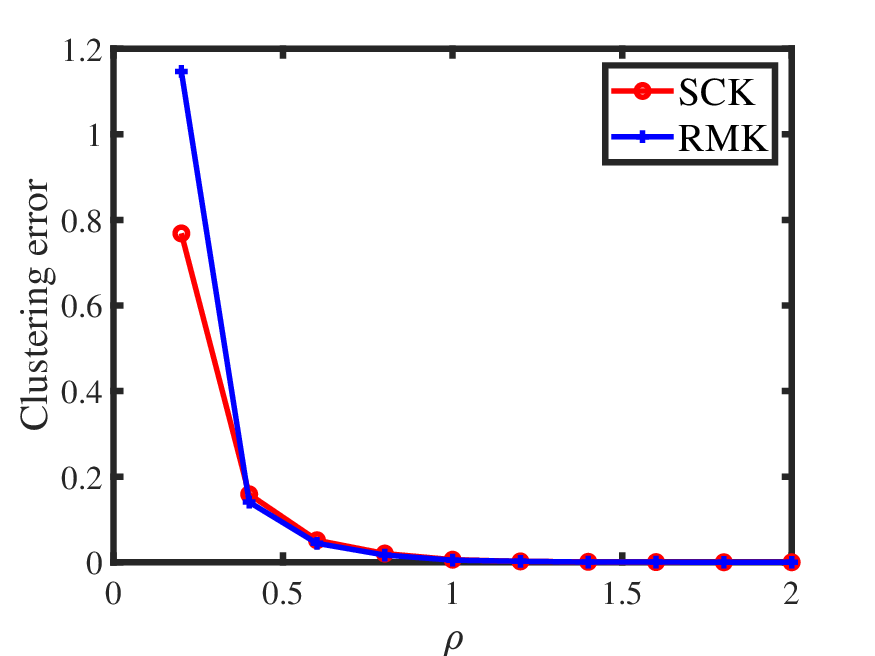}}
\subfigure[Simulation 3(a)]{\includegraphics[width=0.24\textwidth]{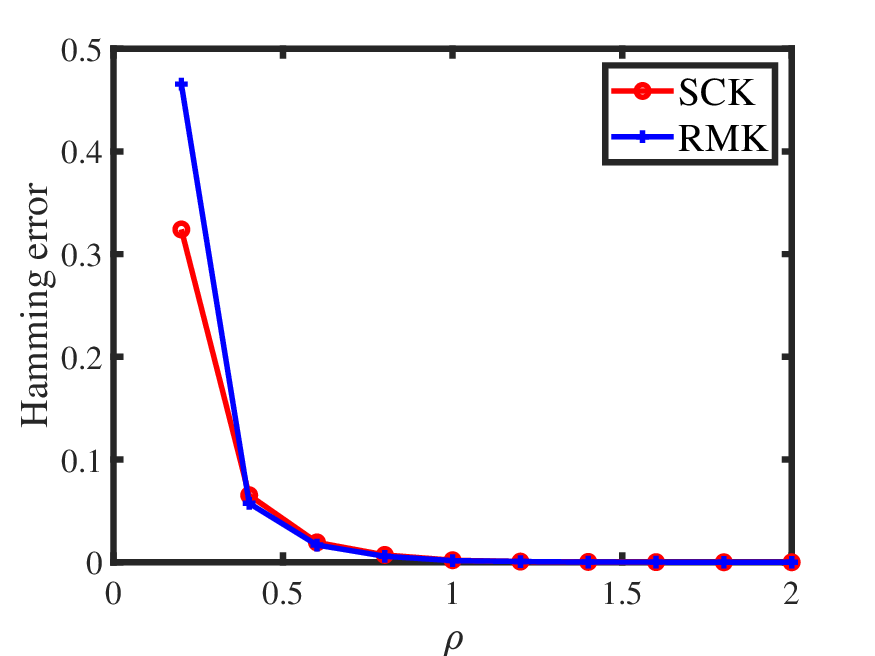}}
\subfigure[Simulation 3(a)]{\includegraphics[width=0.24\textwidth]{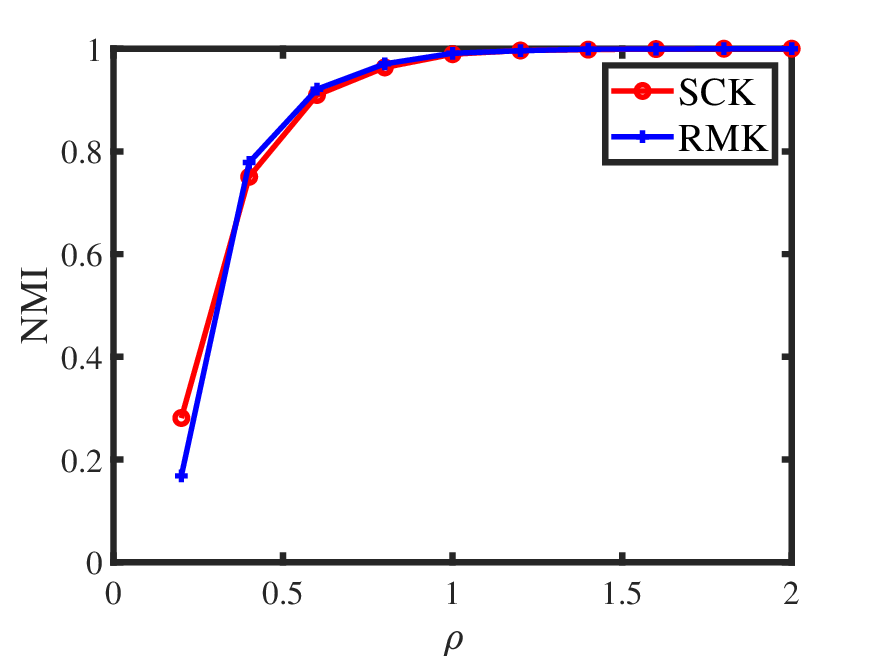}}
\subfigure[Simulation 3(a)]{\includegraphics[width=0.24\textwidth]{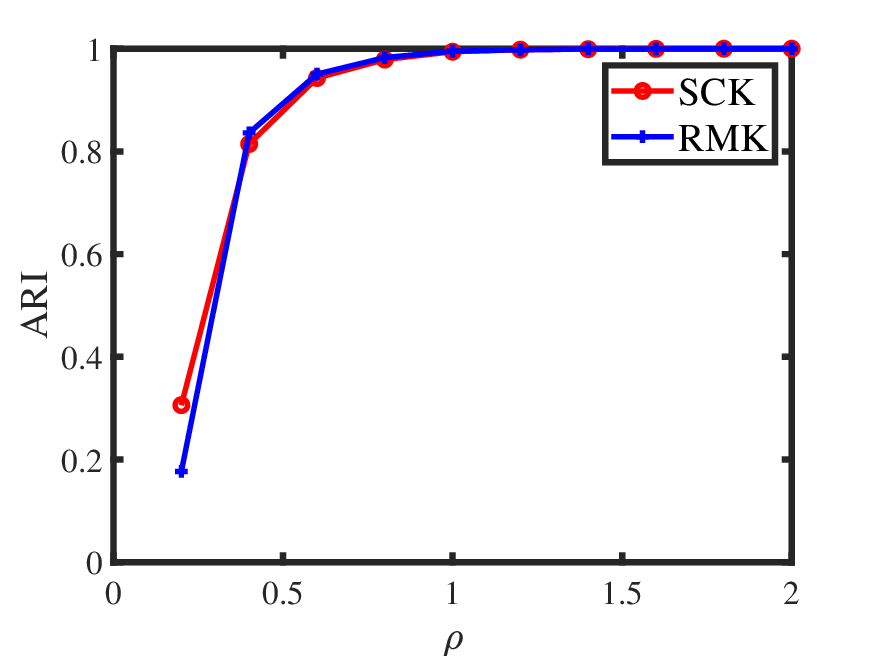}}
\subfigure[Simulation 3(a)]{\includegraphics[width=0.24\textwidth]{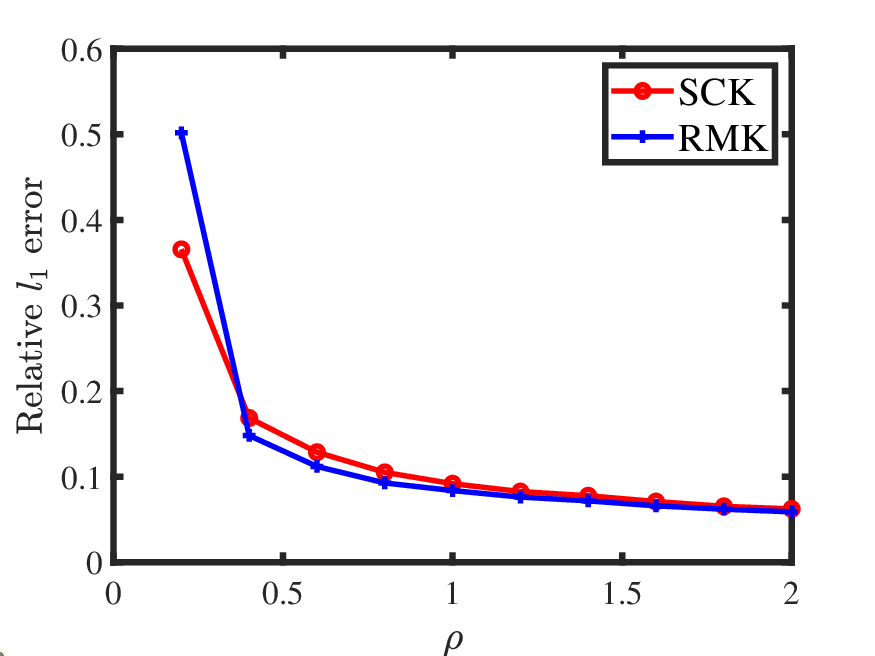}}
\subfigure[Simulation 3(a)]{\includegraphics[width=0.24\textwidth]{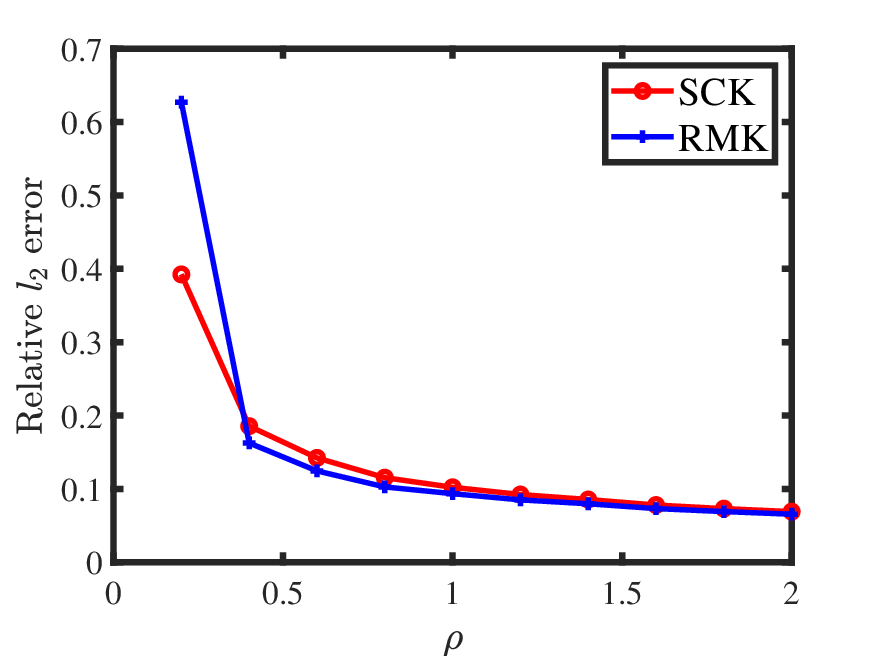}}
\subfigure[Simulation 3(a)]{\includegraphics[width=0.24\textwidth]{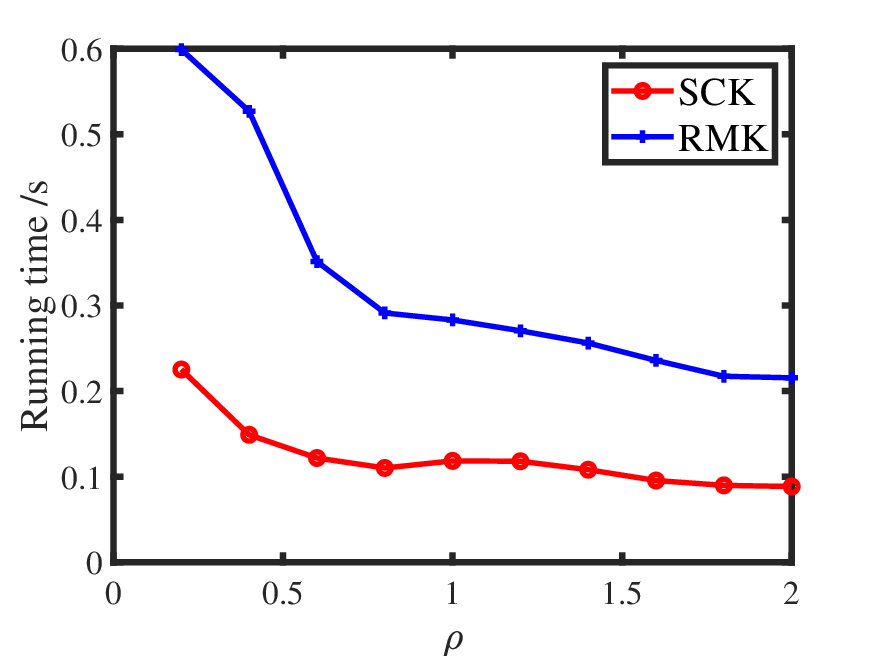}}
\subfigure[Simulation 3(b)]{\includegraphics[width=0.24\textwidth]{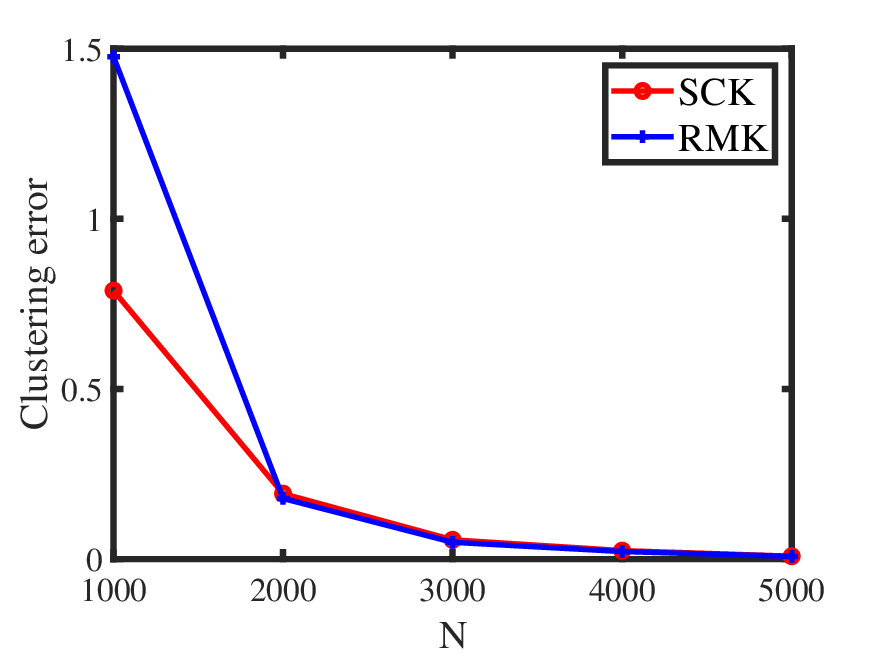}}
\subfigure[Simulation 3(b)]{\includegraphics[width=0.24\textwidth]{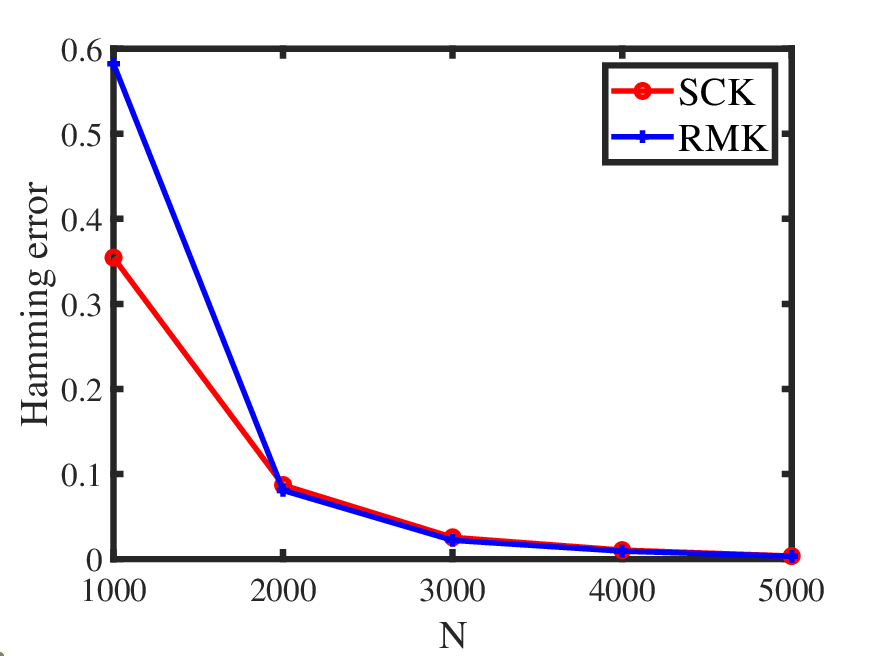}}
\subfigure[Simulation 3(b)]{\includegraphics[width=0.24\textwidth]{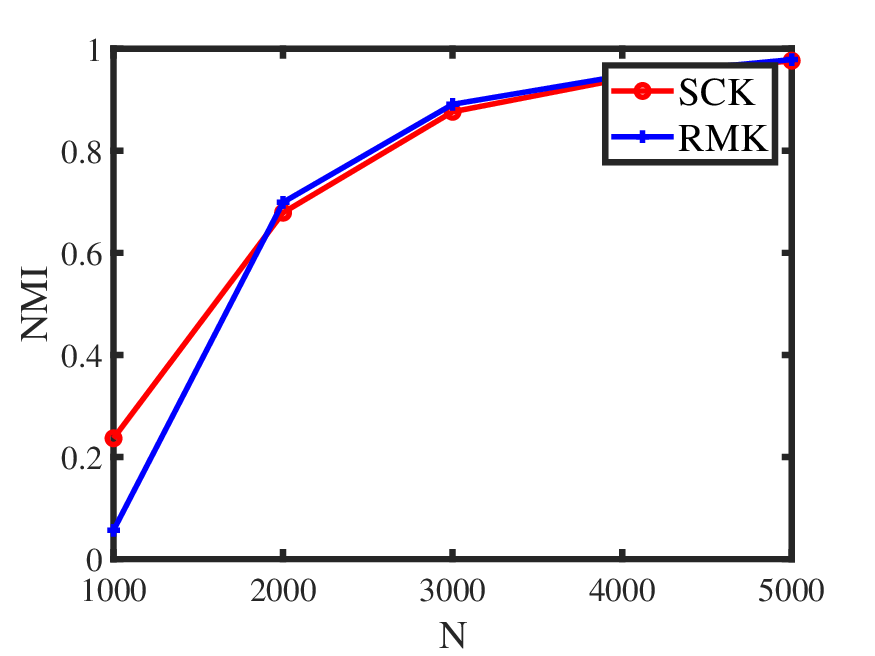}}
\subfigure[Simulation 3(b)]{\includegraphics[width=0.24\textwidth]{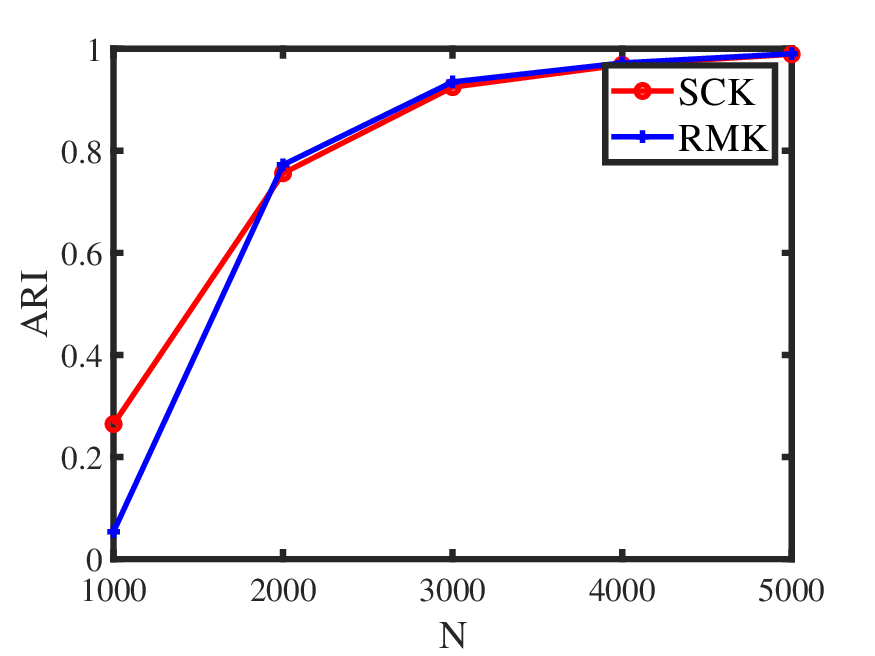}}
\subfigure[Simulation 3(b)]{\includegraphics[width=0.24\textwidth]{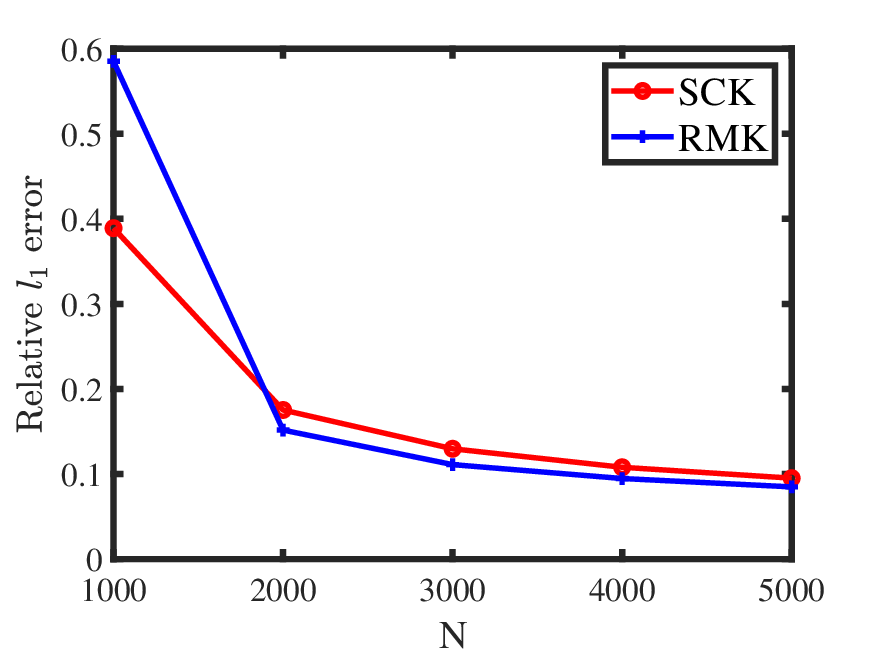}}
\subfigure[Simulation 3(b)]{\includegraphics[width=0.24\textwidth]{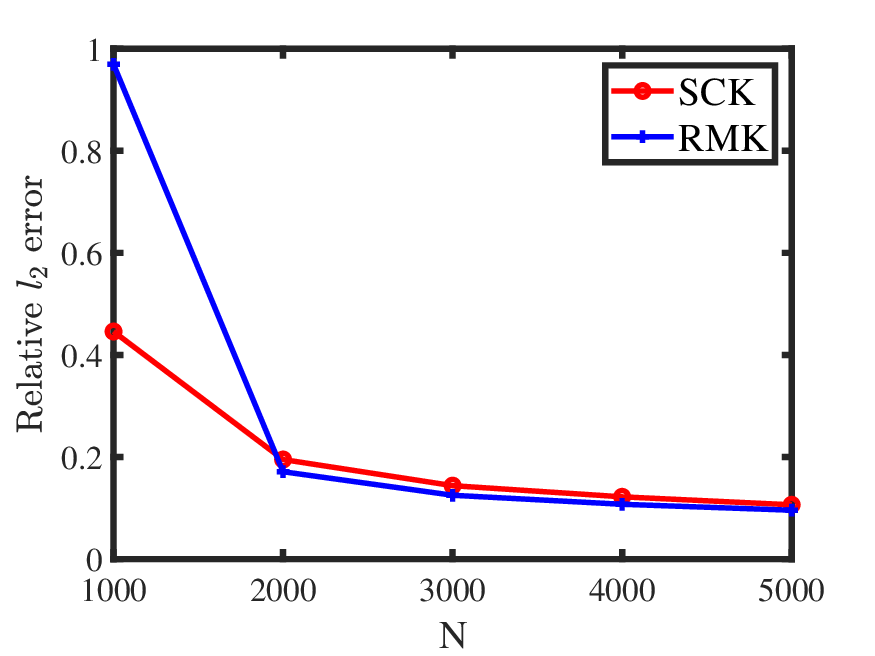}}
\subfigure[Simulation 3(b)]{\includegraphics[width=0.24\textwidth]{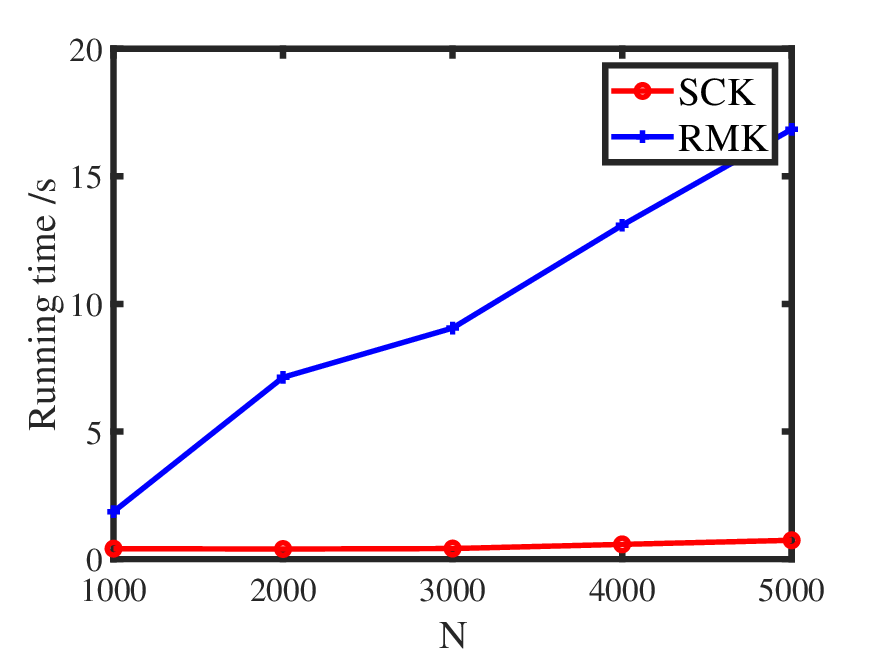}}
\caption{Numerical results of Simulation 3.}
\label{S3} 
\end{figure}
\subsubsection{Normal distribution}
When $R(i,j)\sim \mathrm{Normal}(R_{0}(i,j),\sigma^{2})$ for $i\in[N], j\in[J]$, we consider the following two simulations.

\textbf{Simulation 4(a): changing $\rho$.} Set $N=500$ and $\sigma^{2}=2$. According to Example \ref{Normal}, the scaling parameter $\rho$ can be set as any positive value when $\mathcal{F}$ is Normal distribution. Here, we let $\rho$ range in $\{0.2,0.4,0.6,\ldots,2\}$.

\textbf{Simulation 4(b): changing $N$.} Let $\rho=0.5, \sigma^{2}=2$, and $N$ range in $\{1000,2000,\ldots,5000\}$.

Figure \ref{S4} shows the results. We see that SCK and RMK have similar performances in estimating model parameters $(Z,\Theta)$ while SCK runs faster than RMK. Additionally, the error rates of both approaches decrease when the scaling parameter $\rho$ and the number of subjects $N$ increase, supporting our findings in Example \ref{Normal} and Corollary \ref{AddConditions}.
\begin{figure}
\centering
\subfigure[Simulation 4(a)]{\includegraphics[width=0.24\textwidth]{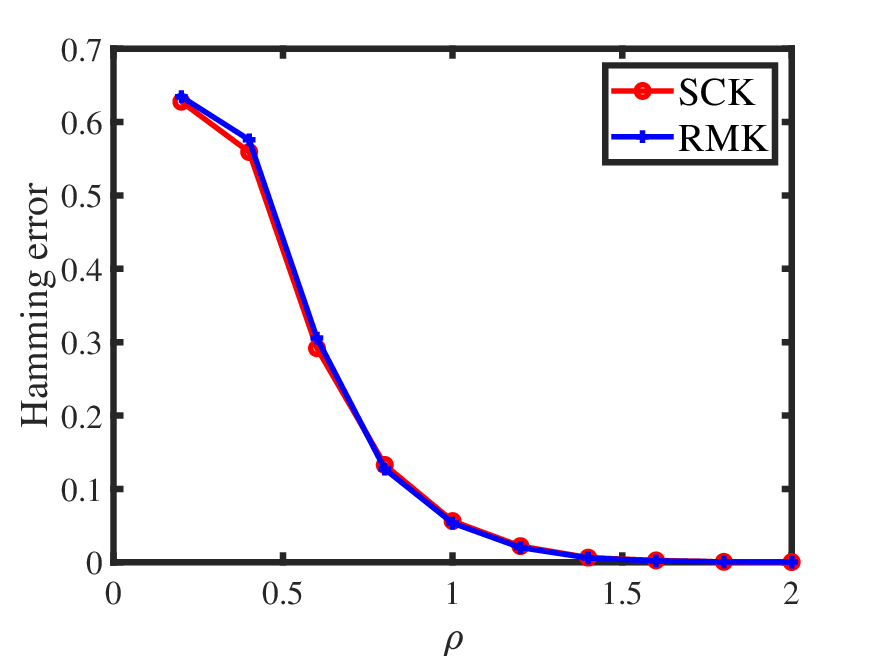}}
\subfigure[Simulation 4(a)]{\includegraphics[width=0.24\textwidth]{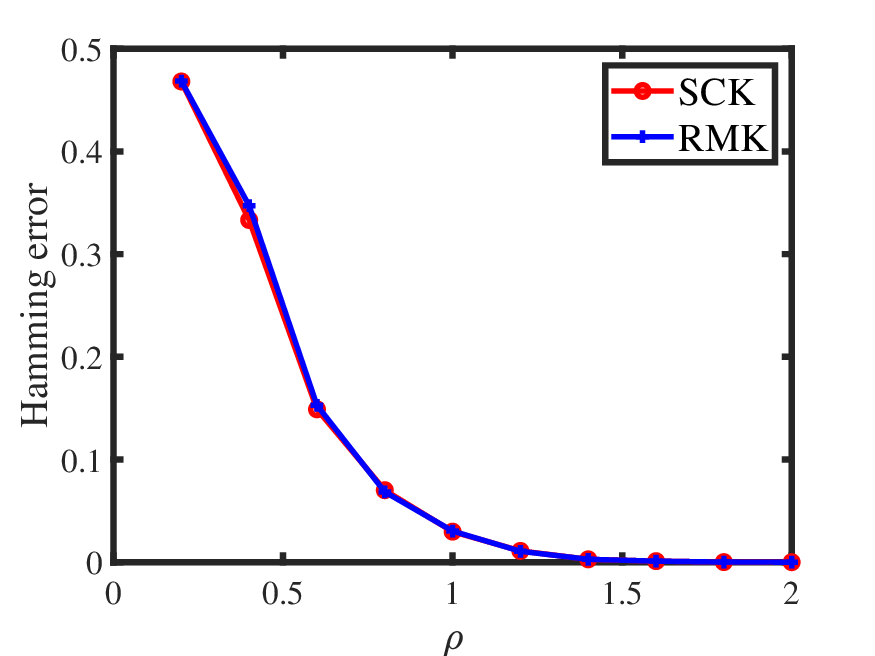}}
\subfigure[Simulation 4(a)]{\includegraphics[width=0.24\textwidth]{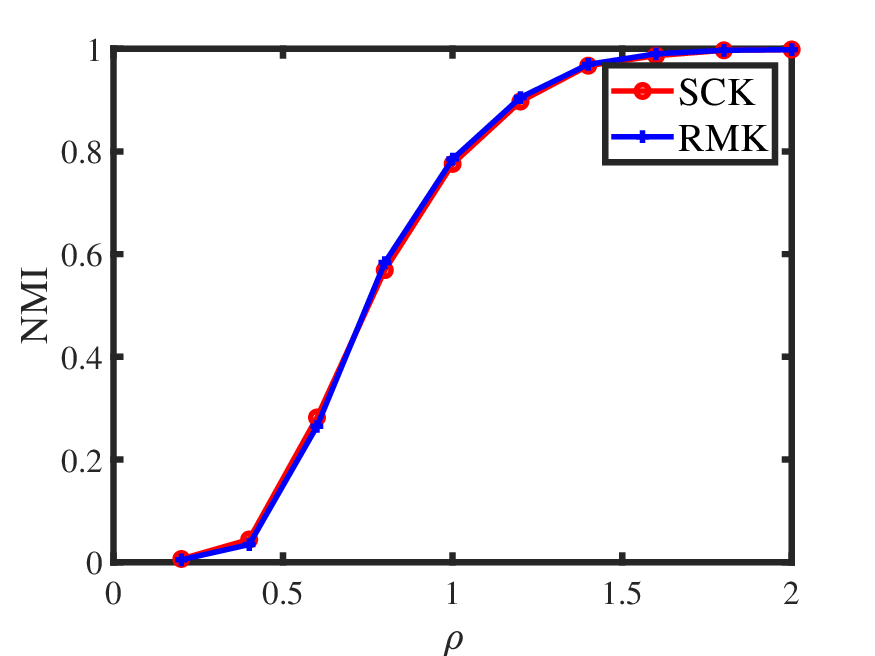}}
\subfigure[Simulation 4(a)]{\includegraphics[width=0.24\textwidth]{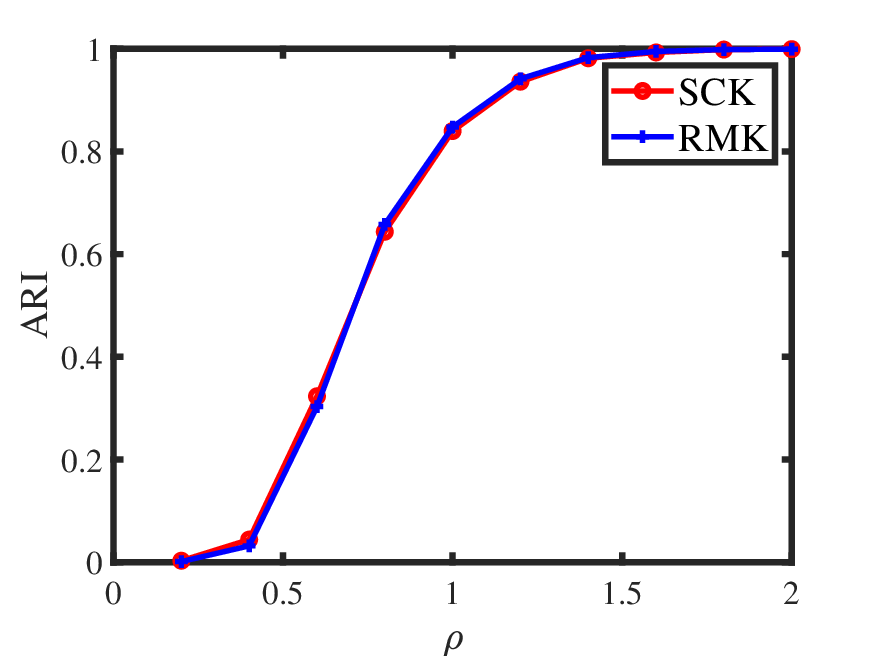}}
\subfigure[Simulation 4(a)]{\includegraphics[width=0.24\textwidth]{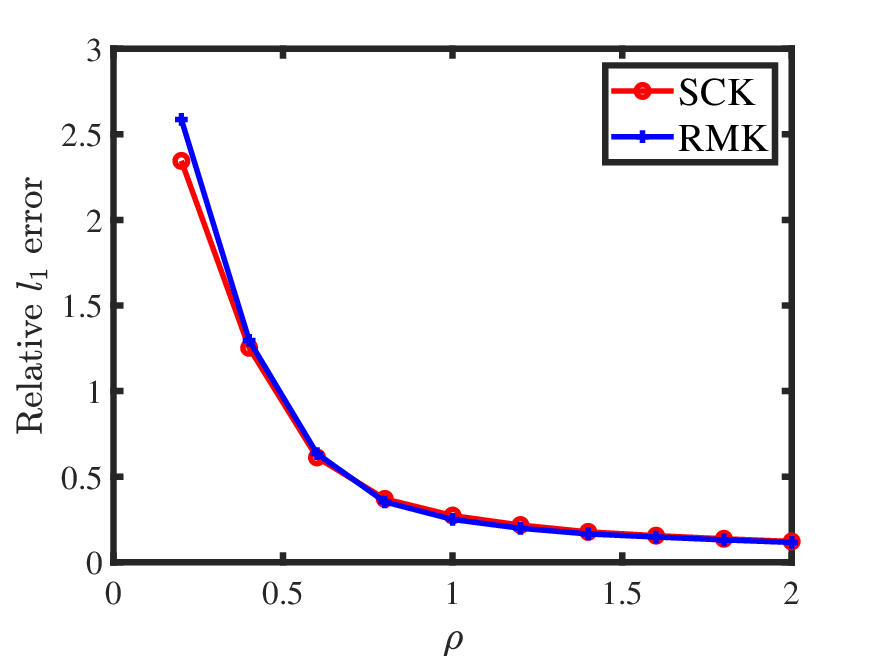}}
\subfigure[Simulation 4(a)]{\includegraphics[width=0.24\textwidth]{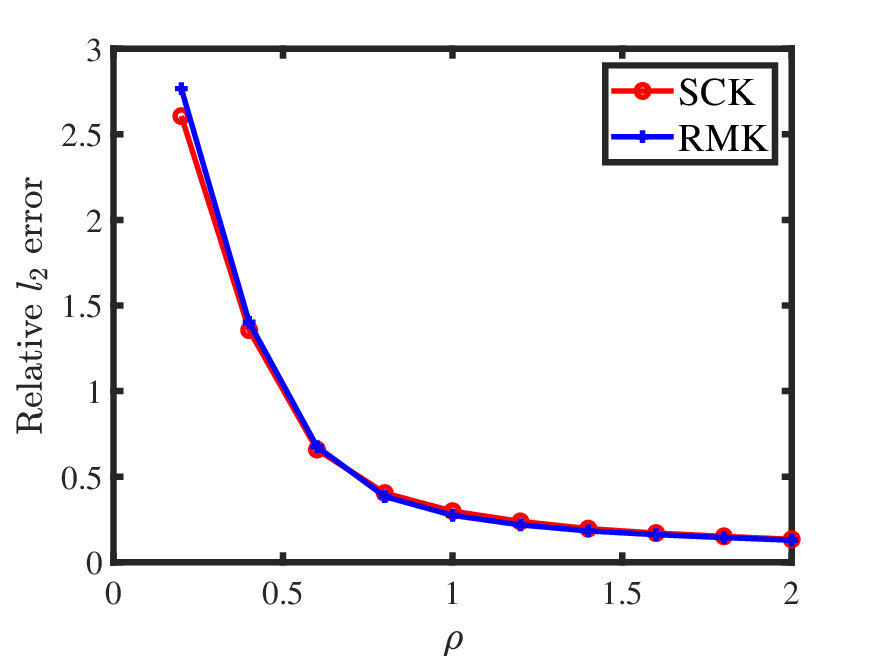}}
\subfigure[Simulation 4(a)]{\includegraphics[width=0.24\textwidth]{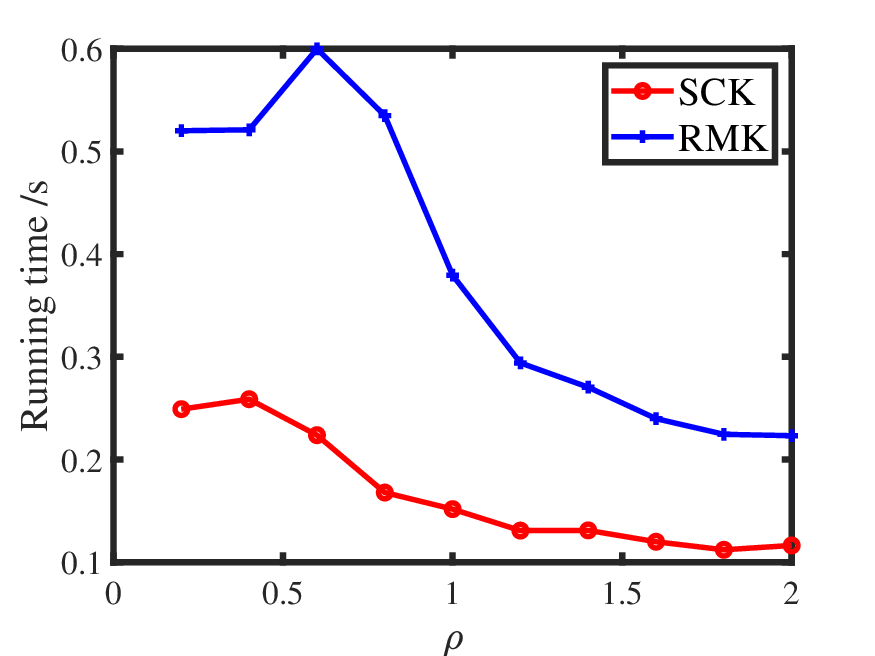}}
\subfigure[Simulation 4(b)]{\includegraphics[width=0.24\textwidth]{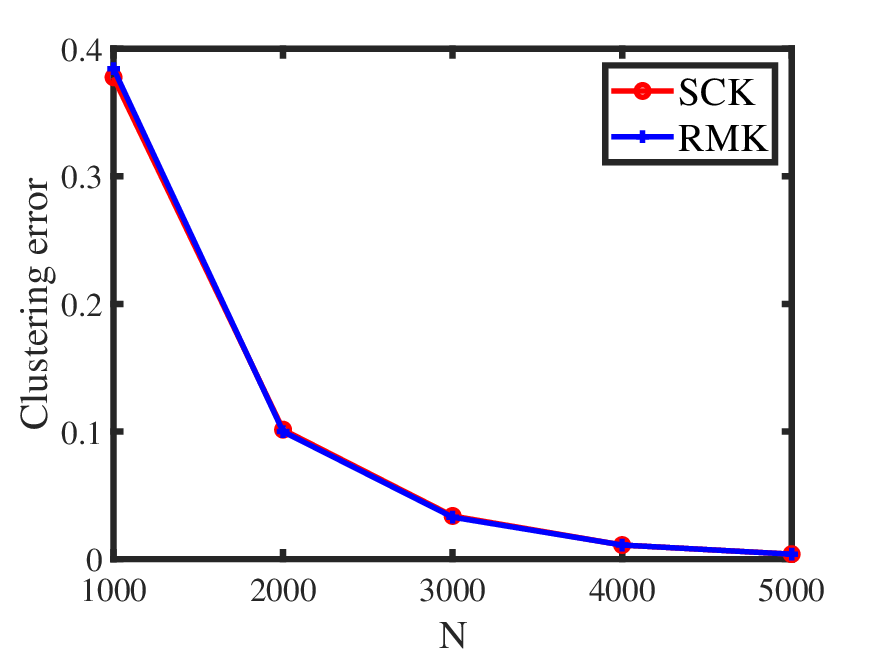}}
\subfigure[Simulation 4(b)]{\includegraphics[width=0.24\textwidth]{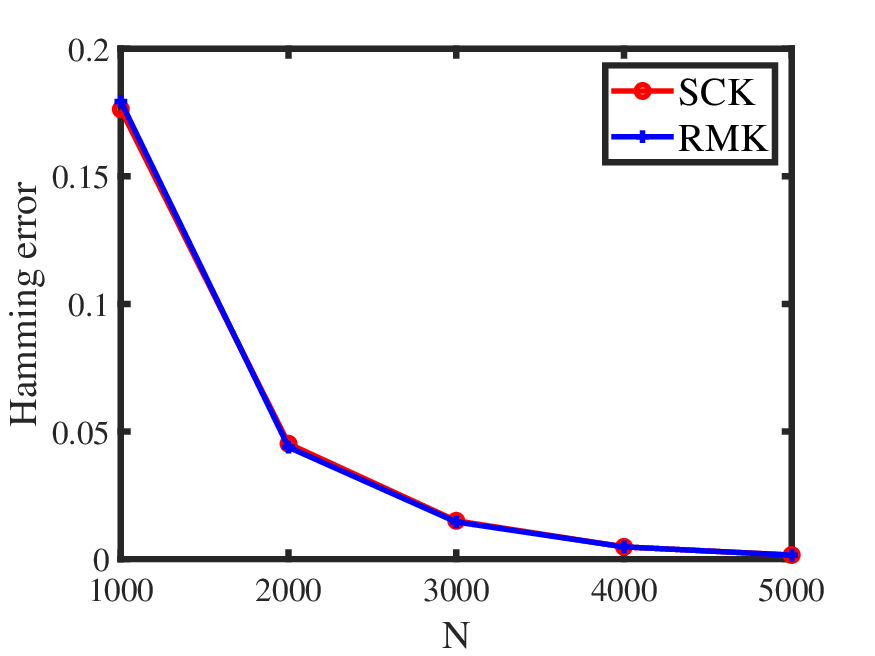}}
\subfigure[Simulation 4(b)]{\includegraphics[width=0.24\textwidth]{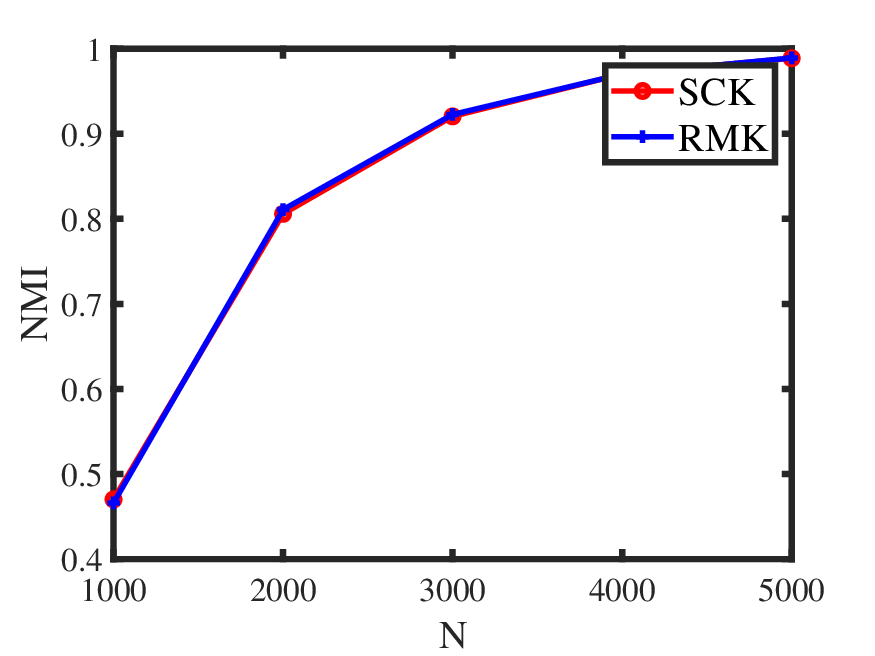}}
\subfigure[Simulation 4(b)]{\includegraphics[width=0.24\textwidth]{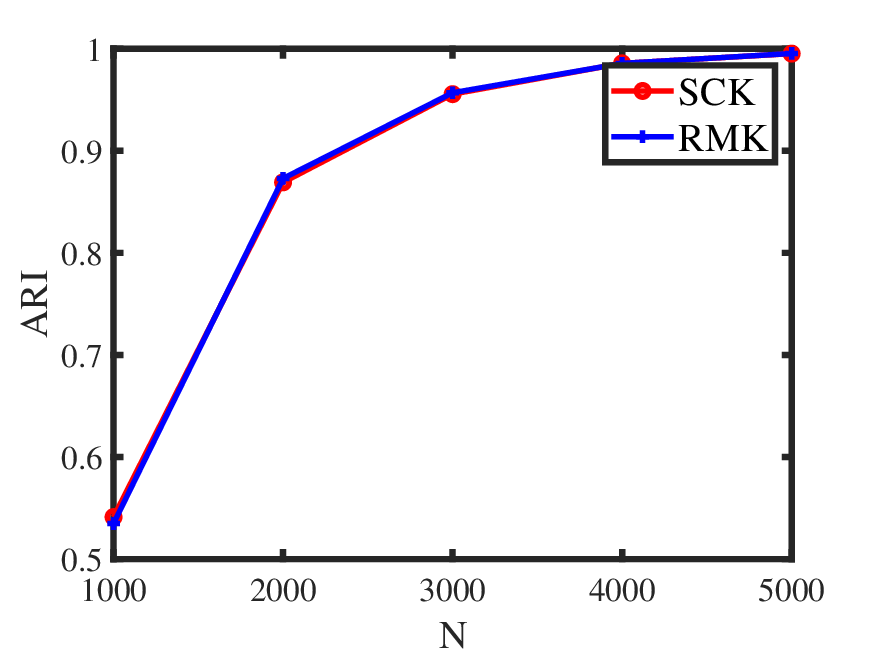}}
\subfigure[Simulation 4(b)]{\includegraphics[width=0.24\textwidth]{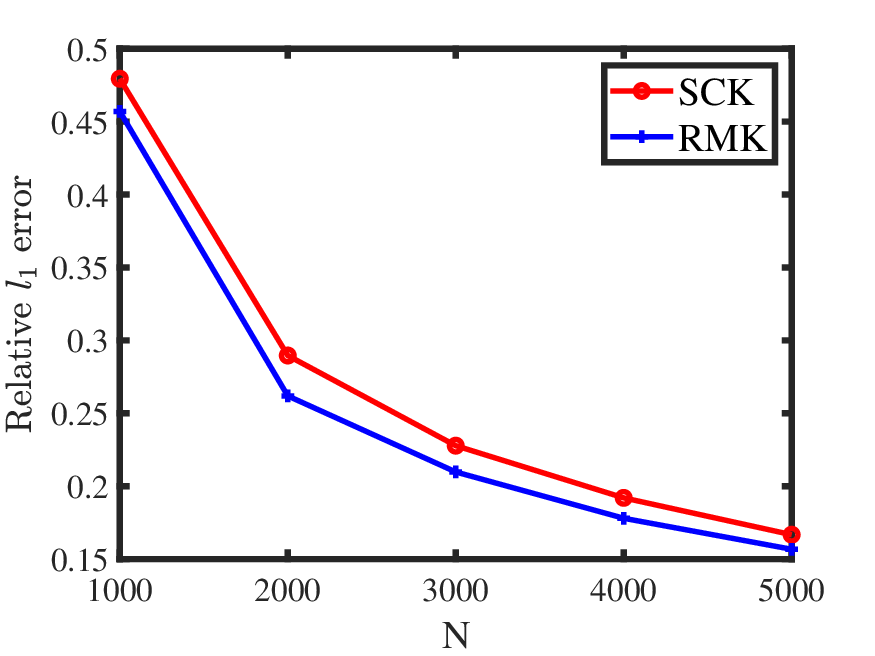}}
\subfigure[Simulation 4(b)]{\includegraphics[width=0.24\textwidth]{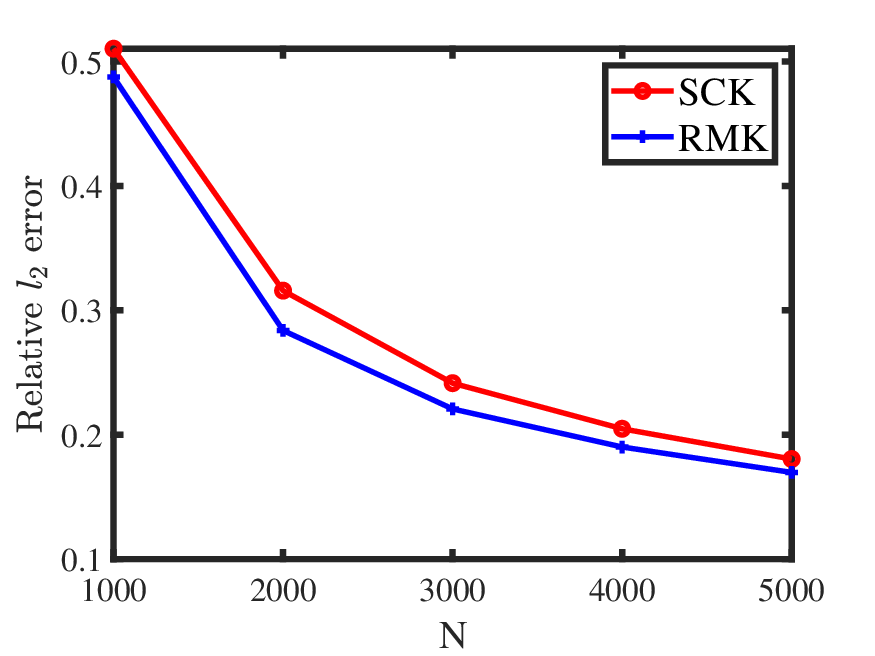}}
\subfigure[Simulation 4(b)]{\includegraphics[width=0.24\textwidth]{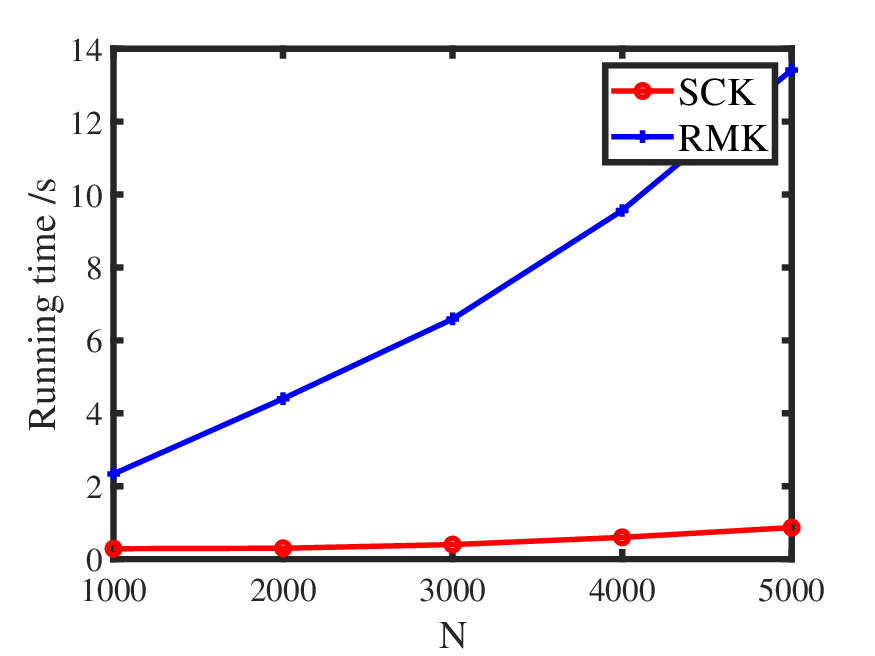}}
\caption{Numerical results of Simulation 4.}
\label{S4} 
\end{figure}
\subsubsection{Exponential distribution}
When $R(i,j)\sim \mathrm{Exponential}(\frac{1}{R_{0}(i,j)})$ for $i\in[N], j\in[J]$, we consider the following two simulations.

\textbf{Simulation 5(a): changing $\rho$.} Set $N=300$. According to Example \ref{Exponential}, the range of the scaling parameter $\rho$ is $(0,+\infty)$ when $\mathcal{F}$ is Exponential distribution. Here, we let $\rho$ range in $\{1,2,\ldots,20\}$ for our numerical studies.

\textbf{Simulation 5(b): changing $N$.} Let $\rho=1$ and $N$ range in $\{300,600,\ldots,3000\}$.

Figure \ref{S5} displays the results. We see that both methods provide satisfactory estimations for $Z$ and $\Theta$ for their small error rates, large NMI, and large ARI. SCK provides more accurate estimations than RMK and SCK takes less time for estimations than RMK. Meanwhile, we find that increasing $\rho$ does not significantly influence the performances of SCK and RMK and this verifies our theoretical analysis in Example \ref{Exponential} that $\rho$ disappears in the theoretical upper bounds of error rates by setting $\gamma=\rho^{2}$ in Theorem \ref{mainWLCM} for Exponential distribution. Furthermore, when we increase $N$, both methods perform better and this supports our analysis after Corollary \ref{AddConditions}.
\begin{figure}
\centering
\subfigure[Simulation 5(a)]{\includegraphics[width=0.24\textwidth]{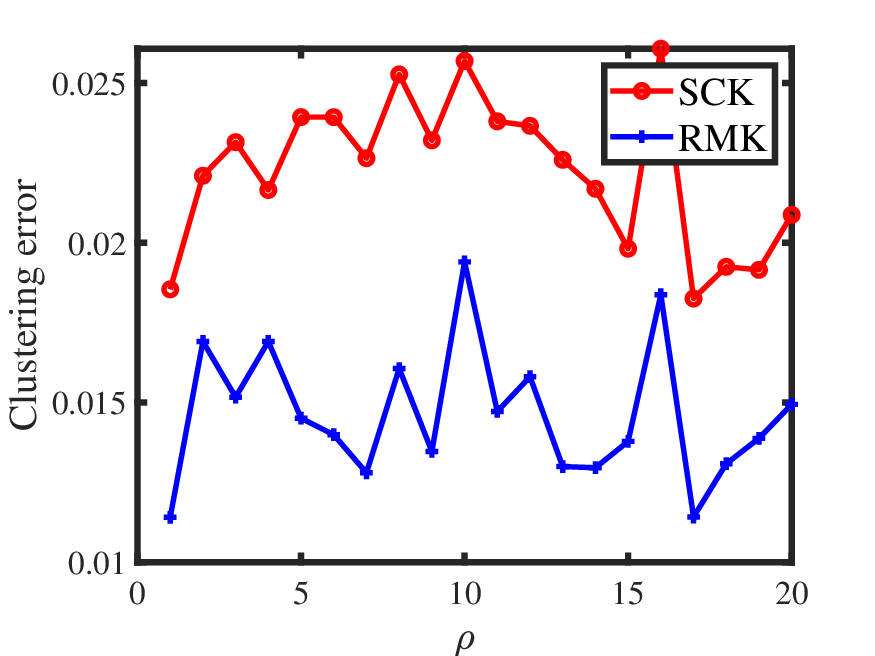}}
\subfigure[Simulation 5(a)]{\includegraphics[width=0.24\textwidth]{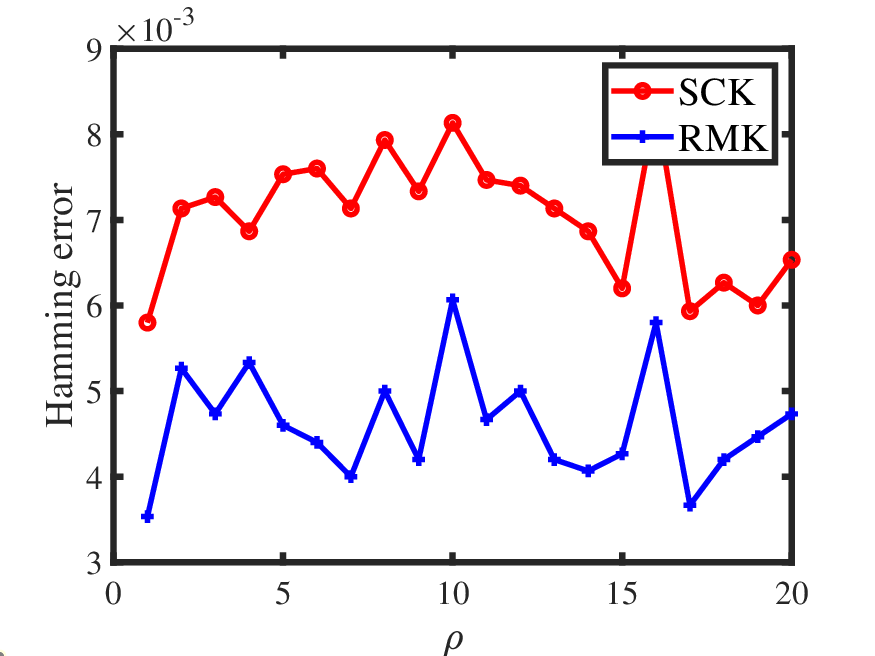}}
\subfigure[Simulation 5(a)]{\includegraphics[width=0.24\textwidth]{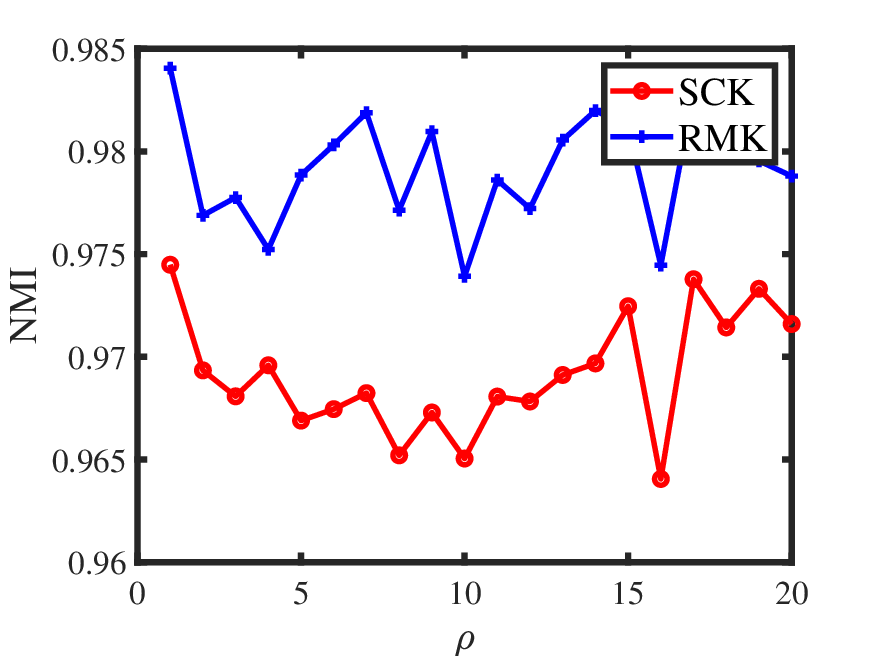}}
\subfigure[Simulation 5(a)]{\includegraphics[width=0.24\textwidth]{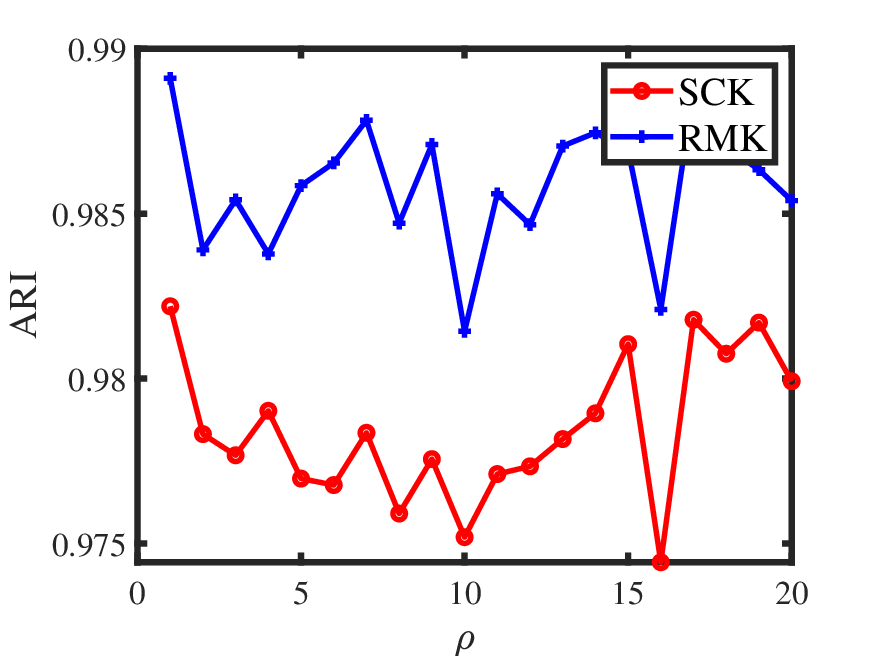}}
\subfigure[Simulation 5(a)]{\includegraphics[width=0.24\textwidth]{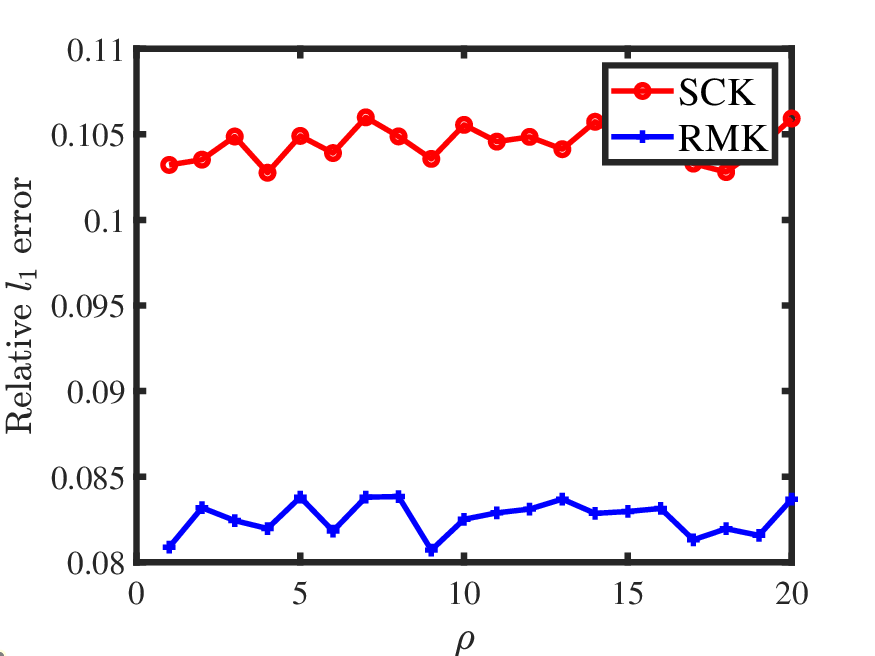}}
\subfigure[Simulation 5(a)]{\includegraphics[width=0.24\textwidth]{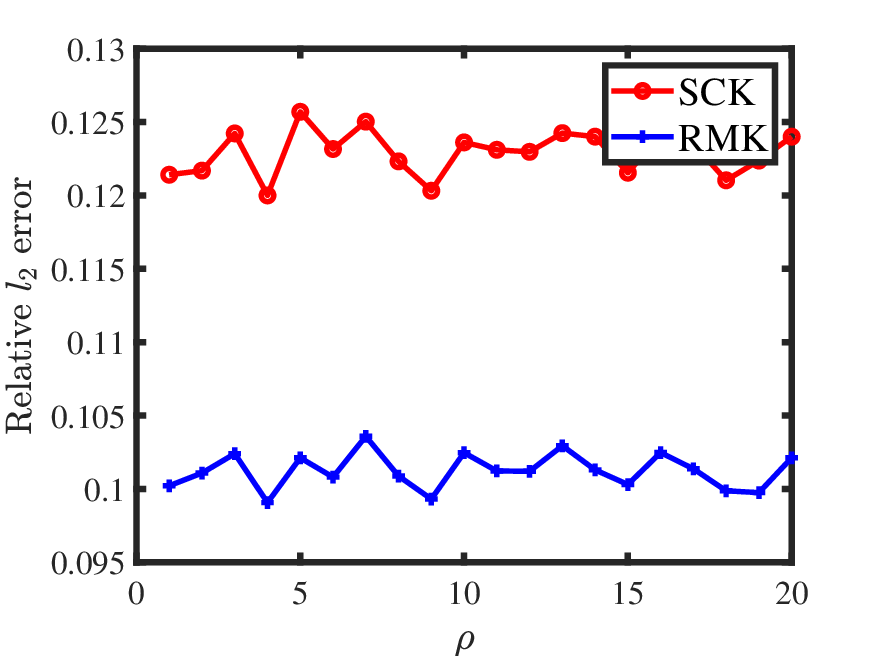}}
\subfigure[Simulation 5(a)]{\includegraphics[width=0.24\textwidth]{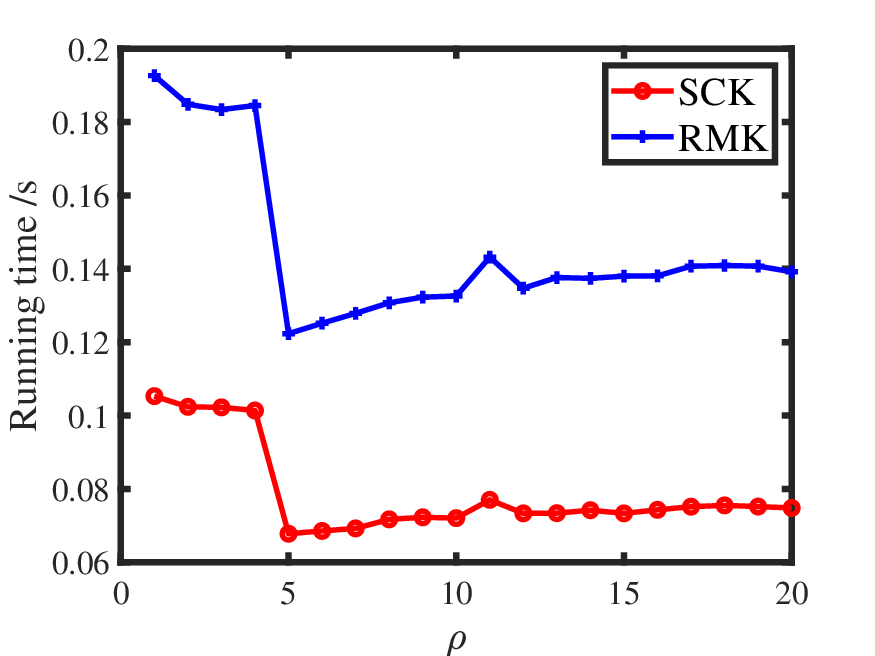}}
\subfigure[Simulation 5(b)]{\includegraphics[width=0.24\textwidth]{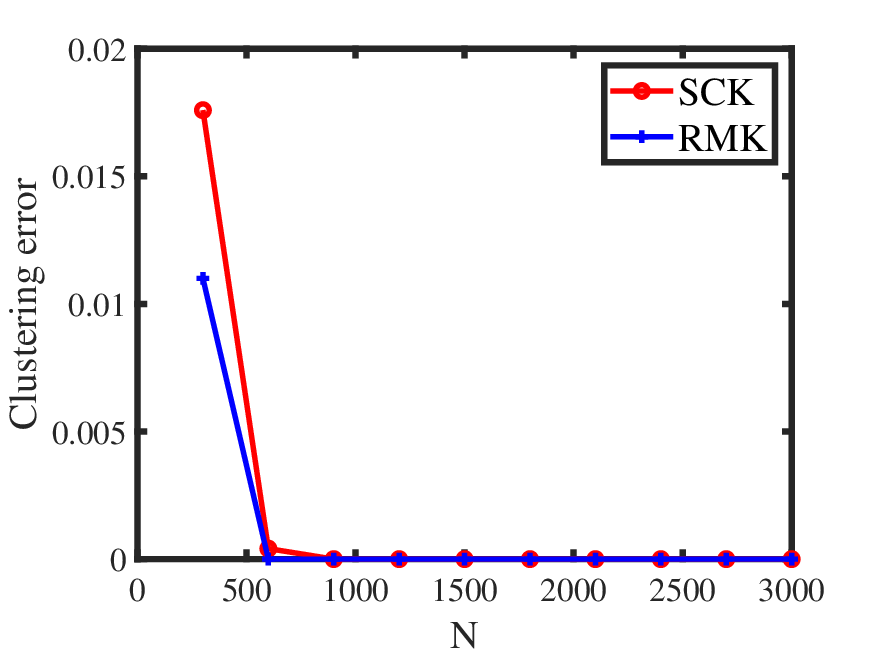}}
\subfigure[Simulation 5(b)]{\includegraphics[width=0.24\textwidth]{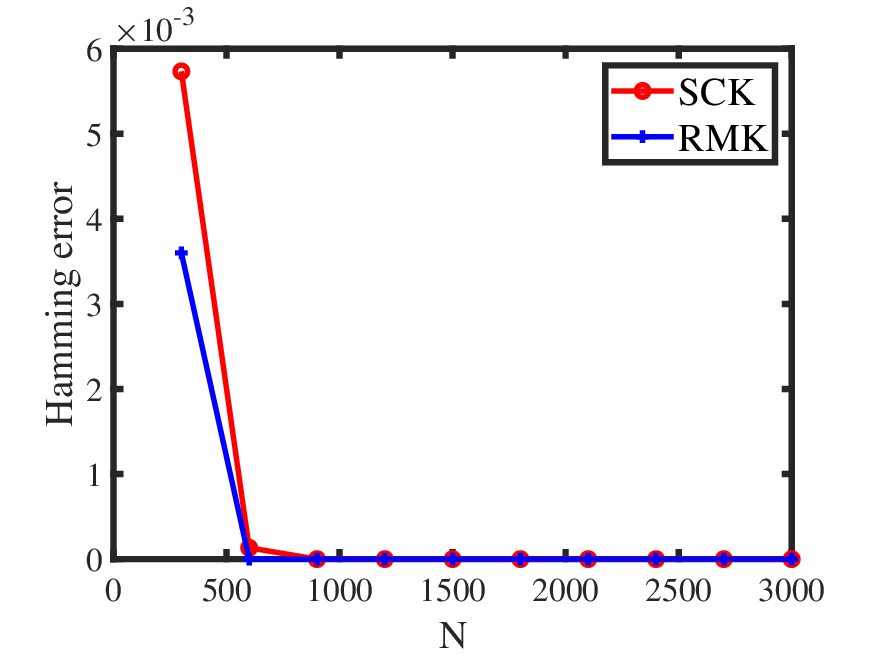}}
\subfigure[Simulation 5(b)]{\includegraphics[width=0.24\textwidth]{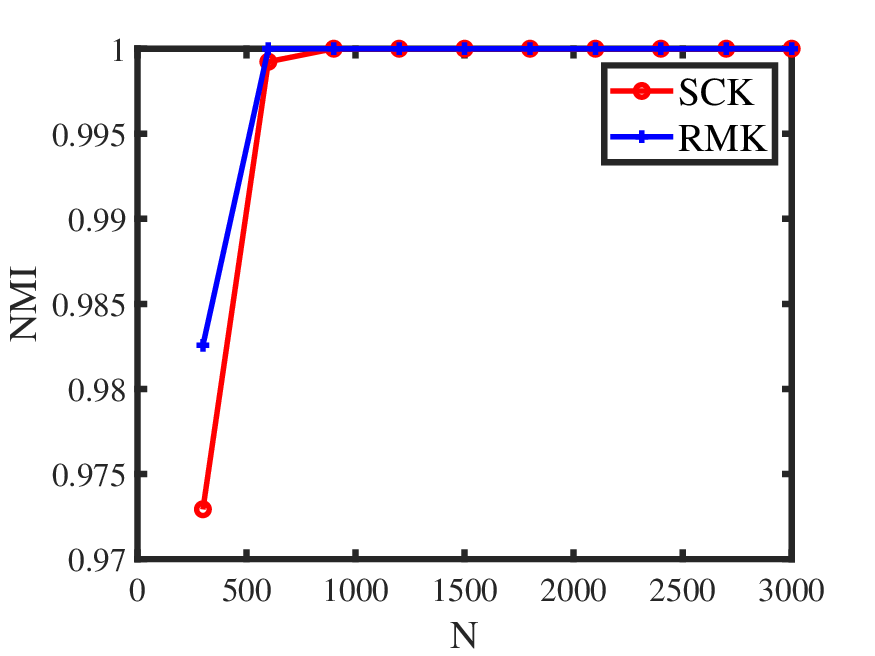}}
\subfigure[Simulation 5(b)]{\includegraphics[width=0.24\textwidth]{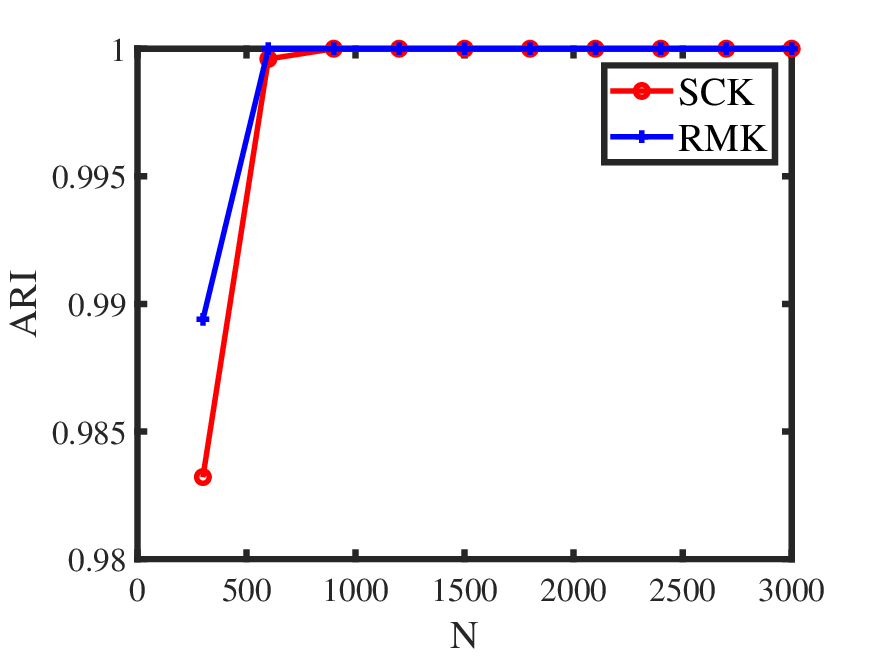}}
\subfigure[Simulation 5(b)]{\includegraphics[width=0.24\textwidth]{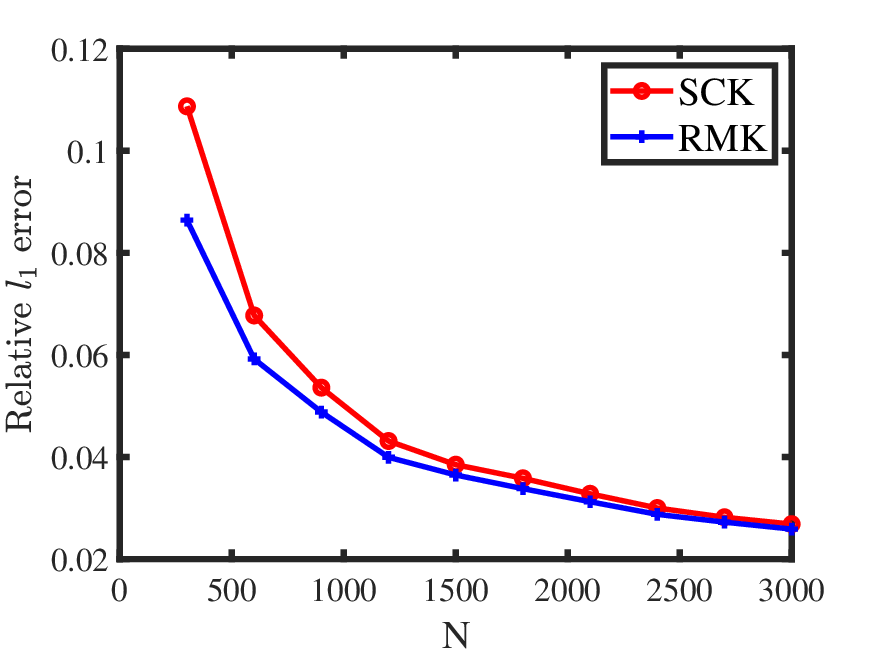}}
\subfigure[Simulation 5(b)]{\includegraphics[width=0.24\textwidth]{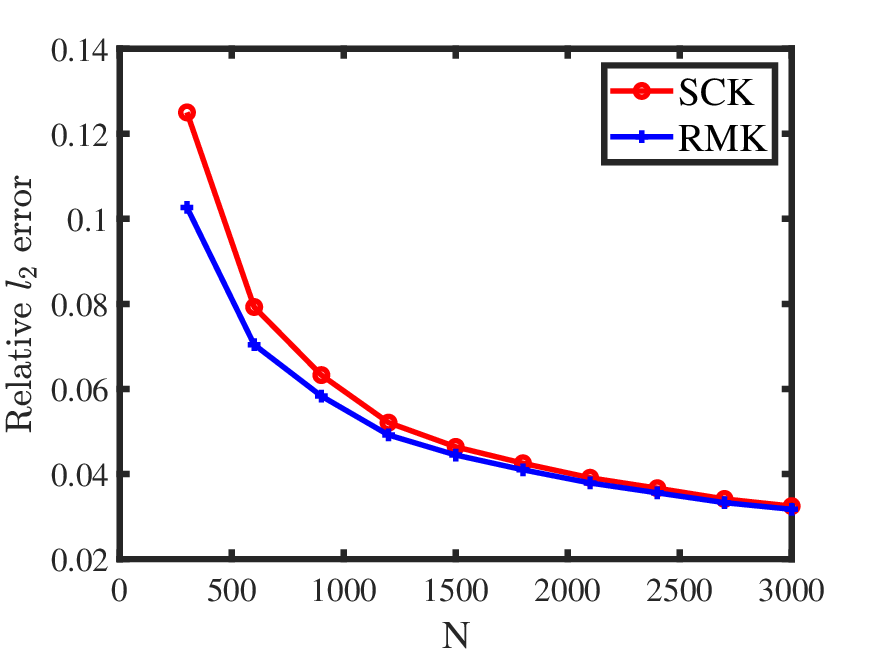}}
\subfigure[Simulation 5(b)]{\includegraphics[width=0.24\textwidth]{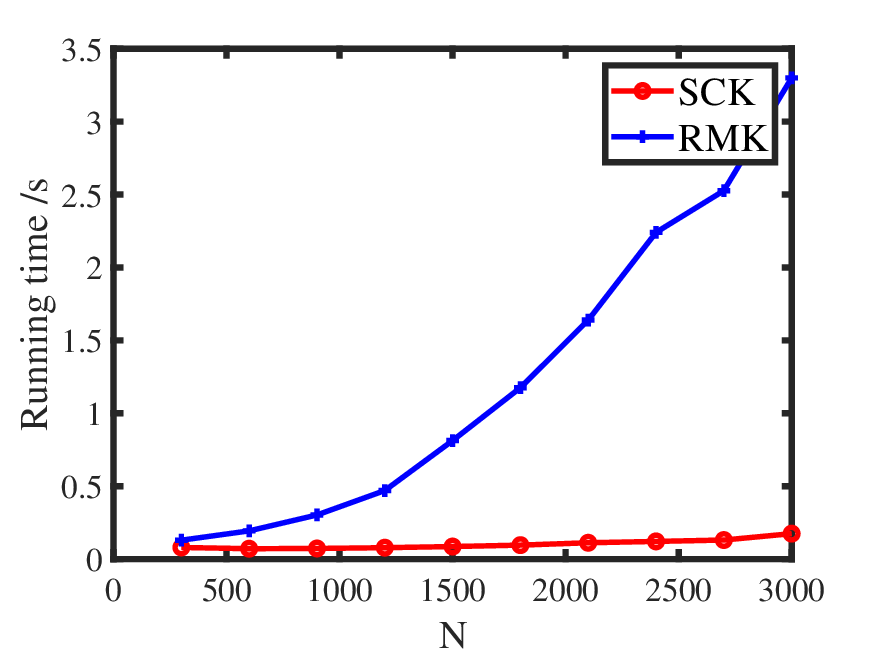}}
\caption{Numerical results of Simulation 5.}
\label{S5} 
\end{figure}
\subsubsection{Uniform distribution}
When $R(i,j)\sim \mathrm{Uniform}(0,2R_{0}(i,j))$ for $i\in[N], j\in[J]$, we consider the following two simulations.

\textbf{Simulation 6(a): changing $\rho$.} Set $N=120$. According to Example \ref{Uniform}, the scaling parameter $\rho$ can be set as any positive value when $\mathcal{F}$ is Uniform distribution. Here, we let $\rho$ range in $\{1,2,\ldots,20\}$.

\textbf{Simulation 6(b): changing $N$.} Let $\rho=1$ and $N$ range in $\{300,600,\ldots,3000\}$.

Figure \ref{S6} displays the numerical results. We see that increasing $\rho$ does not significantly decrease or increase estimation accuracies of SCK and RMK which verifies our theoretical analysis in Example \ref{Uniform}. For all settings, SCK runs faster than RMK. When increasing $N$, the Clustering error and Hamming error (NMI and ARI) for both approaches are 0 (1), and this suggests that SCK and RMK return the exact estimation of the classification matrix $Z$. This phenomenon occurs because $N$ is set quite large for Uniform distribution in Simulation 6(b). For the estimation of $\Theta$, error rates for both methods decrease when we increase $N$ and this is consistent with our findings following Corollary \ref{AddConditions}.
\begin{figure}
\centering
\subfigure[Simulation 6(a)]{\includegraphics[width=0.24\textwidth]{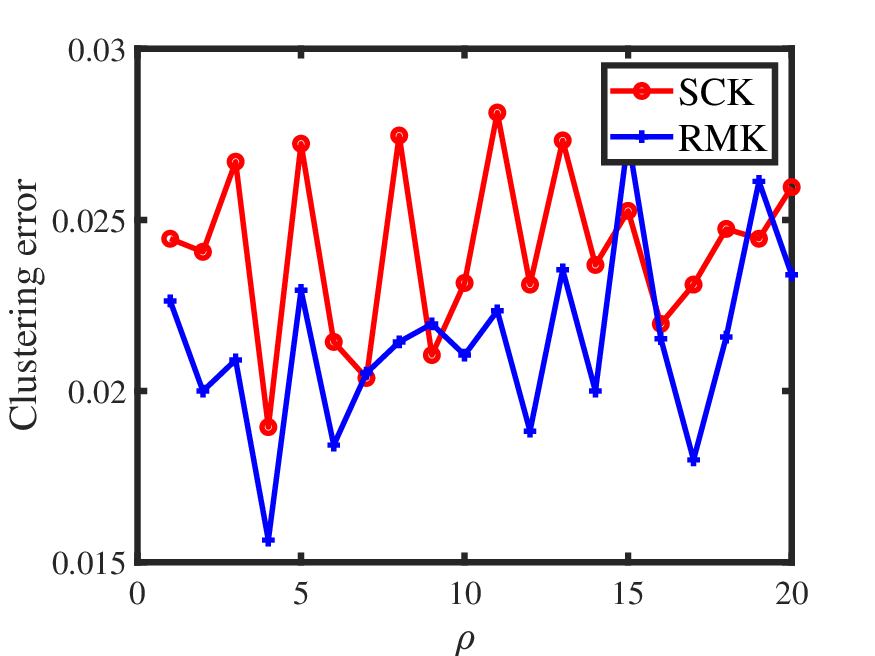}}
\subfigure[Simulation 6(a)]{\includegraphics[width=0.24\textwidth]{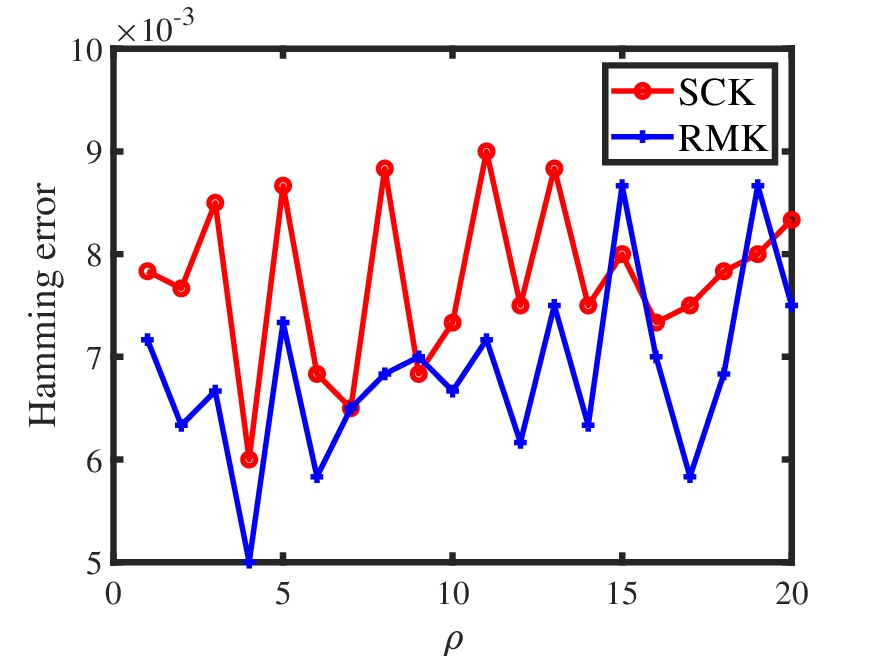}}
\subfigure[Simulation 6(a)]{\includegraphics[width=0.24\textwidth]{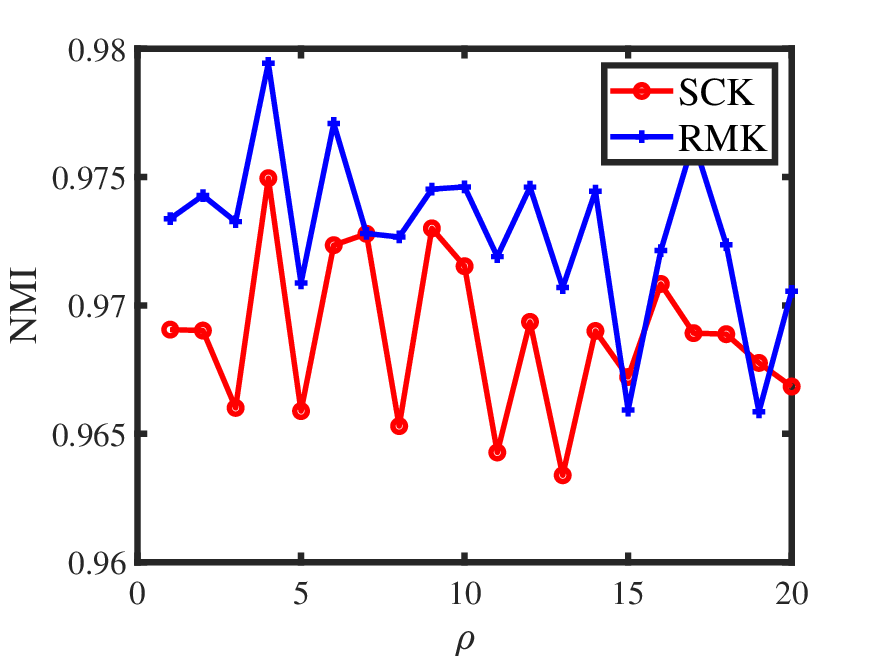}}
\subfigure[Simulation 6(a)]{\includegraphics[width=0.24\textwidth]{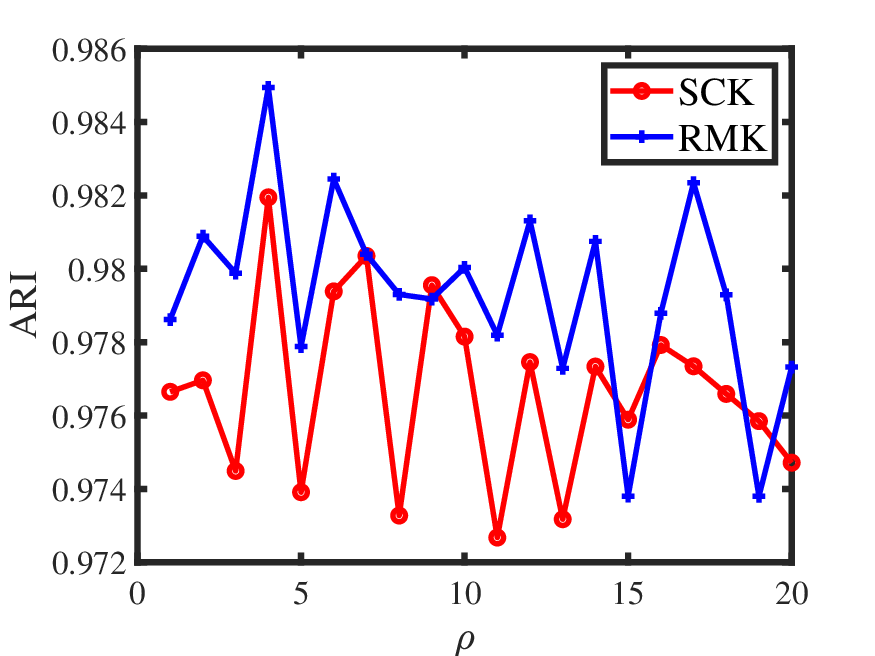}}
\subfigure[Simulation 6(a)]{\includegraphics[width=0.24\textwidth]{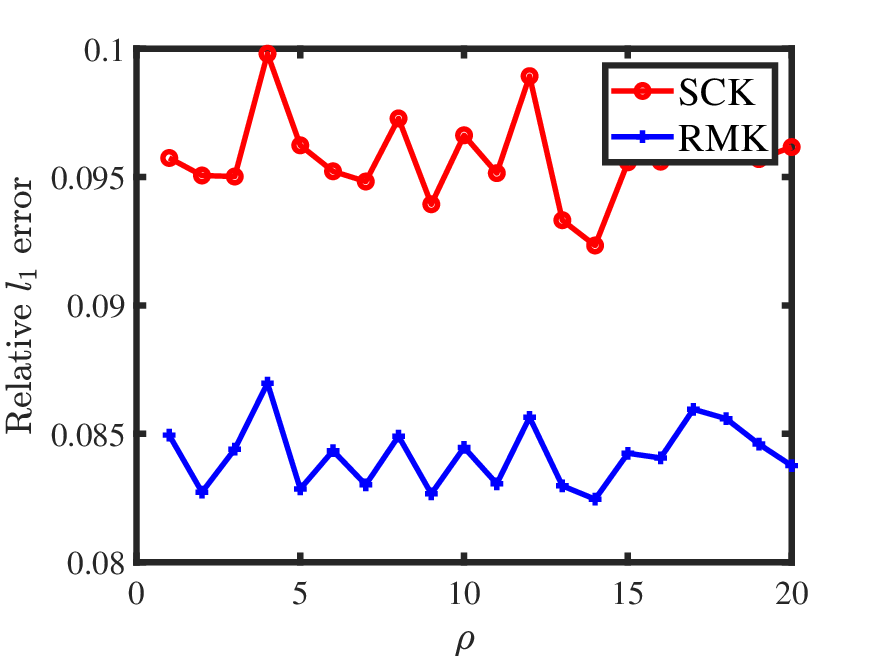}}
\subfigure[Simulation 6(a)]{\includegraphics[width=0.24\textwidth]{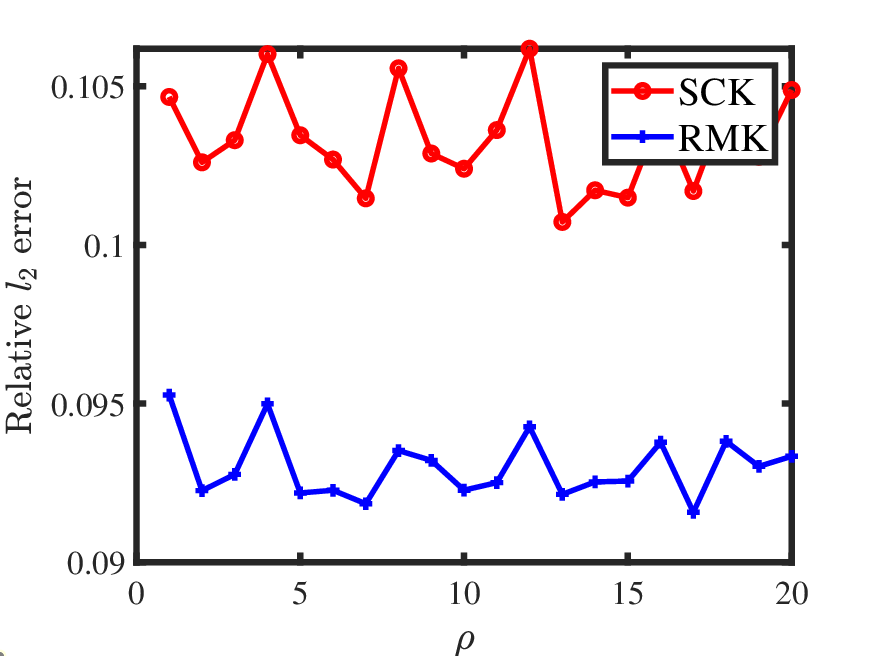}}
\subfigure[Simulation 6(a)]{\includegraphics[width=0.24\textwidth]{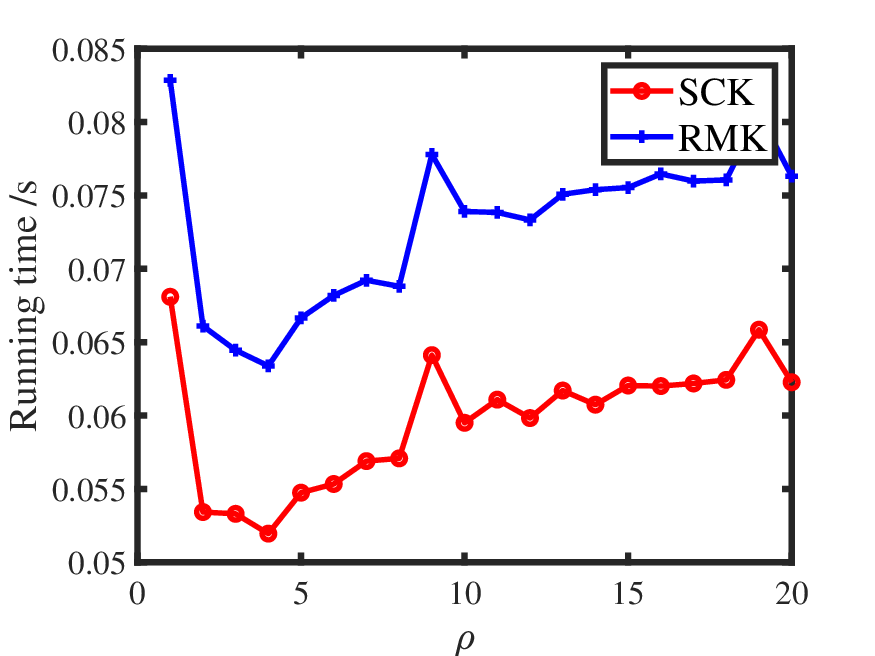}}
\subfigure[Simulation 6(b)]{\includegraphics[width=0.24\textwidth]{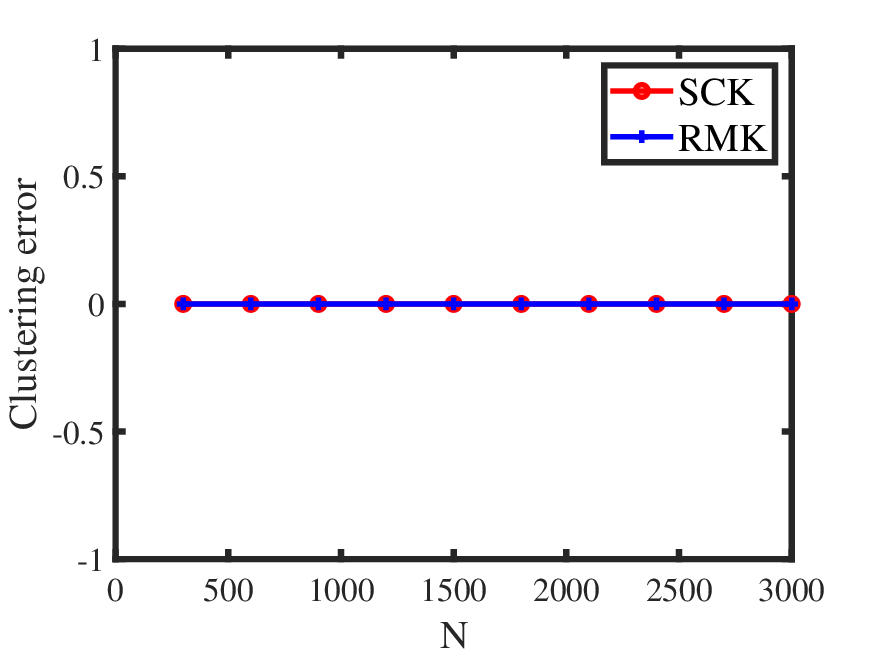}}
\subfigure[Simulation 6(b)]{\includegraphics[width=0.24\textwidth]{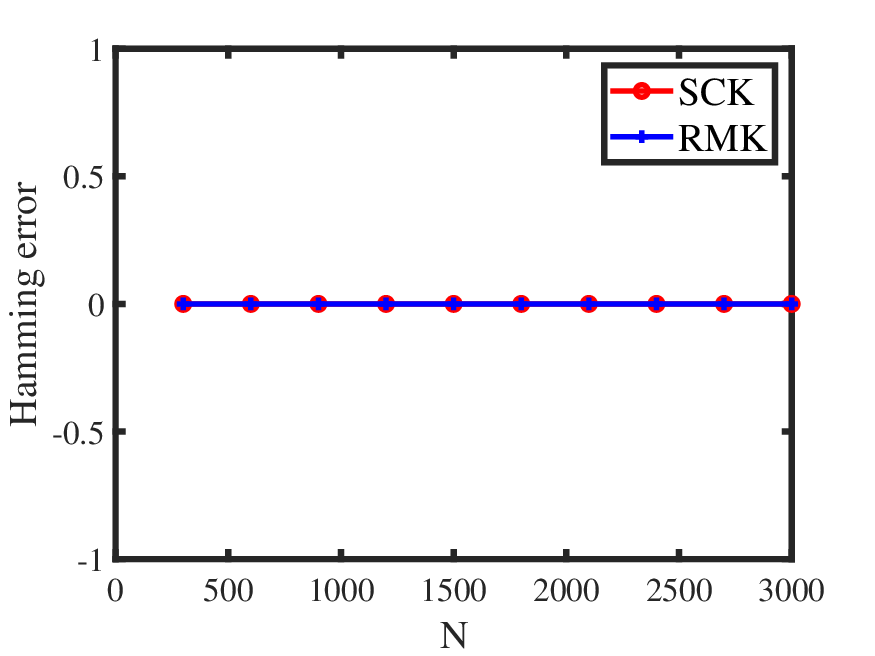}}
\subfigure[Simulation 6(b)]{\includegraphics[width=0.24\textwidth]{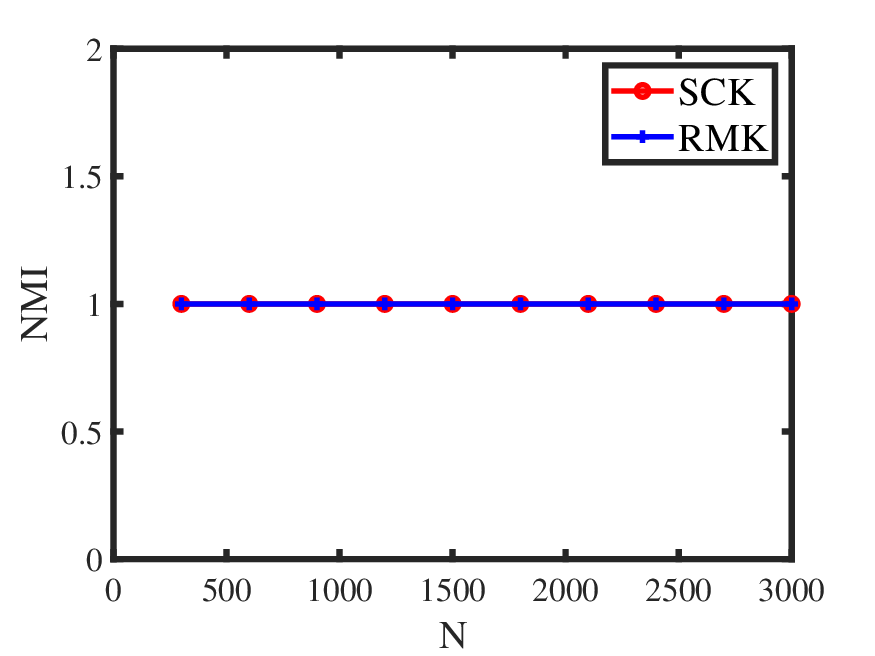}}
\subfigure[Simulation 6(b)]{\includegraphics[width=0.24\textwidth]{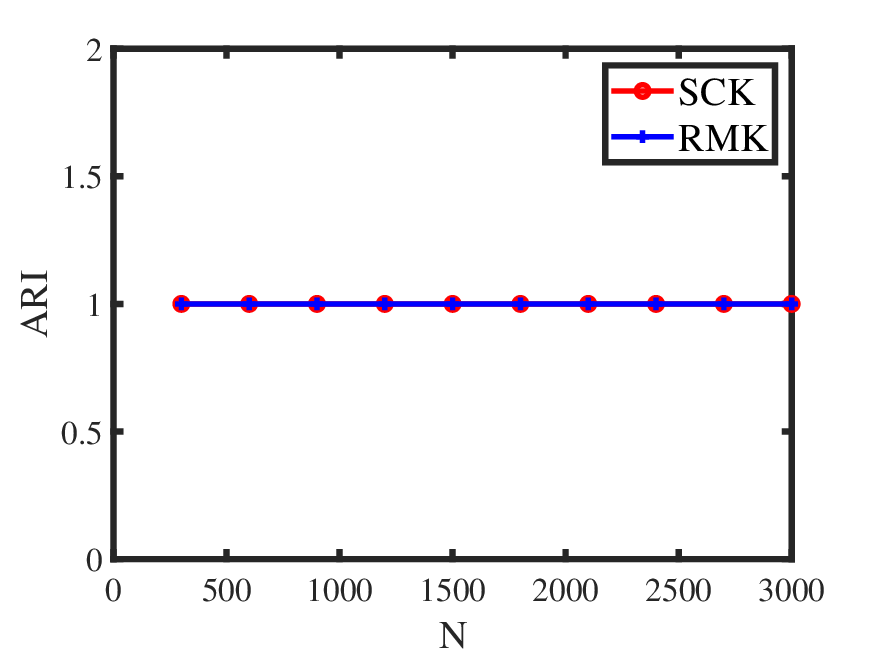}}
\subfigure[Simulation 6(b)]{\includegraphics[width=0.24\textwidth]{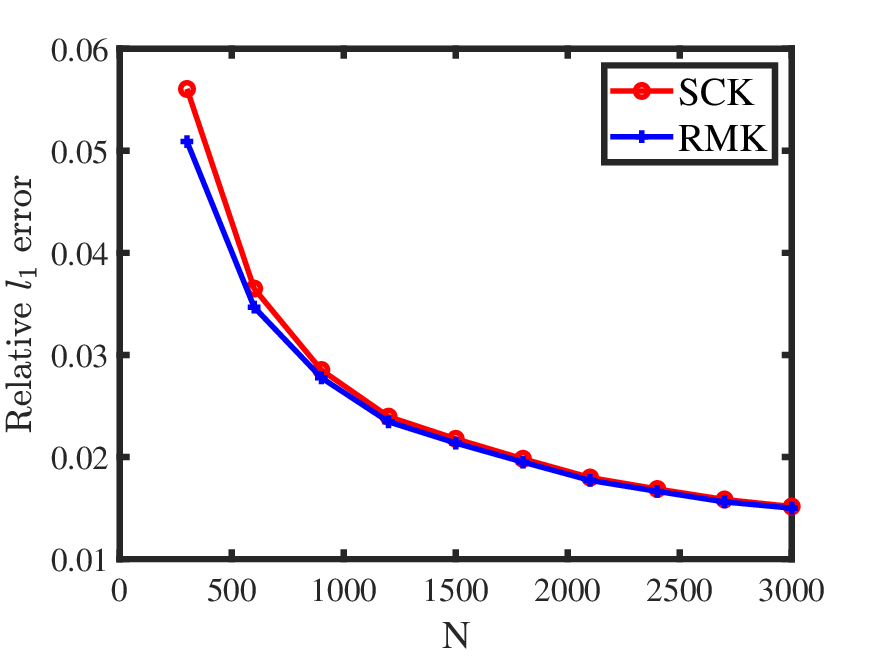}}
\subfigure[Simulation 6(b)]{\includegraphics[width=0.24\textwidth]{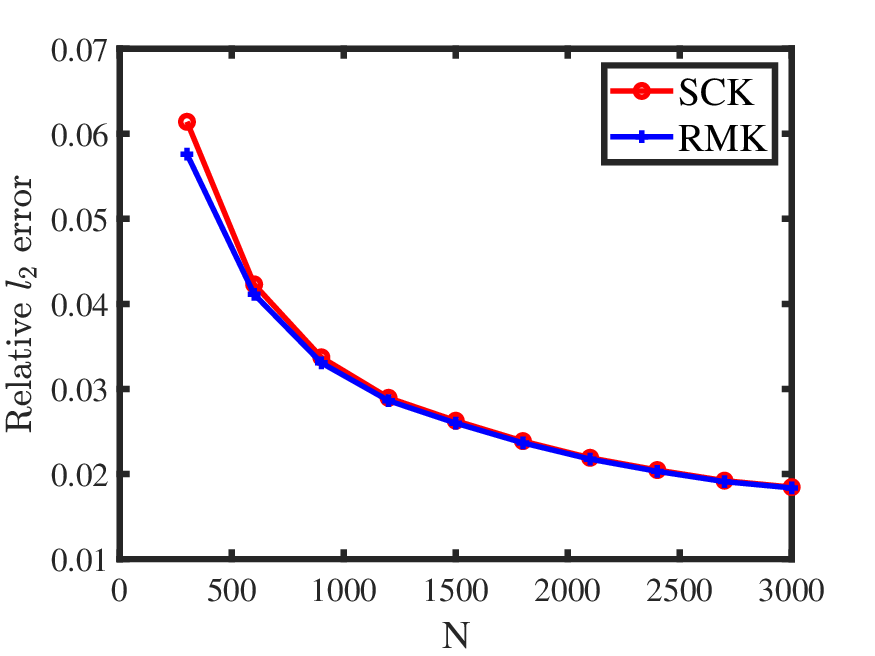}}
\subfigure[Simulation 6(b)]{\includegraphics[width=0.24\textwidth]{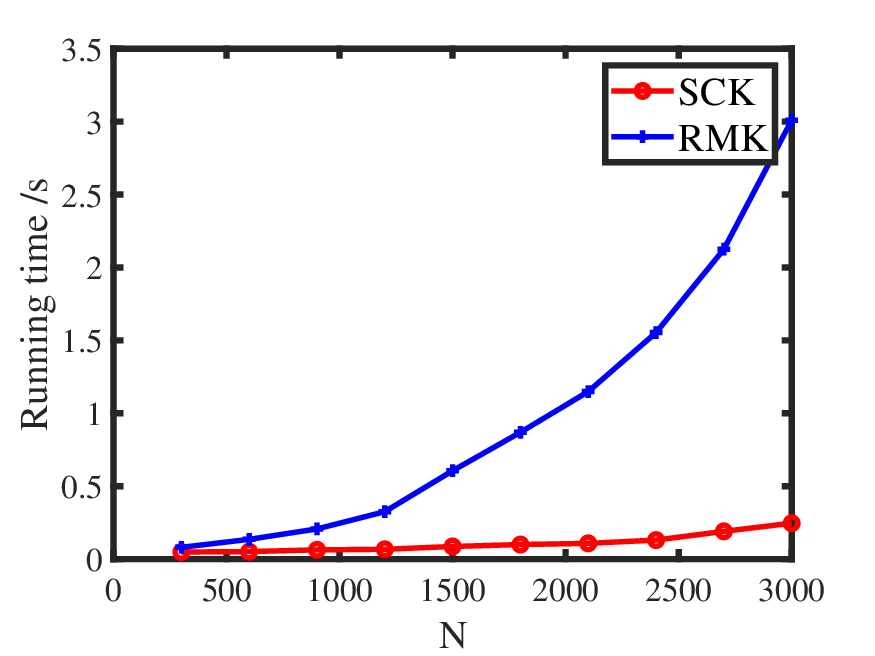}}
\caption{Numerical results of Simulation 6.}
\label{S6} 
\end{figure}
\subsubsection{Signed response matrix}
For signed response matrices when $\mathbb{P}(R(i,j)=1)=\frac{1+R_{0}(i,j)}{2}$ and $\mathbb{P}(R(i,j)=-1)=\frac{1-R_{0}(i,j)}{2}$ for $i\in[N], j\in[J]$, we consider the following two simulations.

\textbf{Simulation 7(a): changing $\rho$.} Set $N=500$. Recall that the theoretical range of the scaling parameter $\rho$ is $(0,1]$ for signed response matrices according to our analysis in Example \ref{Signed}, here, we let $\rho$ range in $\{0.1,0.2,\ldots,1\}$.

\textbf{Simulation 7(b): changing $N$.} Let $\rho=0.2$ and $N$ range in $\{1000,2000,\ldots,5000\}$.

Figure \ref{S7} shows the results. We see that increasing $\rho$ and $N$ improves the estimation accuracies of SCK and RMK, which confirms our analysis in Example \ref{Signed} and Corollary \ref{AddConditions}. Additionally, it is easy to see that both algorithms enjoy similar performances in estimating $Z$ and $\Theta$, and SCK requires less computation time compared to RMK.
\begin{figure}
\centering
\subfigure[Simulation 7(a)]{\includegraphics[width=0.24\textwidth]{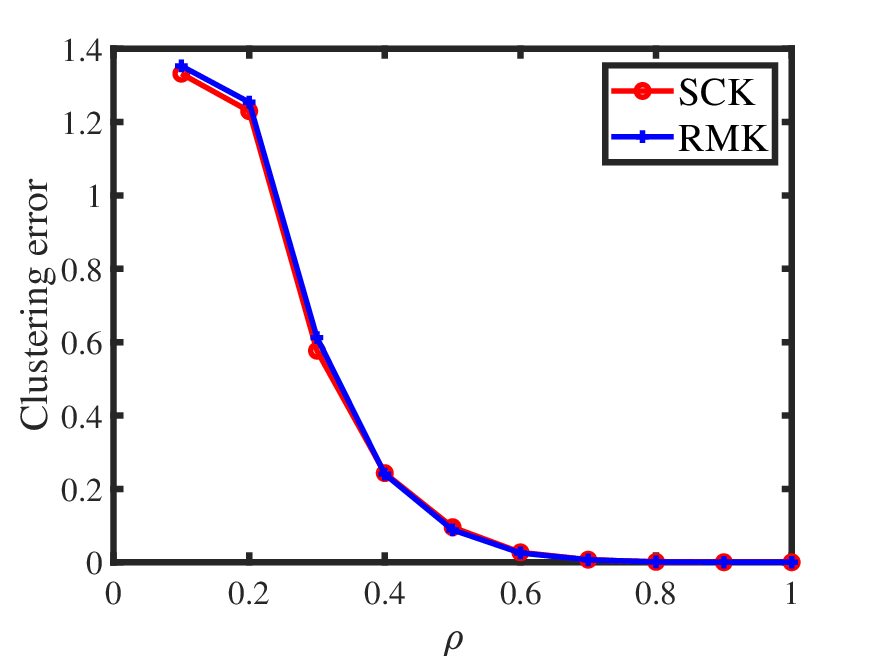}}
\subfigure[Simulation 7(a)]{\includegraphics[width=0.24\textwidth]{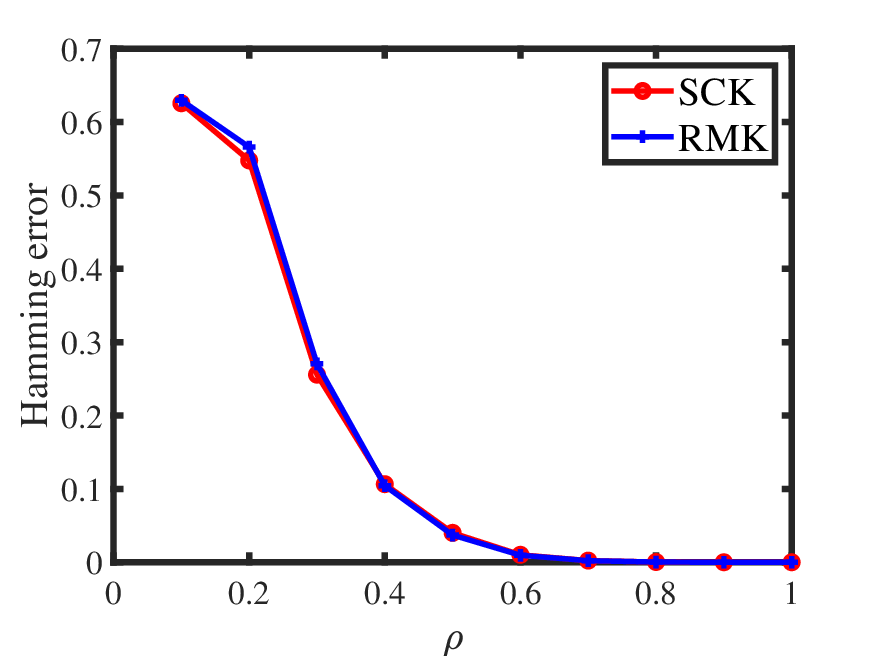}}
\subfigure[Simulation 7(a)]{\includegraphics[width=0.24\textwidth]{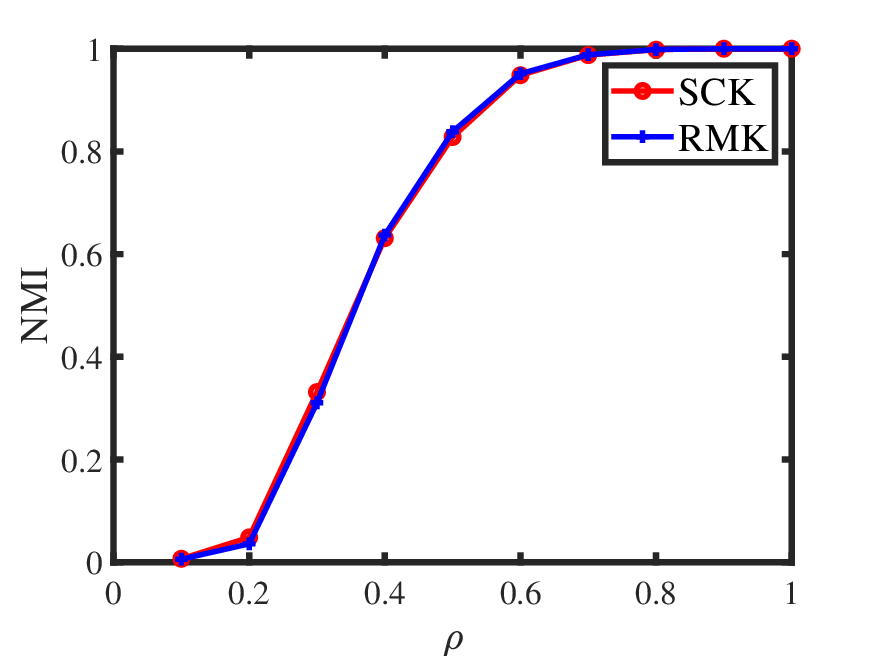}}
\subfigure[Simulation 7(a)]{\includegraphics[width=0.24\textwidth]{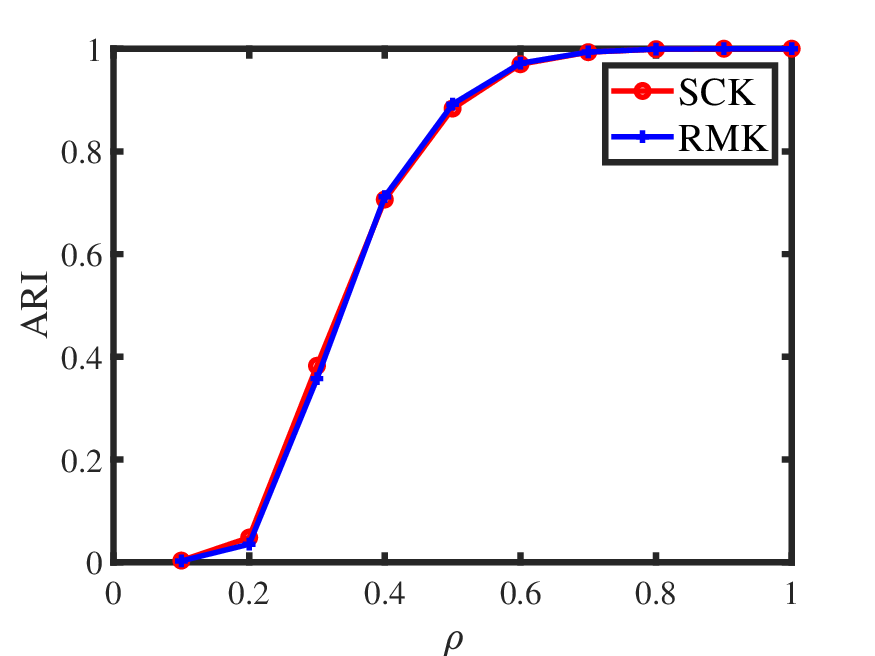}}
\subfigure[Simulation 7(a)]{\includegraphics[width=0.24\textwidth]{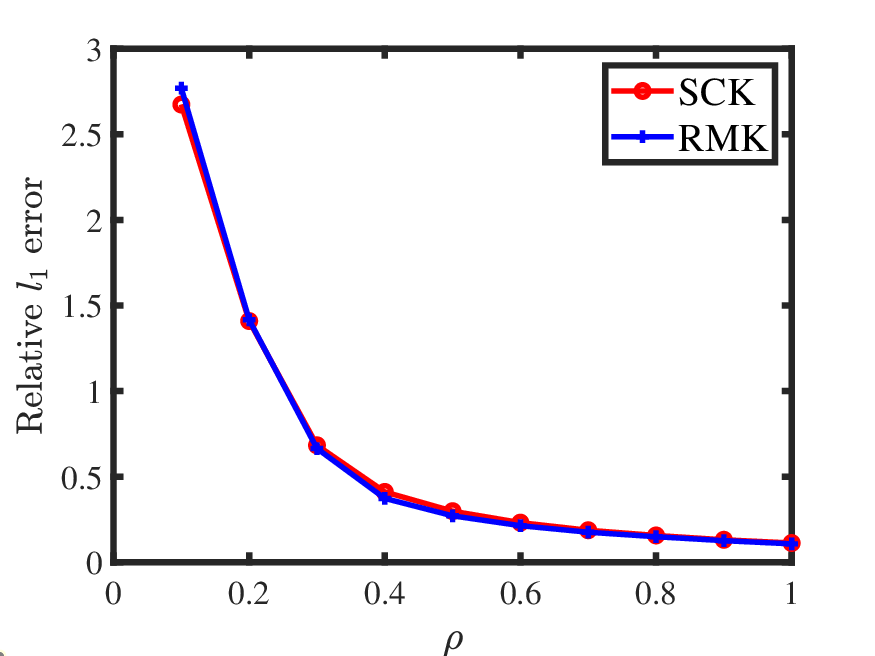}}
\subfigure[Simulation 7(a)]{\includegraphics[width=0.24\textwidth]{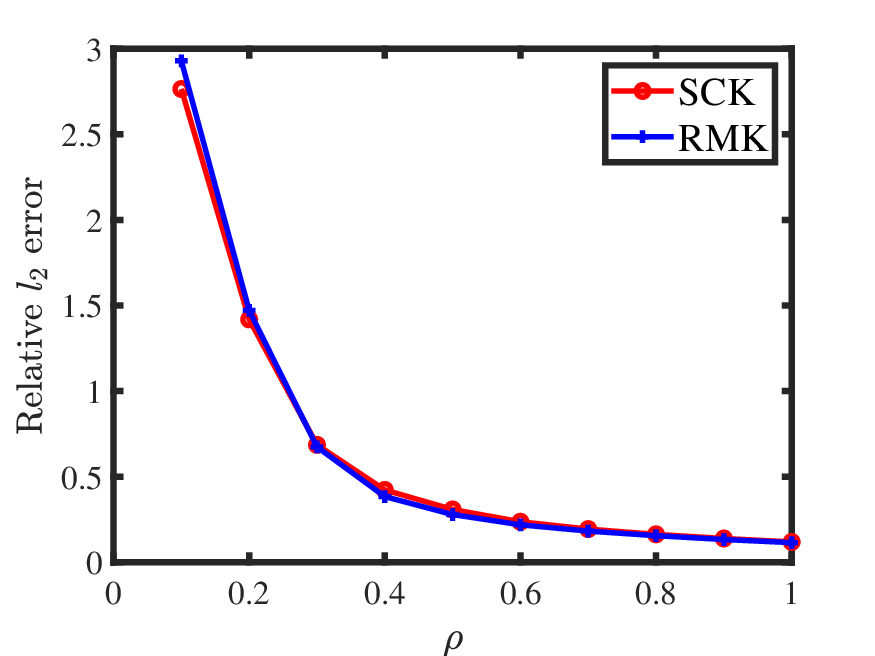}}
\subfigure[Simulation 7(a)]{\includegraphics[width=0.24\textwidth]{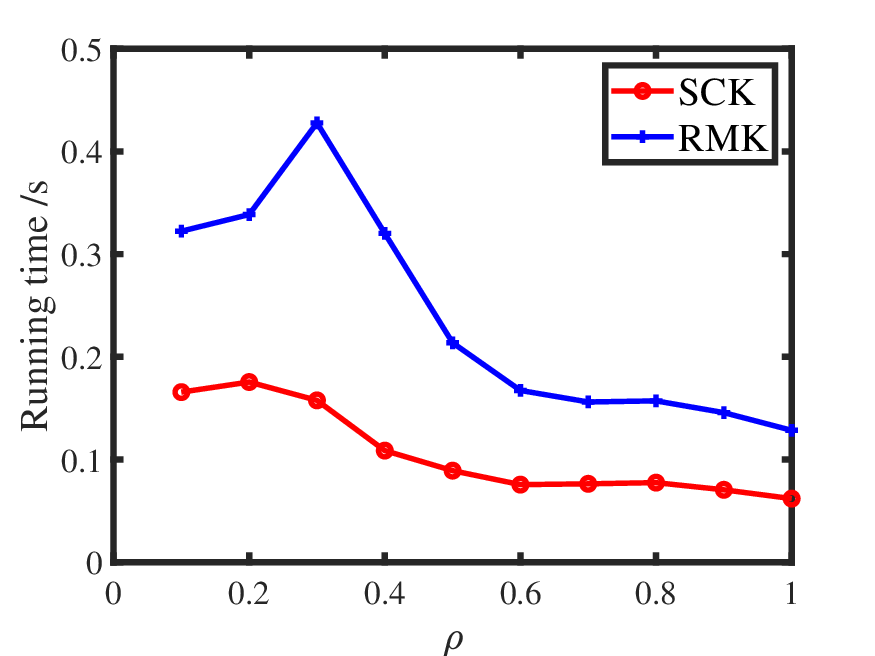}}
\subfigure[Simulation 7(b)]{\includegraphics[width=0.24\textwidth]{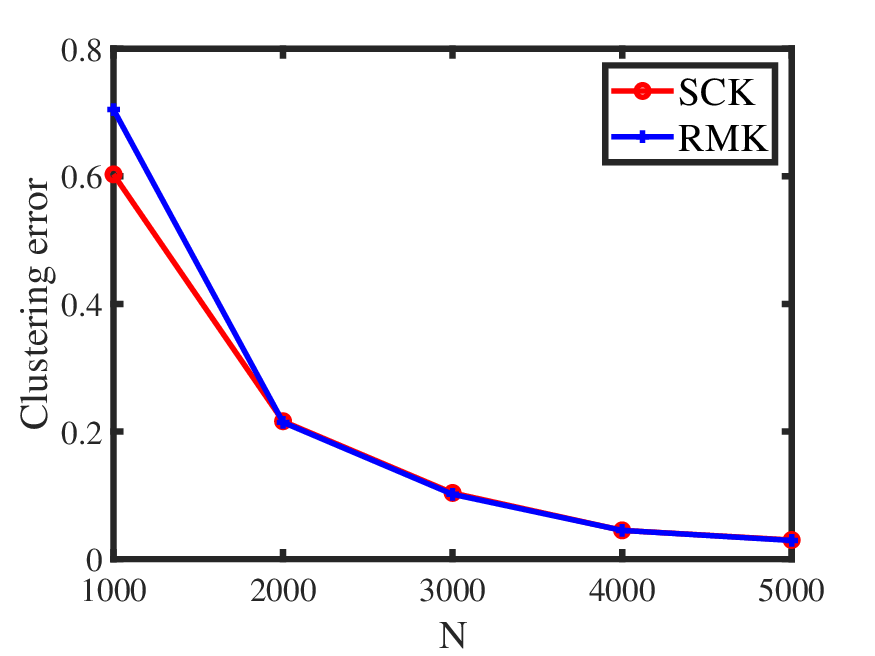}}
\subfigure[Simulation 7(b)]{\includegraphics[width=0.24\textwidth]{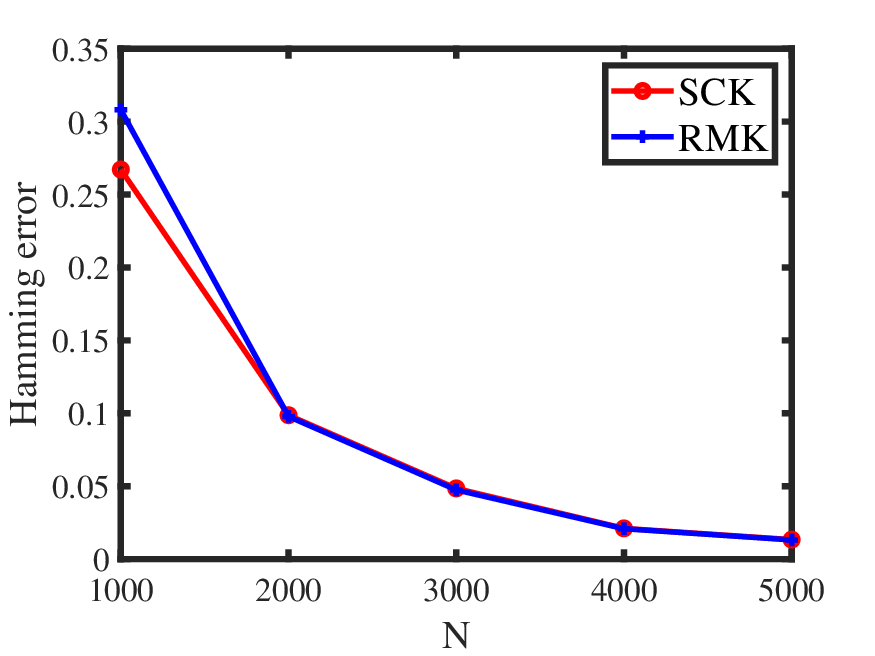}}
\subfigure[Simulation 7(b)]{\includegraphics[width=0.24\textwidth]{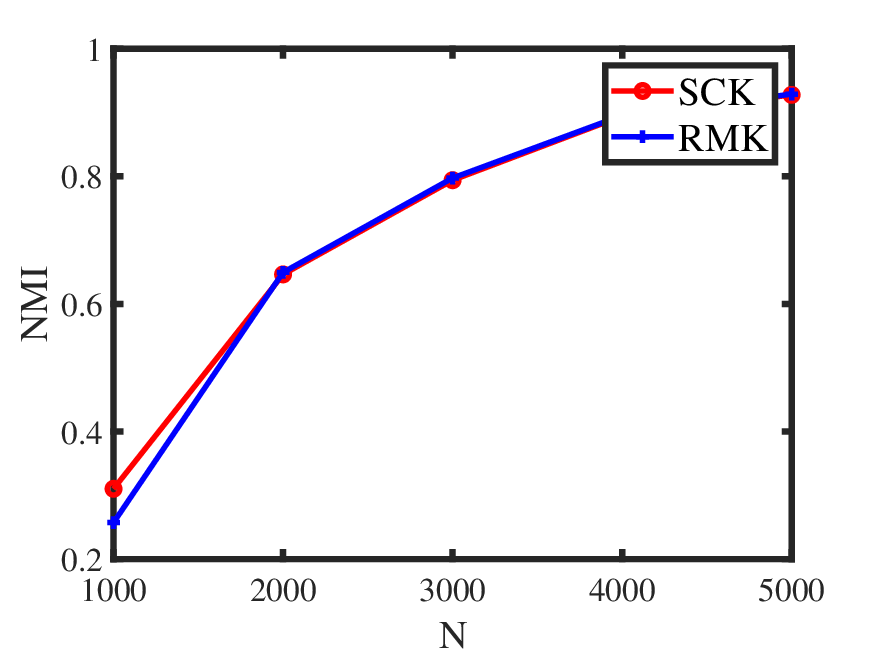}}
\subfigure[Simulation 7(b)]{\includegraphics[width=0.24\textwidth]{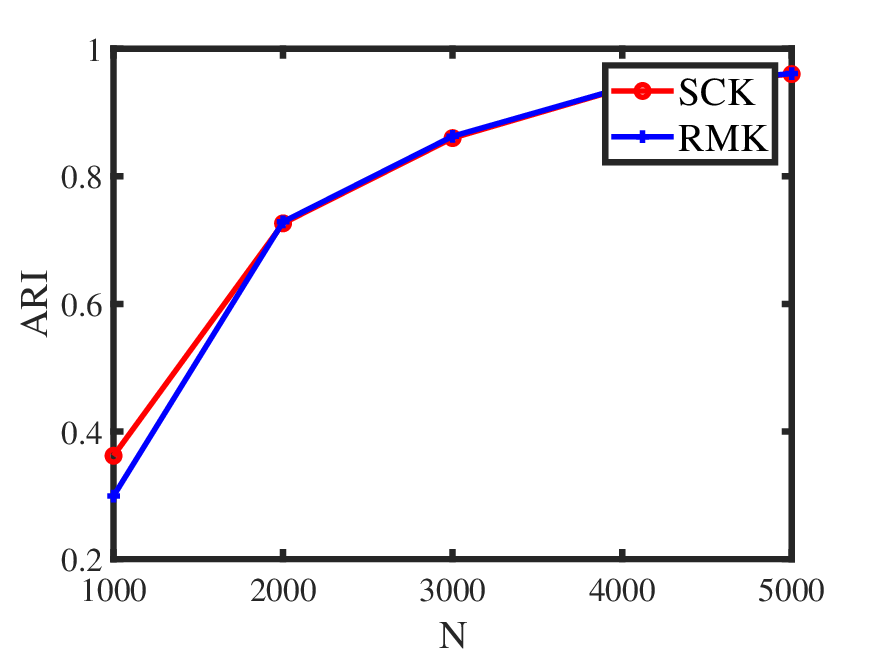}}
\subfigure[Simulation 7(b)]{\includegraphics[width=0.24\textwidth]{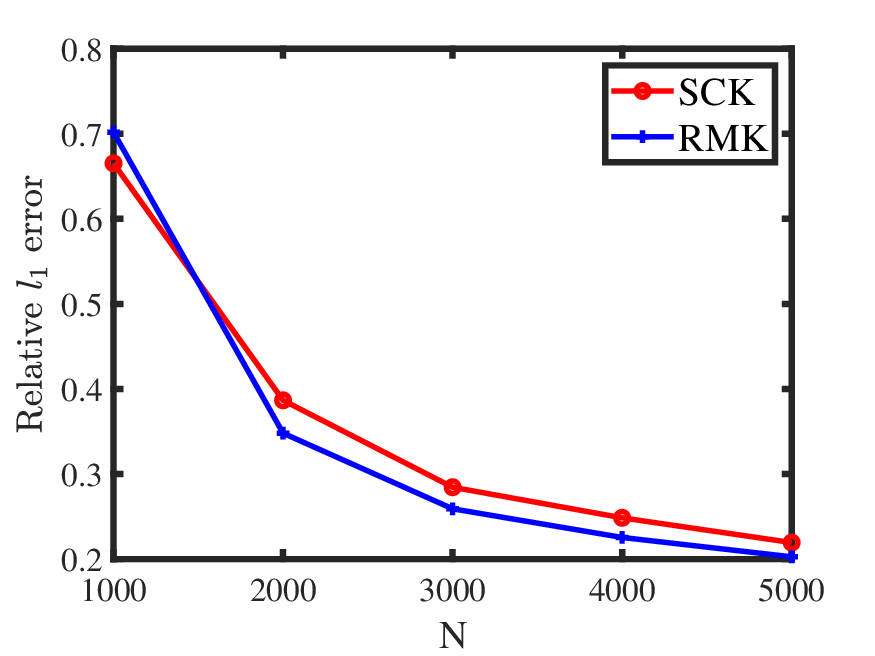}}
\subfigure[Simulation 7(b)]{\includegraphics[width=0.24\textwidth]{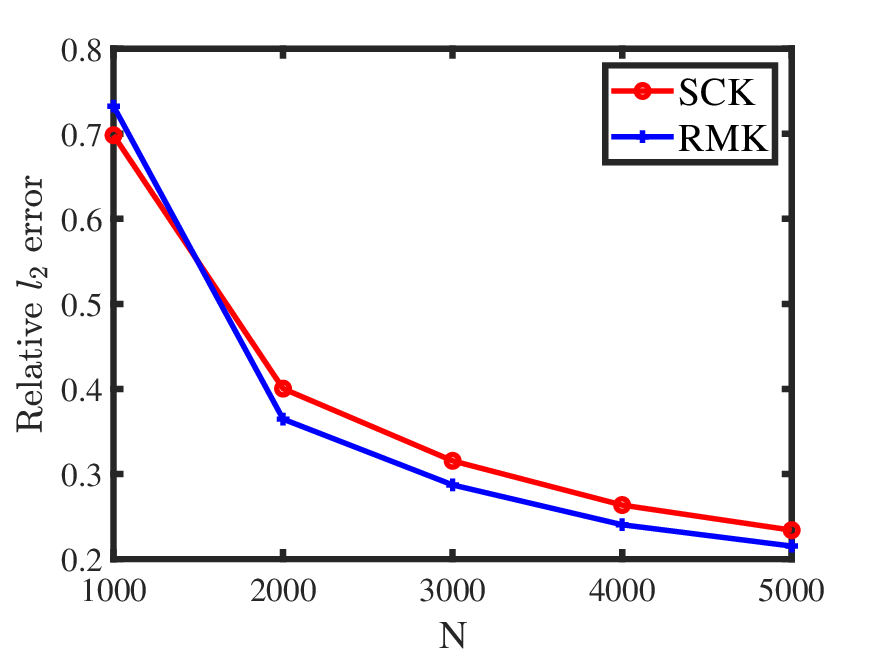}}
\subfigure[Simulation 7(b)]{\includegraphics[width=0.24\textwidth]{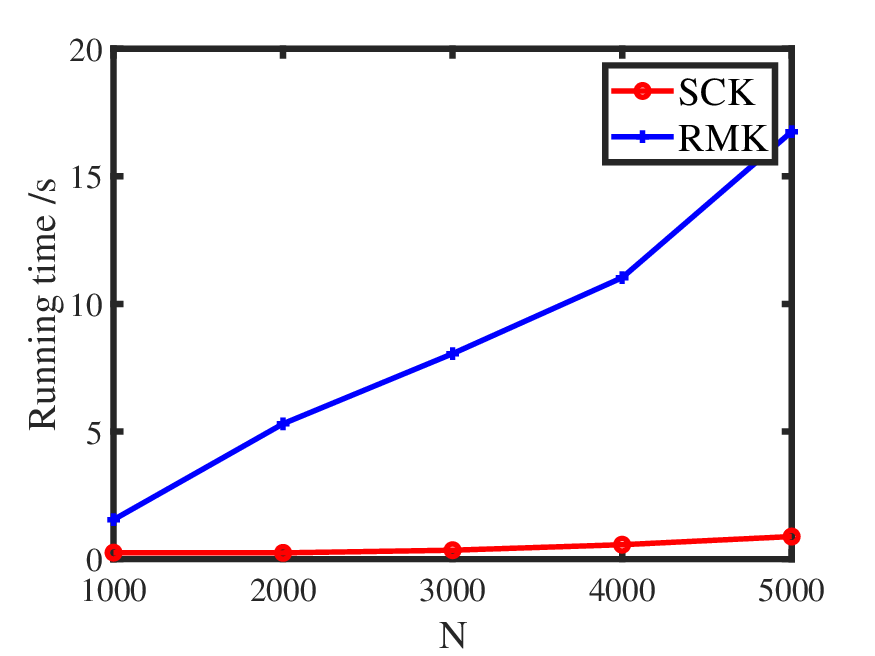}}
\caption{Numerical results of Simulation 7.}
\label{S7} 
\end{figure}
\subsubsection{Simulated weighted response matrices}
For visuality, we plot two weighted response matrices $R$ generated from the Normal distribution and the Poisson distribution under WLCM. Let $K=2,N=16,J=10,\sigma^{2}=1, \ell(i)=1, \ell(i+8)=2$ for $i\in[8]$, and $\Theta(j,1)=100, \Theta(j,2)=110-10j$ for $j\in[10]$. Because $R_{0}=Z\Theta'$ has been set, we can generate $R$ under different distributions with expectation $R_{0}$ under the proposed WLCM model. Here, we consider the following two settings.

\textbf{Simulation 8 (a):} When $R(i,j)\sim\mathrm{Normal}(R_{0}(i,j),\sigma^{2})$ for $i\in[N],j\in[J]$, the left panel of Figure \ref{S8} displays a weighted response matrix $R$ generated from Simulation 8 (a).

\textbf{Simulation 8 (b):} When $R(i,j)\sim\mathrm{Poisson}(R_{0}(i,j))$ for $i\in[N],j\in[J]$, the right panel of Figure \ref{S8} provides a $R$ generated from Simulation 8 (b).

\begin{figure}
\centering
\subfigure[Normal distribution]{\includegraphics[width=0.45\textwidth]{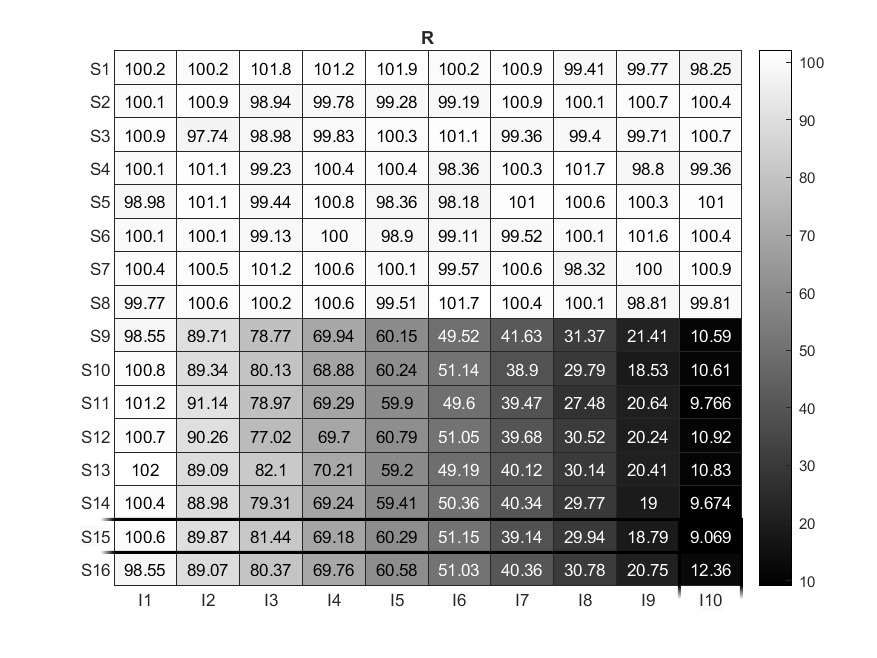}}
\subfigure[Poisson distribution]{\includegraphics[width=0.45\textwidth]{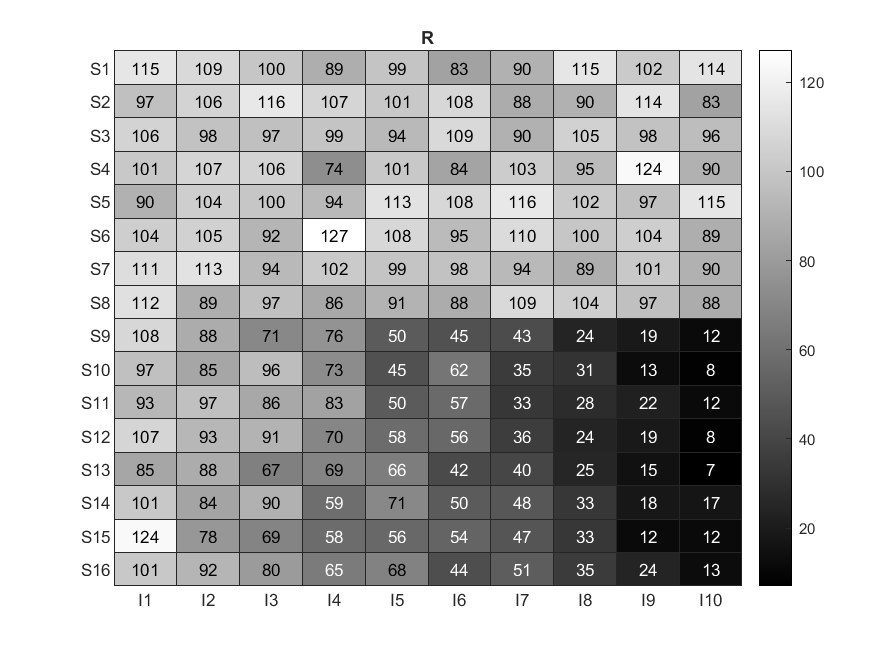}}
\caption{Illustration for weighted response matrices $R$ generated from WLCM. In both panels, S$i$ denote subject $i$ and I$j$ denotes item $j$ for $i\in[16],j\in[10]$.}
\label{S8} 
\end{figure}

\begin{table}[h!]
\footnotesize
	\centering
	\caption{Error rates of SCK and RMK for $R$ in Figure \ref{S8}. Values outside (and inside) the brackets are results for $R$ in panel (a) (and (b)) of Figure \ref{S8}.}
	\label{ErrorRatesSimulatedR}
	\begin{tabular}{cccccccccccc}
\hline\hline&Clustering error&Hamming error&NMI&ARI&Relative $l_{1}$ error&Relative $l_{2}$ error\\
\hline
SCK&0 (0)&0 (0)&1 (1)&1 (1)&0.0024 (0.0254)&0.0032 (0.0295)\\
RMK&0 (0)&0 (0)&1 (1)&1 (1)&0.0024 (0.0245)&0.0032 (0.0295)\\
\hline\hline
\end{tabular}
\end{table}
\begin{figure}
\centering
{\includegraphics[width=0.6\textwidth]{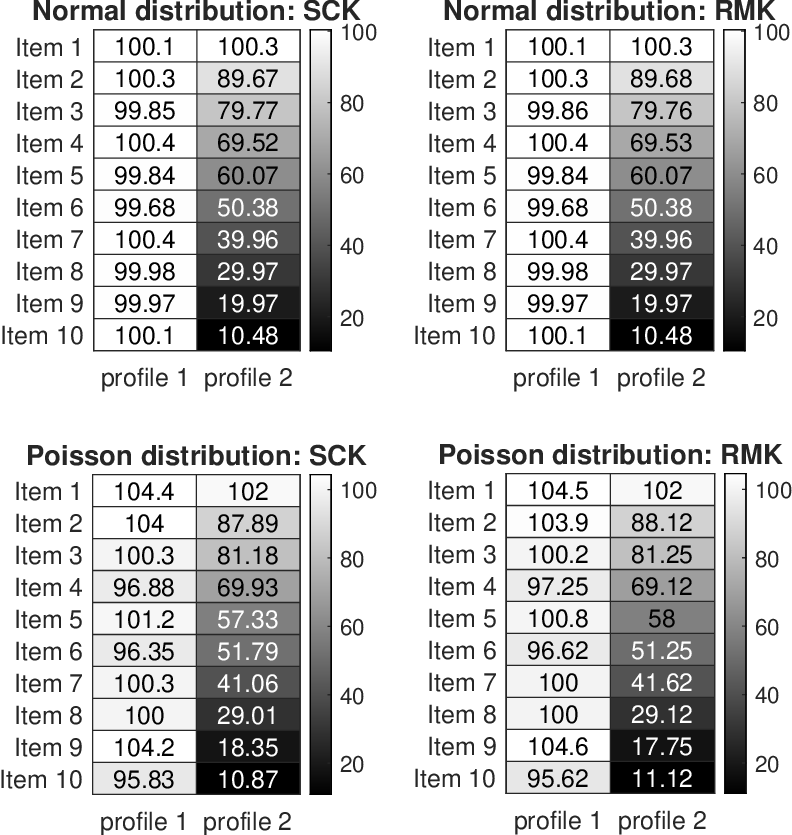}}
\caption{Heatmap of the estimated item parameter matrix $\hat{\Theta}$ of SCK and RMK for $R$ in Figure \ref{S8}.}
\label{heatmapSimulatedThetaHat} 
\end{figure}

Error rates of the proposed methods for the observed weighted response matrices provided in Figure \ref{S8} are displayed in Table \ref{ErrorRatesSimulatedR}. We also plot the estimated item matrix $\hat{\Theta}$ for both methods in Figure \ref{heatmapSimulatedThetaHat}.  We see that both approaches exactly recover $Z$ from $R$ while they estimate $\Theta$ with slight perturbations. Meanwhile, since $Z,\Theta,$ and $K$ are known for this simulation, $R$ provided in Figure \ref{S8} can be regarded as benchmark weighted response matrices, and readers can apply SCK and RMK (and other methods) to $R$ to check their effectiveness in estimating $Z$ and $\Theta$.
\section{Real data applications}\label{sec6realdata}
As the main goal of this paper is to introduce the proposed WLCM model and the SCK algorithm for weighted response matrices, this section reports empirical results on two data sets with weighted response matrices. Because the true classification matrix and the true item parameter matrix are unknown for real data, and SCK runs much faster than RMK, we only report the outcomes of the SCK approach. For real-world datasets, the number of extreme latent profiles $K$ is often unknown. Here, we infer $K$ for real-world weighted response matrices using the following strategy:
\begin{align}\label{EstimateKRealData}
K=\mathrm{arg~min}_{k\in[\mathrm{rank}(R)]}\|R-\hat{Z}\hat{\Theta}'\|,
\end{align}
where $\hat{Z}$ and $\hat{\Theta}$ are outputs in Algorithm \ref{alg:SCK} with inputs $R$ and $k$. The method specified in Equation (\ref{EstimateKRealData}) selects $K$ by picking the one that minimizes the spectral norm difference between $R$ and $\hat{Z}\hat{\Theta}'$. The determination of the number of extreme latent profiles K in our WLCM model in a rigorous manner with theoretical guarantees remains a future direction.
\subsection{International Personality Item Pool (IPIP) personality test data}
\textbf{Background.} We apply SCK to an experiment personality test data called the International Personality Item Pool (IPIP) personality test, which is obtainable for download at \url{https://openpsychometrics.org/_rawdata/}. This data consists of 1005 subjects and 40 items. The IPIP data also records the age and gender of each subject. After dropping subjects with missing entries in their responses, age, or gender, and dropping two subjects that are neither male nor female, there are 896 subjects left, i.e., $N=896, J=40$. All items are rated on a 5-point scale, where 1=Strongly disagree, 2=Disagree, 3=Neither agree not disagree, 4=Agree, 5=Strongly agree, i.e., $R\in\{1,2,3,4,5\}^{896\times40}$, a weighted response matrix. Items 1-10 measure the personality factor Assertiveness (short as ``AS”); Items 11-20 measure the personality factor Social confidence (short as “SC”); Items 21-30 measure the personality factor Adventurousness (short as ``AD”); Items 31-40 measure the personality factor Dominance (short as ``DO”). The details of each item are depicted in Figure \ref{IPIPHeatmap}.

\textbf{Analysis.} We apply Equation (\ref{EstimateKRealData}) to infer $K$ for the IPIP dataset and find that the estimated value of $K$ is 3. We then apply the SCK algorithm to the response matrix $R$ with $K=3$ to obtain the $896\times 3$ matrix $\hat{Z}$ and the $40\times3$ matrix $\hat{\Theta}$. The running time for SCK on this dataset is around 0.2 seconds.
\begin{table}[h!]
\footnotesize
	\centering
	\caption{Basic information for each estimated extreme latent profile obtained from $\hat{Z}$ for the IPIP data.}
	\label{ResultsBasicIPIP}
	\begin{tabular}{cccccccccccc}
\hline\hline&profile 1&profile 2&profile 3\\
\hline
Size&276&226&394\\
\#Male&123&129&241\\
\#Female&153&97&153\\
Average age of male&35.9837&32.8240&35.9004\\
Average age of female&35.5425&31.3814&38.7059\\
\hline\hline
\end{tabular}
\end{table}

\textbf{Results.} For convenience, we denote the estimated three extreme latent profiles as profile 1, profile 2, and profile 3. Based on $\hat{Z}$ and the information of age and gender, we can obtain some basic information (shown in Table \ref{ResultsBasicIPIP}) such as the size of each profile, number of males (females) in each profile, and the average age of males (and females) in each profile. From Table \ref{ResultsBasicIPIP}, we see that the number of females is larger than that of males for profile 1 while profiles 2 and 3 have more males. The average age of males (and females) in profile 2 is smaller than that of profiles 1 and 3 while the average age of females in profile 3 is the largest. We can also obtain the average point on each item for males (and females) in each estimated extreme latent profile and the results are shown in panel (a) (and panel (b)) of Figure \ref{IPIPHeatmap}. We observe that males in profile 3 tend to be more confident, more creative, more social, and more open to changes than males in profiles 1 and 2; males in profile 3 are more (less) dominant than males in profile 1 (profile 2). Males in profile 2 are more confident\&creative\&social\&open to changes\&dominant than males in profile 1. Meanwhile, in the three estimated extreme latent profiles, females enjoy similar personalities to males. We also find that males in profile 3 (profile 2) are more (less) confident\&creative\&social\&open to changes\&dominant than females in profile 3 (profile 2). Furthermore, it is interesting to see that, though males in profile 1 are less confident\&creative\&social\&open to changes than females in profile 1, they are more dominant than females in profile 1. We also plot the average point on each item in each estimated extreme latent profile regardless of gender in panel (c) of Figure \ref{IPIPHeatmap} where we can draw similar conclusions as that for male. In panel (d) of Figure \ref{IPIPHeatmap}, we plot the heatmap of the estimated item parameter matrix $\hat{\Theta}$. By comparing panel (c) with panel (d), we see that the $(j,k)$-th element in the matrix shown in panel (c) is close to $\hat{\Theta}(j,k)$ for $j\in[40], k\in[3]$. Such a result implies that the behavior differences on each item for every extreme latent profile are governed by the item parameter matrix $\Theta$.
\begin{rem}\label{Why}
Recall that $\mathbb{E}(R)=R_{0}=Z\Theta'$ under the WLCM model, we have $R_{0}(i,j)=\Theta(j,\ell(i))$ for $i\in[N],j\in[J]$. Then we have $\sum_{\ell(i)\equiv k}R_{0}(i,j)=\sum_{\ell(i)\equiv k}\Theta(j,\ell(i))=\sum_{\ell(i)\equiv k}\Theta(j,k)=N_{k}\Theta(j,k)$ which gives that $\Theta(j,k)=\frac{\sum_{\ell(i)\equiv k}R_{0}(i,j)}{N_{k}}$ for $k\in[K]$. This interprets why the average value on the $j$-th item in the $k$-th estimated extreme latent profile approximates $\hat{\Theta}(j,k)$ for $j\in[J],k\in[K]$.
\end{rem}
\begin{figure}
\centering
\subfigure[Average point on each item for male in each estimated extreme latent profile]{\includegraphics[width=0.45\textwidth]{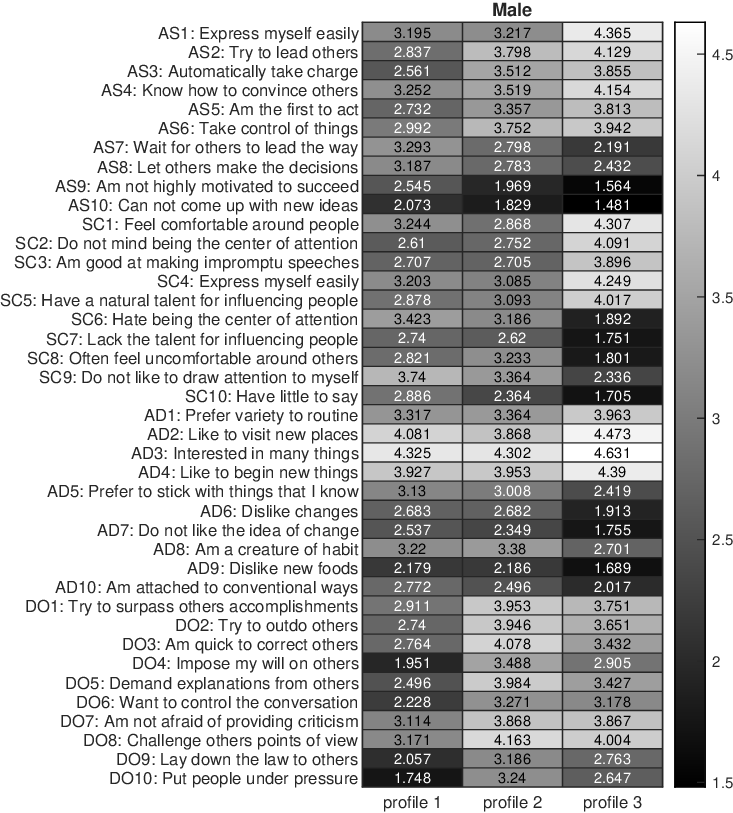}}
\subfigure[Average point on each item for female in each estimated extreme latent profile]{\includegraphics[width=0.45\textwidth]{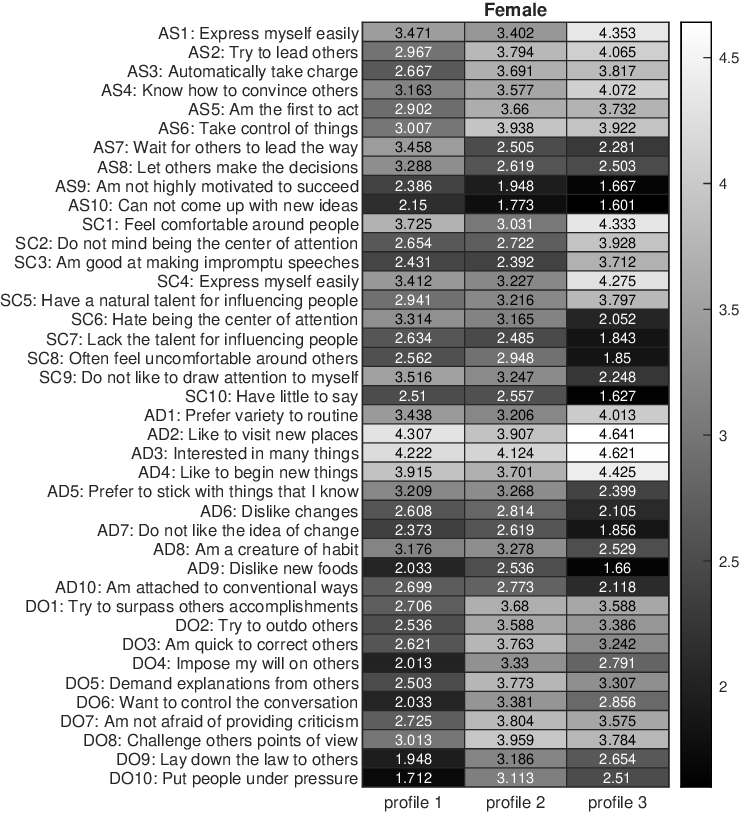}}
\subfigure[Average point on each item in each estimated extreme latent profile]{\includegraphics[width=0.45\textwidth]{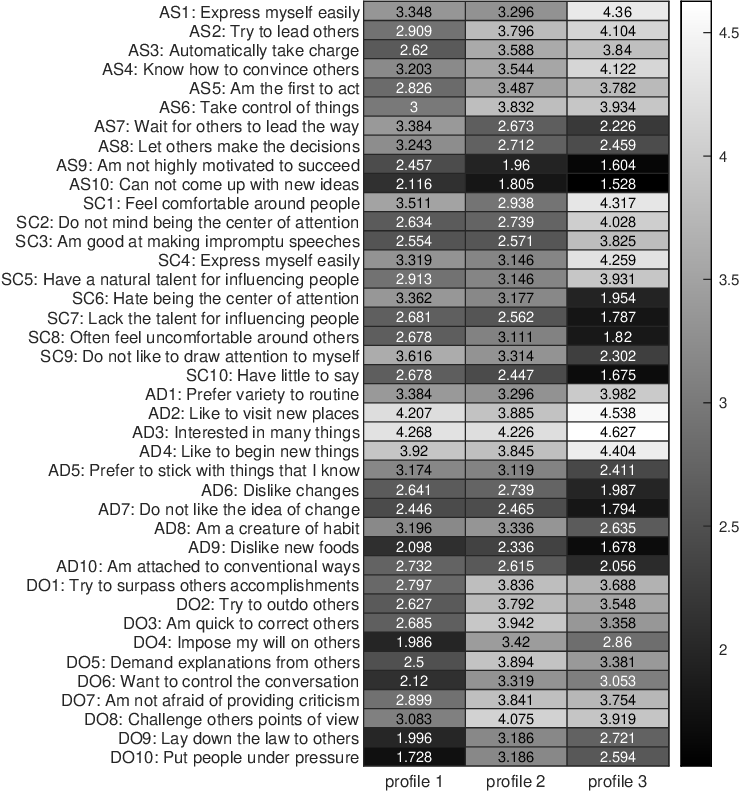}}
\subfigure[Heatmap of $\hat{\Theta}$]{\includegraphics[width=0.45\textwidth]{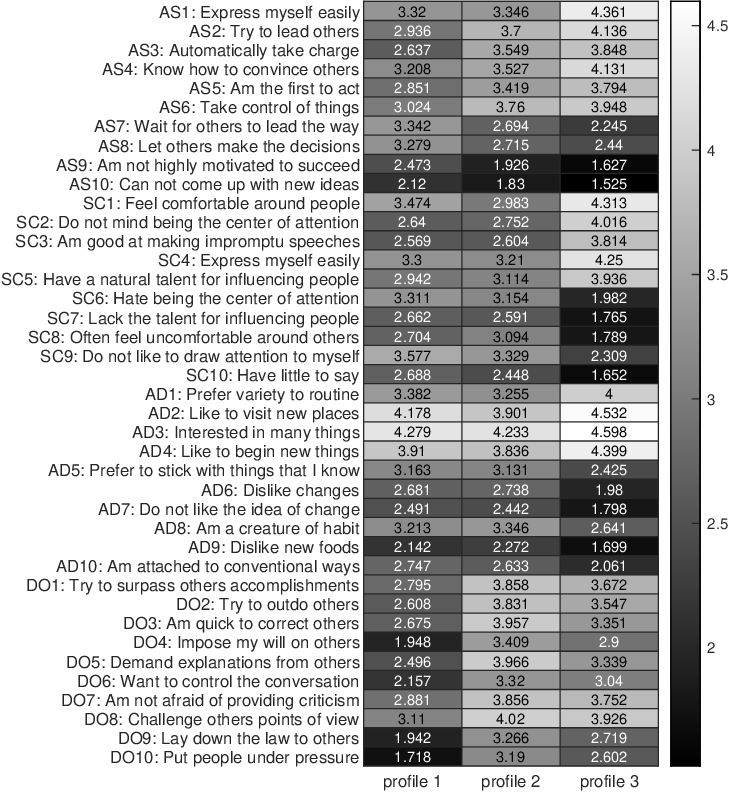}}
\caption{Numerical results for the IPIP data.}
\label{IPIPHeatmap} 
\end{figure}
\subsection{Big Five Personality Test with Random Number (BFPTRN) data}
\textbf{Background.} Our SCK method is also applied to personality test data: the Big Five Personality Test with Random Number (BFPTRN) data. This dataset can be downloaded from the same URL as the IPIP data. This data asks respondents to generate random numbers in certain ranges attached to 50 personality items. The Big Five personality traits are extraversion (items E1-E10), neuroticism (items N1-N10), agreeableness (items A1-A10), conscientiousness (items C1-C10), and openness (items O1-O10). The original BFPTRN data contains 1369 subjects. After excluding subjects with missing responses or missing random numbers and removing those with random numbers exceeding the specified range, there remain 1155 subjects, i.e., $N=1155, J=50$. All items are rated using the same 5-point scale as the IPIP data, which results in $R\in\{1,2,3,4,5\}^{1155\times50}$ being weighted. The detail of each item and each range for random numbers can be found in Figure \ref{BFPTRNHeatmap}.

\textbf{Analysis.} The estimated number of extreme latent profiles for the BFPTRN dataset is 3. Applying the SCK approach to $R$ with $K=3$ produces the $1155\times 3$ matrix $\hat{Z}$ and the $50\times 3$ matrix $\hat{\Theta}$. SCK takes around 1.6 seconds to process this data.

\textbf{Results.} Without confusion, we also let profile 1, profile 2, and profile 3 represent the three estimated extreme latent profiles. Profile 1,2, and 3 have 409, 320, and 426 subjects, respectively. Similar to the IPIP data, based on $\hat{Z}$ and $\hat{\Theta}$, we can also obtain the heatmap of the average point on each subject for every profile, the heatmap of the average random number on each range for every profile, and the heatmap of $\hat{\Theta}$ as shown in Figure \ref{BFPTRNHeatmap}. We observe that there is no significant connection between the average point and the average random number on each item in each estimated extreme latent profile. From panel (a) of Figure \ref{BFPTRNHeatmap}, we find that: for extraversion, subjects in profile 1 are the most extrovertive while subjects in profile 2 are the most introverted; for neuroticism, subjects in profile 3 are emotionally stable while subjects in profiles 1 and 2 are emotionally unstable; for agreeableness, subjects in profiles 1 and 3 are easier to get along with than subjects in profile 2; for conscientiousness, subjects in profile 3 are more responsible that those in profiles 1 and 2; for openness, subjects in profiles 1 and 3 are more open than those in profile 2. Meanwhile, the matrix shown in panel (a) approximates $\hat{\Theta}$ well, which has been explained in Remark \ref{Why}.
\begin{figure}
\centering
\subfigure[Average point on each item in each estimated extreme latent profile]{\includegraphics[width=0.45\textwidth]{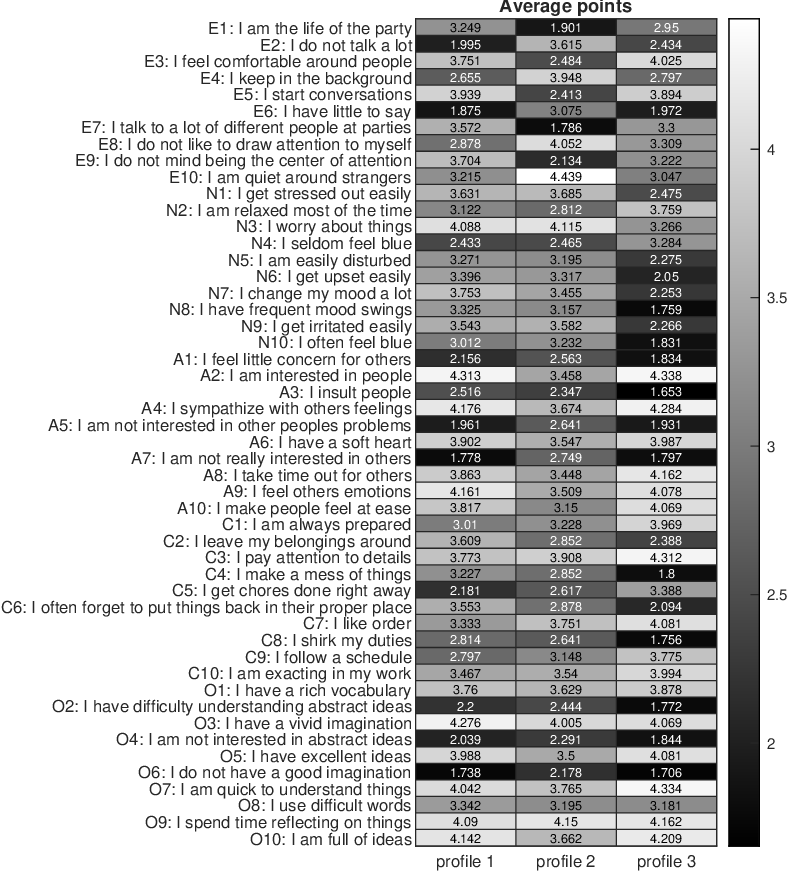}}
\subfigure[Average random number on each range in each estimated extreme latent profile]{\includegraphics[width=0.45\textwidth]{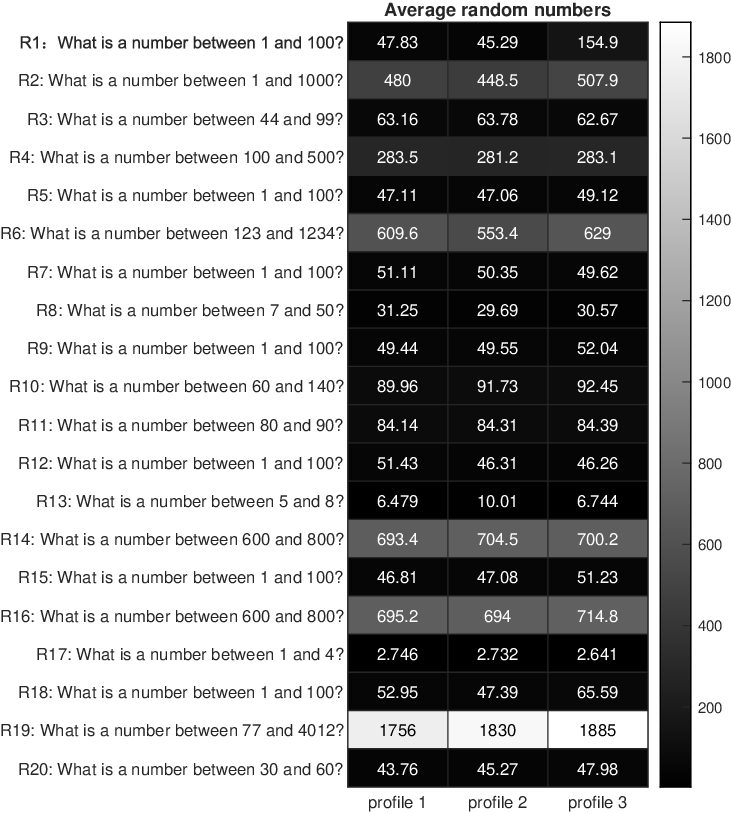}}
\subfigure[Heatmap of $\hat{\Theta}$]{\includegraphics[width=0.45\textwidth]{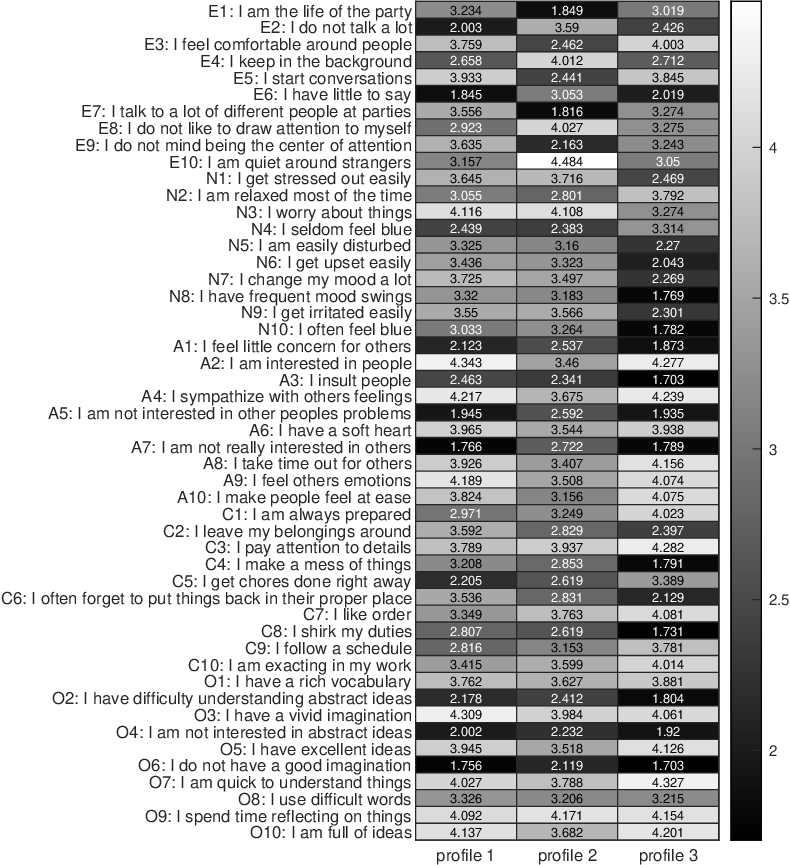}}
\caption{Numerical results for the BFPTRN data.}
\label{BFPTRNHeatmap} 
\end{figure}
\section{Conclusion and future work}\label{sec7}
In this paper, we introduced the weighted latent class model (WLCM), a novel class of latent class analysis models for categorical data with weighted responses. We studied its model identifiability, developed an efficient inference method SCK to fit WLCM, and built a theoretical guarantee of estimation consistency for the proposed method under WLCM. On the methodology
side, the new model WLCM provides exploratory and useful tools for latent class analysis in applications where the categorical data may have weighted responses. WLCM allows the observed weighted response matrix to be generated from any distribution as long as its expectation follows a latent class structure modeled by WLCM. In particular, the popular latent class model is a sub-model of our WLCM, and categorical data with signed responses can also be modeled by WLCM. Ground-truth latent classes of categorical data with weighted responses generated from WLCM serve as benchmarks for evaluating latent class analysis approaches. On the algorithmic side, the SVD-based spectral method SCK is efficient and easy to implement. SCK requires no tuning parameters and it is applicable for any categorical data with weighted responses. This means that researchers in fields such as social, psychological, behavioral, biological sciences, and beyond can design their tests/evaluations/surveys/interviews without worrying that the response should be binary or positive, as our method SCK is applicable for any weighted response matrices in latent class analysis. On the theoretic side, we established the rate of convergence for our method SCK under the proposed model WLCM. We found that SCK exhibits different behaviors when the weighted response matrices are generated from different distributions, and we conducted extensive experiments to verify our theoretical findings. Empirically, we applied our method to two real categorical datasets with weighted responses. We expect that our WLCM model and SCK method will have broad applications for latent class analysis of data with weighted responses in diverse fields, similar to the widespread use of latent class models in recent years.

There are several future directions worth exploring. First, methods with theoretical guarantees should be designed to determine the number of extreme latent profiles $K$ for observed weighted response matrices generated from any distribution $\mathcal{F}$ under WLCM. Second, the grade of membership (GoM) model \citep{woodbury1978mathematical,erosheva2005comparing} provides a richer modeling capacity than the latent class model since GoM allows a subject to belong to multiple extreme latent profiles. Therefore, following the distribution-free idea developed in this work, it is meaningful to extend the model GoM to categorical data with weighted responses.
Third, like the LCM can be equipped with individual covariates \citep{huang2004building,forcina2008identifiability,reboussin2008locally,vermunt2010latent,di2020random,bakk2022two}, it is worth considering additional individual covariates into the WLCM analysis. Fourth, our WLCM only considers static latent class analysis and it is meaningful to extend WLCM to the dynamic case \citep{asparouhov2017dynamic}. Fifth, our SCK is a spectral clustering method and it is possible to speed up it by applications of the random-projection techniques \citep{zhang2022randomized} or the distributed spectral clustering idea \citep{wu2023distributed} to deal with large-scale categorical data for latent class analysis.
\section*{CRediT authorship contribution statement}
\textbf{Huan Qing} is the sole author of this paper.
\section*{Declaration of competing interest}
The author declares no competing interests.
\section*{Data availability}
Data and code will be made available on request.
\appendix
\section{Proofs under WLCM}
\subsection{Proof of Proposition \ref{idWLCM}}
\begin{proof}
According to Lemma \ref{UVWLCM}, we know that $U=ZX$, where $X=\Theta'V\Sigma^{-1}$. Similarly, $U$ can be rewritten as $U=\tilde{Z}\tilde{X}$, where $\tilde{X}=\tilde{\Theta}'V\Sigma^{-1}$. Then, for $i\in[N]$, we have
\begin{align}\label{Ui}
U(i,:)=Z(i,:)X=X(\ell(i),:)=\tilde{Z}(i,:)\tilde{X}=\tilde{X}(\tilde{\ell}(i),:),
\end{align}
where $\tilde{\ell}(i)$ denotes the extreme latent profile that the $i$-th subject belongs in the alternative classification matrix $\tilde{Z}$. For $\bar{i}\in[N]$ and $\bar{i}\neq i$, we have
\begin{align}\label{Uibar}
U(\bar{i},:)=Z(\bar{i},:)X=X(\ell(\bar{i}),:)=\tilde{Z}(\bar{i},:)\tilde{X}=\tilde{X}(\tilde{\ell}(\bar{i}),:).
\end{align}
When $\ell(i)=\ell(\bar{i})$, by the second statement of Lemma \ref{UVWLCM}, we get $U(i,:)=U(\bar{i},:)$. Combining this fact (i.e., $U(i,:)=U(\bar{i},:)$) with Equations (\ref{Ui}) and (\ref{Uibar}) leads to
\begin{align}\label{X4}
X(\ell(i),:)=X(\ell(\bar{i}),:)=\tilde{X}(\tilde{\ell}(i),:)=\tilde{X}(\tilde{\ell}(\bar{i}),:)\mathrm{~when~}\ell(i)=\ell(\bar{i}).
\end{align}
Equation (\ref{X4}) implies that $\tilde{\ell}(i)=\tilde{\ell}(\bar{i})$ when $\ell(i)=\ell(\bar{i})$, i.e., for any two distinct subjects $i$ and $\bar{i}$, they are in the same extreme latent profile under $\tilde{Z}$ when they are in the same extreme latent profile under $Z$. Therefore, we have $\tilde{Z}=Z\mathcal{P}$, where $\mathcal{P}$ is a permutation matrix. Combining $\tilde{Z}=Z\mathcal{P}$ with $Z\Theta'=\tilde{Z}\tilde{\Theta}'$ leads to $Z\Theta'=\tilde{Z}\tilde{\Theta}'=Z\mathcal{P}\tilde{\Theta}'$, which gives that
\begin{align}\label{ZPTheta1}
Z(\Theta'-\mathcal{P}\tilde{\Theta}')=0.
\end{align}
Taking the transposition of Equation (\ref{ZPTheta1}) gives
\begin{align}\label{ZPTheta2}
(\Theta-\tilde{\Theta}\mathcal{P}')Z'=0.
\end{align}
Right multiplying $Z$ at both sides of Equation (\ref{ZPTheta2}) gives
\begin{align}\label{ZPTheta2}
(\Theta-\tilde{\Theta}\mathcal{P}')Z'Z=0.
\end{align}
Since each extreme latent profile is not an empty set, the $N\times K$ classification matrix $Z$ has a rank $K$, which gives that the $K\times K$ matrix $Z'Z$ is nonsingular. Therefore, Right multiplying $(Z'Z)^{-1}$ at both sides of Equation (\ref{ZPTheta2}) gives $\Theta=\tilde{\Theta}\mathcal{P}'$, i.e., $\tilde{\Theta}=\Theta\mathcal{P}$ since $\mathcal{P}$ is a permutation matrix.
\end{proof}
\subsection{Proof of Lemma \ref{UVWLCM}}
\begin{proof}
For the first statement: Since $R_{0}=Z\Theta'=U\Sigma V'$, $V'V=I_{K_{0}\times K_{0}}$, and the $K_{0}\times K_{0}$ diagonal matrix $\Sigma$ is nonsingular, we have $U=Z\Theta'V\Sigma^{-1}\equiv ZX$, where $X=\Theta'V\Sigma^{-1}$. Hence, the first statement holds.

For the second statement: For $i\in[N]$, $U=ZX$ gives $U(i,:)=Z(i,:)X=X(\ell(i),:)$. Then, if $\ell(\bar{i})=\ell(i)$, we have $U(\bar{i},:)=X(\ell(\bar{i}),:)=X(\ell(i),:)=U(i,:)$, i.e., $U$ has $K$ distinct rows. Thus, the second statement holds.

For the third statement: Since $R_{0}=Z\Theta'=U\Sigma V'$, we have $\Theta Z'=V\Sigma U'\Rightarrow \Theta Z'Z=V\Sigma U'Z\Rightarrow\Theta=V\Sigma U'Z(Z'Z)^{-1}$ where the $K\times K$ matrix $Z'Z$ is nonsingular because each extreme latent profile has at least one subject, i.e., $\mathrm{rank}(Z'Z)=\mathrm{rank}(Z)=K$. Thus, the third statement holds.

For the fourth statement: Recall that when $K_{0}=K$, we have $U\in\mathbb{R}^{N\times K},V\in\mathbb{R}^{J\times K}$, $\Sigma$ is a $K\times K$ full-rank diagonal matrix, and $X$ is a $K\times K$ matrix, where $U'U=I_{K\times K}, V'V=I_{K\times K}$. Let $\Delta=\mathrm{diag}(\sqrt{N_{1}},\sqrt{N_{2}},\ldots,\sqrt{N_{K}})$, then
\begin{align}\label{R01}
R_{0}=Z\Theta'=Z\Delta^{-1}\Delta\Theta'.
\end{align}
It is straightforward to verify that $Z\Delta^{-1}$ is a column orthogonal matrix, i.e., $(Z\Delta^{-1})'Z\Delta^{-1}=I_{K\times K}$.

Since $K_{0}=K$, we have $\mathrm{rank}(\Delta\Theta')=K$. Let $\tilde{U}\tilde{\Sigma}\tilde{V}'=\Delta\Theta'$ be the compact SVD of $\Delta\Theta'$, where $\tilde{\Sigma}$ is a $K\times K$ diagonal matrix, $\tilde{U}\in\mathbb{R}^{K\times K}, \tilde{V}\in\mathbb{R}^{J\times K}, \tilde{U}'\tilde{U}=I_{K\times K}$, and $\tilde{V}'\tilde{V}=I_{K\times K}$. Note that $\tilde{U}\in\mathbb{R}^{K\times K}$ and $\tilde{U}'\tilde{U}=I_{K\times K}$ imply $\mathrm{rank}(\tilde{U})=K$. Equation (\ref{R01}) implies
\begin{align}\label{R02}
R_{0}=Z\Theta'=U\Sigma V'=Z\Delta^{-1}\Delta\Theta'=Z\Delta^{-1}\tilde{U}\tilde{\Sigma}\tilde{V}'.
\end{align}
Note that $U, V, Z\Delta^{-1}\tilde{U}$, and $\tilde{V}$ are all orthonormal matrices. Also note that $\Sigma$ and $\tilde{\Sigma}$ are $K\times K$ diagonal matrices. Then we have
\begin{align}\label{R03}
U=Z\Delta^{-1}\tilde{U},\Sigma=\tilde{\Sigma},\mathrm{and~}V=\tilde{V}.
\end{align}
Recall that $U=ZX$, Equation (\ref{R03}) gives that $X=\Delta^{-1}\tilde{U}\in\mathbb{R}^{K\times K}$ and $\mathrm{rank}(X)=K$ because $\mathrm{rank}(\Delta)=K$ and $\mathrm{rank}(\tilde{U})=K$. We can easily verify that the rows of $\Delta^{-1}\tilde{U}$ are perpendicular to each other and the $k$-th row has length $\sqrt{1/N_{k}}$ for $k\in[K]$, i.e., $\sqrt{XX'}=\sqrt{\Delta^{-1}\tilde{U}\tilde{U}'\Delta^{-1}}=\sqrt{\Delta^{-2}}=\Delta^{-1}$. Thus, the fourth statement holds.
\begin{rem}
In this remark, we provide the reason why the fourth statement does not hold when $K_{0}<K$. For this case, the rank of $\Delta\Theta'$ is $K_{0}$, thus $\tilde{U}\in\mathbb{R}^{K\times K_{0}}$ and $\mathrm{rank}(\tilde{U})=K_{0}$. Then we have $X=\Delta^{-1}\tilde{U}\in\mathbb{R}^{K\times K_{0}}$ and $\mathrm{rank}(X)=K_{0}$. Thus, $\mathrm{rank}(XX')=K_{0}<K=\mathrm{rank}(\Delta^{-2})$, which implies $\sqrt{XX'}\neq \Delta^{-1}$ when $K_{0}<K$ and the fourth statement does not hold.
\end{rem}
\end{proof}
\subsection{Proof of Theorem \ref{mainWLCM}}
First, the following two lemmas are provided for our further proof.
\begin{lem}\label{boundUVWLCM}
Under $WLCM(Z,\Theta,\mathcal{F})$, we have
\begin{align*}	\mathrm{max}(\|\hat{U}\hat{O}-U\|_{F},\|\hat{V}\hat{O}-V\|_{F})\leq\frac{2\sqrt{2K}\|R-R_{0}\|}{\rho\sigma_{K}(B)\sqrt{N_{\mathrm{min}}}},
\end{align*}
where $\hat{O}$ is a $K$-by-$K$ orthogonal matrix.
\end{lem}
\begin{proof}
According to the proof of Lemma 3 in \cite{zhou2019analysis}, there is a $K\times K$ orthogonal matrix $\hat{O}$ such that
\begin{align*}	\mathrm{max}(\|\hat{U}\hat{O}-U\|_{F},\|\hat{V}\hat{O}-V\|_{F})\leq\frac{\sqrt{2K}\|\hat{R}-R_{0}\|}{\sqrt{\lambda_{K}(R_{0}R'_{0})}}.
\end{align*}
Because $\hat{R}$ is the top $K$ SVD of $R$ and $\mathrm{rank}(R_{0})=K$, we have $\|R-\hat{R}\|\leq\|R-R_{0}\|$. Then we have $\|\hat{R}-R_{0}\|=\|\hat{R}-R+R-R_{0}\|\leq 2\|R-R_{0}\|$, which gives
\begin{align}\label{BoungZhouZhixin} \mathrm{max}(\|\hat{U}\hat{O}-U\|_{F},\|\hat{V}\hat{O}-V\|_{F})\leq\frac{2\sqrt{2K}\|R-R_{0}\|}{\sqrt{\lambda_{K}(R_{0}R'_{0})}}.
\end{align}
For $\lambda_{K}(R_{0}R'_{0})$, because $R_{0}=Z\Theta'=\rho ZB'$ and $\lambda_{K}(Z'Z)=N_{\mathrm{min}}$, we have
\begin{align*}	
\lambda_{K}(R_{0}R'_{0})&=\lambda_{K}(Z\Theta'\Theta Z')=\lambda_{K}(\rho^{2} ZB'BZ')=\rho^{2}\lambda_{K}(B'BZ'Z)\\ &\geq\rho^{2}\lambda_{K}(Z'Z)\lambda_{K}(B'B)=\rho^{2}N_{\mathrm{min}}\lambda_{K}(B'B).
\end{align*}
Combining Equation (\ref{BoungZhouZhixin}) with $\lambda_{K}(R_{0}R'_{0})\geq\rho^{2}N_{\mathrm{min}}\lambda_{K}(B'B)$ gives
\begin{align*} \mathrm{max}(\|\hat{U}\hat{O}-U\|_{F},\|\hat{V}\hat{O}-V\|_{F})\leq\frac{2\sqrt{2K}\|R-R_{0}\|}{\rho\sigma_{K}(B)\sqrt{N_{\mathrm{min}}}}.
\end{align*}
\end{proof}
\begin{lem}\label{boundRR0}
	Under $WLCM(Z,\Theta,\mathcal{F})$, if Assumption \ref{asump} is satisfied, then with probability at least $1-o((N+J)^{-3})$,
	\begin{align*}
	\|R-R_{0}\|\leq C\sqrt{\gamma\mathrm{~max}(N,J)\mathrm{log}(N+J)},
	\end{align*}
	where $C$ is a positive constant.
\end{lem}
\begin{proof}
This lemma holds by setting $\alpha$ in Lemma 2 \cite{qing2023community} as 3, where Lemma 2 of \cite{qing2023community} is obtained from the rectangular version of Bernstein inequality in \cite{tropp2012user}.
\end{proof}
\begin{proof}
Now, we prove the first statement of Theorem \ref{mainWLCM}. Set $\varsigma>0$ as a small value, by Lemma 2 of \cite{joseph2016impact} and the fourth statement of Lemma \ref{UVWLCM}, if
\begin{align}\label{holdWLCM}	\frac{\sqrt{K}}{\varsigma}\|U-\hat{U}\hat{O}\|_{F}(\frac{1}{\sqrt{N_{k}}}+\frac{1}{\sqrt{N_{l}}})\leq \sqrt{\frac{1}{N_{k}}+\frac{1}{N_{l}}}, \mathrm{~for~each~}1\leq k\neq l\leq K,
\end{align}
then the Clustering error $\hat{f}=O(\varsigma^{2})$ using the K-means algorithm. By setting $\varsigma=\sqrt{\frac{2K N_{\mathrm{max}}}{N_{\mathrm{min}}}}\|U-\hat{U}\hat{O}\|_{F}$, we see that Equation (\ref{holdWLCM}) always holds for all $1\leq k\neq l\leq K$. So we get $\hat{f}=O(\varsigma^{2})=O(\frac{K N_{\mathrm{max}}\|U-\hat{U}\hat{O}\|^{2}_{F}}{N_{\mathrm{min}}})$. According to Lemma \ref{boundUVWLCM}, we have
\begin{align*}
\hat{f}=O(\frac{K^{2}N_{\mathrm{max}}\|R-R_{0}\|^{2}}{\rho^{2}\sigma^{2}_{K}(B)N^{2}_{\mathrm{min}}}).
\end{align*}
By Lemma \ref{boundRR0}, we have
\begin{align*}
\hat{f}=O(\frac{\gamma K^{2}N_{\mathrm{max}}\mathrm{max}(N,J)\mathrm{log}(N+J)}{\rho^{2}\sigma^{2}_{K}(B)N^{2}_{\mathrm{min}}}).
\end{align*}
Next, we prove the second statement of Theorem \ref{mainWLCM}. Since $U=ZX$ by Equation (\ref{UX}) in Lemma \ref{UVWLCM} and $U'U=I_{K\times K}$, we have $X'Z'ZX=I_{K\times K}$ which gives that $(Z'Z)^{-1}=XX'$ and $\lambda_{1}(XX')=\sigma^{2}_{1}(XX')=\frac{1}{\lambda_{K}(Z'Z)}=\frac{1}{N_{\mathrm{min}}}$. We also have $Z(Z'Z)^{-1}=ZXX'=UX'$. Similarly, we have $\hat{Z}(\hat{Z}'\hat{Z})^{-1}\approx \hat{U}\hat{X}'$, where $\hat{X}$ is the $K\times K$ centroid matrix returned by K-means method for $\hat{U}$. Recall that $\hat{R}=\hat{U}\hat{\Sigma}\hat{V}'$, combine it with Equation (\ref{ThetaZ}) and Lemma \ref{boundRR0}, we have
\begin{align*}
\|\hat{\Theta}-\Theta\mathcal{P}\|&=\|\hat{V}\hat{\Sigma}\hat{U}'\hat{Z}(\hat{Z}'\hat{Z})^{-1}-V\Sigma U'Z(Z'Z)^{-1}\mathcal{P}\|=\|\hat{R}'\hat{Z}(\hat{Z}'\hat{Z})^{-1}-R'_{0}Z(Z'Z)^{-1}\mathcal{P}\|\\
&=\|(\hat{R}'-R'_{0})\hat{Z}(\hat{Z}'\hat{Z})^{-1}+R'_{0}(\hat{Z}(\hat{Z}'\hat{Z})^{-1}-Z(Z'Z)^{-1}\mathcal{P})\|\leq\|(\hat{R}'-R'_{0})\hat{Z}(\hat{Z}'\hat{Z})^{-1}\|+\|R'_{0}(\hat{Z}(\hat{Z}'\hat{Z})^{-1}-Z(Z'Z)^{-1}\mathcal{P})\|\\
&\leq\|\hat{R}-R_{0}\|\|\hat{Z}(\hat{Z}'\hat{Z})^{-1}\|+\|R_{0}\|\hat{Z}(\hat{Z}'\hat{Z})^{-1}-Z(Z'Z)^{-1}\mathcal{P}\|\leq2\|R-R_{0}\|\|\hat{Z}(\hat{Z}'\hat{Z})^{-1}\|+\|R_{0}\|\hat{Z}(\hat{Z}'\hat{Z})^{-1}-Z(Z'Z)^{-1}\mathcal{P}\|\\
&=2\|R-R_{0}\|\|\hat{Z}(\hat{Z}'\hat{Z})^{-1}\|+\|\rho ZB'\|\|\hat{Z}(\hat{Z}'\hat{Z})^{-1}-Z(Z'Z)^{-1}\mathcal{P}\|\leq2\|R-R_{0}\|\|\hat{Z}(\hat{Z}'\hat{Z})^{-1}\|+\rho \|Z\|\|B\|\|\hat{Z}(\hat{Z}'\hat{Z})^{-1}-Z(Z'Z)^{-1}\mathcal{P}\|\\
&=2\|R-R_{0}\|\|\hat{Z}(\hat{Z}'\hat{Z})^{-1}\|+\rho\sigma_{1}(B)\sqrt{N_{\mathrm{max}}}\|\hat{Z}(\hat{Z}'\hat{Z})^{-1}-Z(Z'Z)^{-1}\mathcal{P}\|\\
&=2\|R-R_{0}\|\|\hat{Z}(\hat{Z}'\hat{Z})^{-1}-Z(Z'Z)^{-1}\mathcal{P}+Z(Z'Z)^{-1}\mathcal{P}\|+\rho\sigma_{1}(B)\sqrt{N_{\mathrm{max}}}\|\hat{Z}(\hat{Z}'\hat{Z})^{-1}-Z(Z'Z)^{-1}\mathcal{P}\|\\
&\leq2\|R-R_{0}\|(\|\hat{Z}(\hat{Z}'\hat{Z})^{-1}-Z(Z'Z)^{-1}\mathcal{P}\|+\|Z(Z'Z)^{-1}\mathcal{P}\|)+\rho\sigma_{1}(B)\sqrt{N_{\mathrm{max}}}\|\hat{Z}(\hat{Z}'\hat{Z})^{-1}-Z(Z'Z)^{-1}\mathcal{P}\|\\
&\leq2\|R-R_{0}\|(\|\hat{Z}(\hat{Z}'\hat{Z})^{-1}-Z(Z'Z)^{-1}\mathcal{P}\|+\|Z(Z'Z)^{-1}\|\|\mathcal{P}\|)+\rho\sigma_{1}(B)\sqrt{N_{\mathrm{max}}}\|\hat{Z}(\hat{Z}'\hat{Z})^{-1}-Z(Z'Z)^{-1}\mathcal{P}\|\\
&=2\|R-R_{0}\|(\|\hat{Z}(\hat{Z}'\hat{Z})^{-1}-Z(Z'Z)^{-1}\mathcal{P}\|+\frac{1}{\sqrt{N_{\mathrm{min}}}})+\rho\sigma_{1}(B)\sqrt{N_{\mathrm{max}}}\|\hat{Z}(\hat{Z}'\hat{Z})^{-1}-Z(Z'Z)^{-1}\mathcal{P}\|\\
&=O(\|R-R_{0}\|(\|\hat{U}\hat{X}'-UX'\mathcal{P}\|+\frac{1}{\sqrt{N_{\mathrm{min}}}}))\leq O(\|R-R_{0}\|(\|\hat{U}\hat{X}'\|+\|UX'\mathcal{P}\|+\frac{1}{\sqrt{N_{\mathrm{min}}}}))\\
&\leq O(\|R-R_{0}\|(\|\hat{U}\|\|\hat{X}\|+\|U\|\|X\|\|\mathcal{P}\|+\frac{1}{\sqrt{N_{\mathrm{min}}}}))=O(\|R-R_{0}\|(\|\hat{X}\|+\|X\|+\frac{1}{\sqrt{N_{\mathrm{min}}}}))\\
&=O(\frac{\|R-R_{0}\|}{\sqrt{N_{\mathrm{min}}}})=O(\sqrt{\frac{\gamma\mathrm{~max}(N,J)\mathrm{log}(N+J)}{N_{\mathrm{min}}}})\\
\end{align*}
Since $(\hat{\Theta}-\Theta\mathcal{P})$ is a $J\times K$ matrix and $K\ll J$ in this paper, we have $\mathrm{rank}(\hat{\Theta}-\Theta\mathcal{P})=K$. Since $\|M\|_{F}\leq\sqrt{\mathrm{rank}(M)}\|M\|$ holds for any matrix $M$, we have $\|\hat{\Theta}-\Theta\mathcal{P}\|_{F}\leq\sqrt{K}\|\hat{\Theta}-\Theta\mathcal{P}\|$. Thus, we have
\begin{align}\label{l1}
\|\hat{\Theta}-\Theta\mathcal{P}\|_{F}=O(\sqrt{\frac{\gamma K\mathrm{~max}(N,J)\mathrm{log}(N+J)}{N_{\mathrm{min}}}}).
\end{align}
Combing Equation (\ref{l1}) with the fact that $\|\Theta\|_{F}\geq\|\Theta\|=\|\rho B\|=\rho\|B\|=\rho\sigma_{1}(B)\geq\rho\sigma_{K}(B)$ gives
\begin{align*}
\frac{\|\hat{\Theta}-\Theta\mathcal{P}\|_{F}}{\|\Theta\|_{F}}\leq\frac{\|\hat{\Theta}-\Theta\mathcal{P}\|_{F}}{\rho\sigma_{K}(B)}=O(\frac{\sqrt{\gamma K\mathrm{max}(N,J)\mathrm{log}(N+J)}}{\rho\sigma_{K}(B)\sqrt{N_{\mathrm{min}}}}).
\end{align*}
Recall that the $J$-by-$K$ matrix $B$ satisfies $\mathrm{max}_{j\in[J],k\in[K]}|B(j,k)|=1$, we have $\sigma_{K}(B)$ is at least of the order $\sqrt{J}-\sqrt{K-1}$ with high probability by applying the lower bound of the smallest singular value of a random rectangular matrix in \cite{rudelson2009smallest}. Since $K\ll J$ in this paper, we have
\begin{align*}
&\hat{f}=O(\frac{\gamma K^{2}N_{\mathrm{max}}\mathrm{max}(N,J)\mathrm{log}(N+J)}{\rho^{2}N^{2}_{\mathrm{min}}J}) \mathrm{~and~}\frac{\|\hat{\Theta}-\Theta\mathcal{P}\|_{F}}{\|\Theta\|_{F}}=O(\frac{\sqrt{\gamma K\mathrm{max}(N,J)\mathrm{log}(N+J)}}{\rho\sqrt{N_{\mathrm{min}}J}}).
\end{align*}
\end{proof}
\bibliographystyle{elsarticle-num}
\bibliography{refWLCM}
\end{document}